\documentclass{article}
\usepackage{amsmath,amssymb,amsfonts}

\def\op#1{\mathop{{\it\fam0} #1}\limits}
\newcommand{\id}{{\rm Id\,}}

\newcommand{\pr}{{\rm pr}}

\newcommand{\di}{{\rm dim\,}}

\newcommand{\Ker}{{\rm Ker\,}}
\newcommand{\im}{{\rm Im\, }}

\newcommand{\hm}{{\rm Hom\,}}
\newcommand{\dif}{{\rm Diff\,}}

\newcommand{\nm}[1]{|{#1}|}

\newcommand{\lng}{\langle}
\newcommand{\rng}{\rangle}
\newcommand{\bll}{\bullet}

\newcommand{\bite}{\begin{itemize}}
\newcommand{\eite}{\end{itemize}}
\newcommand{\benu}{\begin{enumerate}}
\newcommand{\eenu}{\end{enumerate}}
\newcommand{\bde}{\begin{description}}
\newcommand{\ede}{\end{description}}
\newcommand{\bquo}{\begin{quote}}
\newcommand{\equo}{\end{quote}}
\newcommand{\bquot}{\begin{quotation}}
\newcommand{\equot}{\end{quotation}}

\newcommand{\beq}{\begin{equation}}
\newcommand{\eeq}{\end{equation}}
\newcommand{\ben}{\begin{eqnarray}}
\newcommand{\een}{\end{eqnarray}}
\newcommand{\be}{\begin{eqnarray*}}
\newcommand{\ee}{\end{eqnarray*}}
\newcommand{\bea}{\begin{eqalph}}
\newcommand{\eea}{\end{eqalph}}

\newcommand{\gU}{{\mathfrak U}}

\newcommand{\cG}{{\mathfrak g}}
\newcommand{\gd}{{\mathfrak d}}

\newcommand{\gE}{{\mathfrak E}}

\newcommand{\gF}{{\mathfrak F}}

\newcommand{\cA}{{\cal A}}
\newcommand{\cO}{{\cal O}}
\newcommand{\cT}{{\cal T}}

\newcommand{\cL}{{\cal L}}
\newcommand{\cV}{{\cal V}}

\newcommand{\cE}{{\cal E}}
\newcommand{\cH}{{\cal H}}
\newcommand{\cF}{{\cal F}}
\newcommand{\cC}{{\cal C}}

\newcommand{\cK}{{\cal K}}
\newcommand{\cM}{{\cal M}}

\newcommand{\cS}{{\cal S}}
\newcommand{\cD}{{\cal D}}

\newcommand{\bL}{{\bf L}}

\newcommand{\bb}{{\bf 1}}

\newcommand{\bH}{{\bf H}}

\newcommand{\rA}{{\rm Ann\,}}
\newcommand{\rO}{{\rm Orth}}

\newcommand{\al}{\alpha}
\newcommand{\bt}{\beta}
\newcommand{\dl}{\delta}
\newcommand{\la}{\lambda}
\newcommand{\La}{\Lambda}
\newcommand{\f}{\phi}
\newcommand{\vf}{\varphi}
\newcommand{\vs}{\varsigma}
\newcommand{\F}{\Phi}
\newcommand{\p}{\pi}
\newcommand{\s}{\psi}

\newcommand{\Om}{\Omega}
\newcommand{\m}{\mu}
\newcommand{\n}{\nu}
\newcommand{\g}{\gamma}
\newcommand{\G}{\Gamma}
\newcommand{\e}{\epsilon}
\newcommand{\ve}{\varepsilon}
\newcommand{\thh}{\theta}

\newcommand{\vr}{\varrho}
\newcommand{\up}{\upsilon}
\newcommand{\vt}{\vartheta}

\newcommand{\si}{\sigma}
\newcommand{\Si}{\Sigma}

\newcommand{\fl}{\flat}
\newcommand{\sh}{\sharp}

\newcommand{\bom}{{\bf\Omega}}
\newcommand{\bth}{{\bf\Theta}}
\newcommand{\bT}{{\bf T}}
\newcommand{\bs}{{\bf s}}
\newcommand{\Y}{Y\to X}
\newcommand{\w}{\wedge}

\newcommand{\wt}{\widetilde}
\newcommand{\wh}{\widehat}
\newcommand{\ol}{\overline}

\newcommand{\dr}{\partial}

\newcommand{\ar}{\op\longrightarrow}

\newcommand{\xx}{\times}

\newcommand{\ot}{\otimes}
\newcommand{\ap}{\approx}

\let\ssection=\section
\renewcommand{\section}{\setcounter{equation}{0}\ssection}

\newcounter{eqalph}[section]
\newcounter{equationa}[section]
\newcounter{example}[section]
\newcounter{remark}[section]
\newcounter{theorem}[section]
\newcounter{proposition}[section]
\newcounter{lemma}[section]
\newcounter{corollary}[section]
\newcounter{definition}[section]

\def\theremark{\arabic{section}.\arabic{remark}}

\def\thedefinition{\arabic{section}.\arabic{definition}}

\newenvironment{proof}{\noindent {\it Proof:}\small
}{$\Box$\medskip}
\newenvironment{example}{\refstepcounter{remark}\medskip\noindent{\bf
Example \theremark:}\ }{$\Box$ \medskip}
\newenvironment{remark}{\refstepcounter{remark}\medskip\noindent{\bf
Remark \theremark:} }{$\Box$\medskip}
\newenvironment{theorem}{\refstepcounter{definition}
\medskip\noindent{\sc Theorem \thedefinition}:}{$\Box$\medskip}
\newenvironment{proposition}{\refstepcounter{definition}\medskip\noindent{\sc
Proposition \thedefinition}:}{$\Box$\medskip}
\newenvironment{lemma}{\refstepcounter{definition}\medskip\noindent{\sc
Lemma \thedefinition}:}{ $\Box$\medskip }

\newenvironment{definition}{\refstepcounter{definition}\medskip\noindent{\sc
Definition \thedefinition}:}{$\Box$\medskip}

\newenvironment{eqalph}{\stepcounter{equation}
\setcounter{equationa}{\value{equation}} \setcounter{equation}{0}

\begin{eqnarray}}{\end{eqnarray}
\setcounter{equation}{\value{equationa}}}

\newcommand{\mar}[1]{}

\begin{document}

\hbox{}

\begin{center}

{\Large\bf

Lectures on integrable Hamiltonian systems}

\bigskip
\bigskip

G. SARDANASHVILY

\bigskip

Department of Theoretical Physics, Moscow State University, 117234
Moscow, Russia

\bigskip

\end{center}

\begin{abstract}

\noindent We consider integrable Hamiltonian systems in a general
setting of invariant submanifolds which need not be compact. For
instance, this is the case a global Kepler system, non-autonomous
integrable Hamiltonian systems and integrable systems with
time-dependent parameters.
\end{abstract}

\bigskip

\tableofcontents



\addcontentsline{toc}{section}{Introduction}

\section*{Introduction}

The Liouville -- Arnold theorem for completely integrable systems
\cite{arn,laz}, the Poincar\'e -- Lyapounov -- Nekhoroshev theorem
for partially integrable systems \cite{gaeta,nekh94} and the
Mishchenko -- Fomenko theorem for the superintegrable ones
\cite{bols03,fasso05,mishc} state the existence of action-angle
coordinates around a compact invariant submanifold of a
Hamiltonian integrable system which is a torus $T^m$. However, it
is well known that global extension of these action-angle
coordinates meets a certain topological obstruction
\cite{bate,daz,dust}.

Note that superintegrable systems sometimes are called
non-commutative (or non-Abelian) completely integrable systems.

In these Lectures, we consider integrable Hamiltonian systems in a
general setting of invariant submanifolds which need not be
compact \cite{fior,fior2,jmp07a,f2,jmp03,book10,ijgmmp09a,vin}.
These invariant submanifolds are proved to be diffeomorphic to
toroidal cylinders $\mathbb R^{m-r}\times T^r$ (Theorem
\ref{bi0}). A key point is that, in accordance with Theorem
\ref{11t2}, a fibred manifold whose fibres are diffeomorphic
either to a compact manifold or $\mathbb R^r$ is a fibre bundle,
but this is not the case of toroidal cylinders.

In particular, this is the case of non-autonomous integrable
Hamiltonian systems (Section 4.3) and Hamiltonian mechanics with
time-dependent parameters (Section 6).

It may happen that a Hamiltonian system on a phase space $Z$ falls
into different integrable Hamiltonian systems on different open
subsets of $Z$. For instance, this the case of the Kepler system
(Section 3). It contains two different globally superintegrable
systems on different open subsets of a phase space $Z=\mathbb
R^4$. Their integrals of motion form the Lie algebras $so(3)$ and
$so(2,1)$ with compact and non-compact invariant submanifolds,
respectively \cite{book10,ijgmmp09a}.

Geometric quantization of completely integrable and
superintegrable Hamiltonian systems with respect to action-angle
variables is considered (Section 5.3). The reason is that, since a
Hamiltonian of an integrable system depends only on action
variables, it seems natural to provide the Schr\"odinger
representation of action variables by first order differential
operators on functions of angle coordinates.

Throughout the Lectures, all functions and maps are smooth, and
manifolds are real smooth and paracompact. We are not concerned
with the real-analytic case because a paracompact real-analytic
manifold admits the partition of unity by smooth functions. As a
consequence, sheaves of modules over real-analytic functions need
not be acyclic that is essential for our consideration.

\section{Partially integrable systems}

This Section addresses partially integrable systems on Poisson and
symplectic manifolds.  Completely integrable systems can be
regarded as the particular partially integrable ones (Remark
\ref{000}). A key point is that a partially integrable system
admits different compatible Poisson structures (Theorem
\ref{bi92}).

Our goal are Theorem \ref{bi100} on partial integrable systems on
a Poisson manifold, Theorem \ref{nc6} on partial integrable system
on a symplectic manifold, and Theorem \ref{cmp35} as the global
generalization of Theorem \ref{nc6}.

\subsection{Geometry of symplectic and Poisson manifolds}

This Section summarize some relevant material on symplectic
manifolds, Poisson manifolds and symplectic foliations
\cite{abr,book05,book10,libe,vais}.

 Let $Z$ be a smooth
manifold. Any exterior two-form $\Om$ on $Z$ yields a linear
bundle morphism
\mar{m52}\beq
\Om^\fl: TZ\op\to_Z T^*Z, \qquad \Om^\fl:v \to - v\rfloor\Om(z),
\quad v\in T_zZ, \quad z\in Z. \label{m52}
\eeq
One says that a two-form $\Om$ is of rank  $r$ if the morphism
(\ref{m52}) has a rank $r$. A kernel
 $\Ker \Om$ of $\Om$
is defined as the kernel of the morphism (\ref{m52}). In
particular, $\Ker \Om$ contains the canonical zero section $\wh 0$
of $TZ\to Z$. If $\Ker \Om=\wh 0$, a two-form $\Om$ is said to be
non-degenerate. A closed non-degenerate two-form $\Om$ is called
symplectic. Accordingly, a manifold equipped with a symplectic
form is a symplectic manifold. A symplectic manifold $(Z,\Om)$
always is even dimensional and orientable.

A manifold morphism $\zeta$ of a symplectic manifold $(Z,\Om)$ to
a symplectic manifold $(Z',\Om')$ is called symplectic if
$\Om=\zeta^*\Om'$. Any symplectic morphism is an immersion. A
symplectic isomorphism is called the symplectomorphism.

A vector field $u$ on a symplectic manifold $(Z,\Om)$ is an
infinitesimal generator of a local one-parameter group of local
symplectomorphism iff the Lie derivative $\bL_u\Om$ vanishes. It
is called  the canonical vector field. A canonical vector field
$u$ on a symplectic manifold $(Z,\Om)$ is said to be Hamiltonian
if a closed one-form $u\rfloor\Om$ is exact. Any smooth function
$f\in C^\infty(Z)$ on $Z$ defines a unique Hamiltonian vector
field $\vt_f$ such that
\mar{z100}\beq
\vt_f\rfloor\Om=-df, \qquad \vt_f=\Om^\sh(df), \label{z100}
\eeq
where $\Om^\sh$ is the inverse isomorphism to $\Om^\fl$
(\ref{m52}).

\begin{example}\label{e11.0} \mar{e11.0} Given an $m$-dimensional manifold $M$ coordinated by
$(q^i)$, let
\be
\pi_{*M}:T^*M\to M
\ee
be its cotangent bundle equipped with the holonomic coordinates
$(q^i, p_i=\dot q_i)$. It is endowed with the canonical Liouville
form
\be
\Xi=p_idq^i
\ee
and the canonical symplectic form
\mar{m83}\beq
\Om_T= d\Xi=dp_i\w dq^i. \label{m83}
\eeq
Their coordinate expressions are maintained under holonomic
coordinate transformations. The Hamiltonian vector field $\vt_f$
(\ref{z100}) with respect to the canonical symplectic form
(\ref{m83}) reads
\be
\vt_f=\dr^if\dr_i -\dr_i f\dr^i.
\ee
\end{example}

The canonical symplectic form (\ref{m83}) plays a prominent role
in symplectic geometry in view of the classical Darboux theorem.

\begin{theorem} \label{spr825} \mar{spr825}
Each point of a symplectic manifold $(Z,\Om)$ has an open
neighborhood equipped with coordinates $(q^i,p_i)$, called
canonical or Darboux coordinates, such that $\Om$ takes the
coordinate form (\ref{m83}).
\end{theorem}

Let $i_N:N\to Z$ be a submanifold of a $2m$-dimensional symplectic
manifold $(Z,\Om)$. A subset
\be
\rO_\Om TN=\op\bigcup_{z\in N} \{v\in T_zZ: \, v\rfloor u\rfloor
\Om=0, \,\,  u\in T_zN\}
\ee
of $TZ|_N$ is called orthogonal to $TN$ relative to a symplectic
form $\Om$. One considers the following special types of
submanifolds of a symplectic manifold such that the pull-back
$\Om_N=i_N^*\Om$ of a symplectic form $\Om$ onto a submanifold $N$
is of constant rank. A submanifold $N$ of $Z$ is said to be:

$\bullet$  coisotropic if $\rO_\Om TN\subseteq TN$, $\di N\geq m$;

$\bullet$ symplectic if $\Om_N$ is a symplectic form on $N$;

$\bullet$ isotropic  if $TN\subseteq \rO_\Om TN$, $\di N\leq m$.

A Poisson bracket on a ring $C^\infty(Z)$ of smooth real functions
on a manifold $Z$ (or a Poisson structure on $Z$) is defined as an
$\mathbb R$-bilinear map
\be
C^\infty(Z)\times C^\infty(Z)\ni (f,g)\to \{f,g\}\in C^\infty(Z)
\ee
which satisfies the following conditions:

$\bullet$ $\{g,f\}=-\{f,g\}$;

$\bullet$ $\{f,\{g,h\}\} + \{g,\{h,f\}\} +\{h,\{f,g\}\}=0$;

$\bullet$ $\{h,fg\}=\{h,f\}g +f\{h,g\}$.

A  Poisson bracket makes $C^\infty(Z)$ into a real Lie algebra,
called the Poisson algebra.  A Poisson structure is characterized
by a particular bivector field as follows.

\begin{theorem}\label{p11.1} \mar{p11.1} Every Poisson bracket on a
manifold $Z$ is uniquely defined as
\mar{b450}\beq
\{f,f'\}=w(df,df')=w^{\m\nu}\dr_\m f\dr_\nu f' \label{b450}
\eeq
by a bivector field $w$ whose Schouten -- Nijenhuis bracket
$[w,w]_{\rm SN}$ vanishes. It is called a Poisson bivector field.
\end{theorem}

A manifold $Z$ endowed with a Poisson structure is called a
Poisson manifold.

\begin{example} \label{w206} \mar{w206}
Any manifold admits a zero Poisson structure characterized by a
zero Poisson bivector field $w=0$.
\end{example}

A function $f\in C^\infty(Z)$ is called the Casimir function of a
Poisson structure on $Z$ if its Poisson bracket with any function
on $Z$ vanishes. Casimir functions form a real ring $\cC(Z)$.

Any bivector field $w$ on a manifold $Z$ yields a linear bundle
morphism
\mar{m51}\beq
w^\sh: T^*Z\op\to_Z TZ, \qquad w^\sh:\al\to -w(z)\lfloor \al,
\qquad \al\in T^*_zZ. \label{m51}
\eeq
One says that $w$ is of rank  $r$ if the morphism (\ref{m51}) is
of this rank. If a Poisson bivector field is of constant rank, the
Poisson structure is called regular. Throughout the Lectures, only
regular Poisson structures are considered. A Poisson structure
determined by a Poisson bivector field $w$ is said to be
non-degenerate if $w$ is of maximal rank.

There is one-to-one correspondence $\Om_w \leftrightarrow w_\Om$
between the symplectic forms and the non-degenerate Poisson
bivector fields  which is given by the equalities
\be
&&w_\Om(\f,\si)=\Om_w(w_\Om^\sh(\f),w_\Om^\sh(\si)), \qquad
\f,\si\in\cO^1(Z),\\
&&\Om_w(\vt,\nu)=w_\Om(\Om_w^\fl(\vt),\Om_w^\fl(\nu)), \qquad
\vt,\nu\in \cT(Z),
\ee
where the morphisms $w_\Om^\sh$ (\ref{m51}) and $\Om_w^\fl$
(\ref{m52}) are mutually inverse, i.e.,
\be
 w_\Om^\sh=\Om^\sh_w, \qquad w^{\al\nu}_\Om\Om_{w\al\bt}= \dl_\bt^\nu.
\ee

However, this correspondence is not preserved under manifold
morphisms in general. Namely, let $(Z_1,w_1)$ and $(Z_2,w_2)$ be
Poisson manifolds. A manifold  morphism $\varrho:Z_1\to Z_2$ is
said to be a Poisson morphism if
\be
\{f\circ \varrho,f'\circ\varrho\}_1=\{f,f'\}_2\circ\varrho, \qquad
f,f'\in C^\infty(Z_2),
\ee
or, equivalently, if $w_2=T\varrho \circ w_1$, where $T\varrho$ is
the tangent map to $\varrho$. Herewith, the rank of $w_1$ is
superior or equal to that of $w_2$. Therefore, there are no
pull-back and push-forward operations of Poisson structures in
general. Nevertheless, let us mention the following construction.

\begin{theorem}\label{p11.3} \mar{p11.3} Let $(Z,w)$ be a Poisson manifold and
$\pi:Z\to Y$ a fibration such that, for every pair of functions
$(f,g)$ on $Y$ and for each point $y\in Y$, the restriction of a
function $\{\pi^*f,\pi^*g\}$ to a fibre $\pi^{-1}(y)$ is constant,
i.e., $\{\pi^*f,\pi^*g\}$ is the pull-back onto $Z$ of some
function on $Y$. Then there exists a coinduced Poisson structure
$w'$ on $Y$ for which $\pi$ is a Poisson morphism.
\end{theorem}

\begin{example} \label{spr862} \mar{spr862}
The direct product $Z\times Z'$ of Poisson manifolds $(Z,w)$ and
$(Z',w')$ can be endowed with the product of Poisson structures,
given by a bivector field $w + w'$ such that the surjections
$\pr_1$ and $\pr_2$ are Poisson morphisms.
\end{example}

A vector field $u$ on a Poisson manifold $(Z,w)$ is an
infinitesimal generator of a local one-parameter group of Poisson
automorphisms iff the Lie derivative
\mar{gm508}\beq
\bL_uw=[u,w]_{\rm SN} \label{gm508}
\eeq
vanishes. It is called the canonical vector field for a Poisson
structure $w$. In particular, for any real smooth function $f$ on
a Poisson manifold $(Z,w)$, let us put
\mar{gm509}\beq
\vt_f=w^\sh (df)= -w\lfloor df=w^{\m\nu}\dr_\m f\dr_\nu.
\label{gm509}
\eeq
It is a canonical vector field, called the Hamiltonian vector
field of a function $f$ with respect to a Poisson structure $w$.
Hamiltonian vector fields fulfil the relations
\mar{m81',1}\ben
&& \{f,g\}=\vt_f\rfloor dg, \label{m81'}\\
&& [\vt_f,\vt_g]=\vt_{\{f,g\}}, \qquad f,g\in C^\infty(Z). \label{m81}
\een

For instance, the Hamiltonian vector field $\vt_f$ (\ref{z100}) of
a function $f$ on a symplectic manifold $(Z,\Om)$ coincides with
that (\ref{gm509}) with respect to the corresponding Poisson
structure $w_\Om$. The Poisson bracket defined by a symplectic
form $\Om$ reads
\be
\{f,g\}=\vt_g\rfloor\vt_f\rfloor\Om.
\ee

Since a Poisson manifold $(Z,w)$ is assumed to be regular, the
range $\bT=w^\sh(T^*Z)$ of the morphism (\ref{m51}) is a subbundle
of $TZ$ called the characteristic distribution on $(Z,w)$. It is
spanned by Hamiltonian vector fields, and it is involutive by
virtue of the relation (\ref{m81}). It follows that a Poisson
manifold $Z$ admits local adapted coordinates in Theorem
\ref{c11.0}. Moreover, one can choose particular adapted
coordinates which bring a Poisson structure into the following
canonical form.

\begin{theorem}\label{canpo} \mar{canpo} For any point $z$ of a $k$-dimensional Poisson
manifold $(Z,w)$, there exist coordinates
\mar{m111}\beq
(z^1,\dots,z^{k-2m},q^1,\dots,q^m,p_1,\dots,p_m) \label{m111}
\eeq
on a neighborhood of $z$ such that
\be
w=\frac{\dr}{\dr p_i}\w\frac{\dr}{\dr q^i},\qquad
\{f,g\}=\frac{\dr f}{\dr p_i} \frac{\dr g}{\dr q^i}
   - \frac{\dr f}{\dr q^i} \frac{\dr g}{\dr p_i}.
\ee
\end{theorem}

The coordinates (\ref{m111}) are called the canonical or Darboux
coordinates for the Poisson structure $w$. The Hamiltonian vector
field of a function $f$ written in this coordinates is
\be
\vt_f= \dr^if\dr_i - \dr_if\dr^i.
\ee
Of course, the canonical coordinates for a symplectic form $\Om$
in Theorem \ref{spr825} also are canonical coordinates in Theorem
\ref{canpo} for the corresponding non-degenerate Poisson bivector
field $w$, i.e.,
\be
\Om=dp_i\w dq^i, \qquad w=\dr^i\w \dr_i.
\ee
With respect to these coordinates, the mutually inverse bundle
isomorphisms $\Om^\fl$ (\ref{m52}) and $w^\sh$ (\ref{m51}) read
\be
&& \Om^\fl: v^i\dr_i + v_i\dr^i\to -v_idq^i+ v^idp_i, \\
&& w^\sh: v_idq^i+ v^idp_i \to v^i\dr_i- v_i\dr^i.
\ee

Integral manifolds of the characteristic distribution $\bT$ of a
$k$-dimensional Poisson manifold $(Z,w)$ constitute a (regular)
foliation $\cF$ of $Z$ whose tangent bundle $T\cF$ is $\bT$. It is
called the characteristic foliation of a Poisson manifold. By the
very definition of the characteristic distribution $\bT=T\cF$, a
Poisson bivector field $w$ is subordinate to $\op\w^2 T\cF$.
Therefore, its restriction $w|_F$ to any leaf $F$ of $\cF$ is a
non-degenerate Poisson bivector field on $F$. It provides $F$ with
a non-degenerate Poisson structure $\{,\}_F$ and, consequently, a
symplectic structure. Clearly, the local Darboux coordinates for
the Poisson structure $w$ in Theorem \ref{canpo} also are the
local adapted coordinates
\be
(z^1,\ldots,z^{k-2m}, z^i=q^i,z^{m+i}=p_i), \qquad i=1,\ldots,m,
\ee
(\ref{spr850}) for the characteristic foliation $\cF$, and the
symplectic structures along its leaves read
\be
\Om_F=dp_i\w dq^i.
\ee

Since any foliation is locally simple, a local structure of an
arbitrary Poisson manifold reduces to the following
\cite{vais,wein}.

\begin{theorem}\label{wenst1} \mar{wenst1} Each point of a Poisson manifold has an
open neighborhood which is Poisson equivalent to the product of a
manifold with the zero Poisson structure and a symplectic
manifold.
\end{theorem}

Provided with this symplectic structure, the leaves of the
characteristic foliation of a Poisson manifold $Z$ are assembled
into a symplectic foliation of $Z$ as follows (see Section 7.4).

Let $\cF$ be an even dimensional foliation of a manifold $Z$. A
$\wt d$-closed non-degenerate leafwise two-form $\Om_\cF$ on a
foliated manifold $(Z,\cF)$ is called symplectic. Its pull-back
$i_F^*\Om_\cF$ onto each leaf $F$ of $\cF$ is a symplectic form on
$F$. A foliation $\cF$ provided with a symplectic leafwise form
$\Om_\cF$ is called the symplectic foliation.

If a symplectic leafwise form $\Om_\cF$ exists, it yields a bundle
isomorphism
\be
\Om_\cF^\fl: T\cF\op\to_Z T\cF^*, \qquad \Om_\cF^\fl:v\to -
v\rfloor\Om_\cF(z), \qquad v\in T_z\cF.
\ee
The inverse isomorphism $\Om_\cF^\sh$ determines a bivector field
\mar{spr904}\beq
w_\Om(\al,\bt)=\Om_\cF(\Om_\cF^\sh(i^*_\cF\al),\Om_\cF^\sh(i^*_\cF\bt)),
\qquad  \al,\bt\in T_z^*Z, \quad z\in Z, \label{spr904}
\eeq
on $Z$ subordinate to $\op\w^2T\cF$. It is a Poisson bivector
field. The corresponding Poisson bracket reads
\mar{spr902}\beq
\{f,f'\}_\cF=\vt_f\rfloor \wt df', \qquad \vt_f\rfloor\Om_\cF=-\wt
df, \qquad \vt_f=\Om_\cF^\sh(\wt df).\label{spr902}
\eeq
Its kernel is $S_\cF(Z)$.

Conversely, let $(Z,w)$ be a Poisson manifold and $\cF$ its
characteristic foliation. Since Ann$\,T\cF\subset T^*Z$ is
precisely the kernel of a Poisson bivector field $w$, a bundle
homomorphism
\be
w^\sh: T^*Z\op\to_Z TZ
\ee
factorizes in a unique fashion
\mar{lmp03}\beq
w^\sh: T^*Z\ar_Z^{i^*_\cF} T\cF^*\ar_Z^{w^\sh_\cF}
T\cF\ar_Z^{i_\cF} TZ \label{lmp03}
\eeq
through a bundle isomorphism
\mar{lmp02}\beq
w_\cF^\sh: T\cF^*\op\to_Z T\cF,  \qquad w^\sh_\cF:\al\to
-w(z)\lfloor \al, \qquad \al\in T_z\cF^*. \label{lmp02}
\eeq
The inverse isomorphism $w_\cF^\fl$ yields a symplectic leafwise
form
\mar{spr903}\beq
\Om_\cF(v,v')=w(w_\cF^\fl(v),w_\cF^\fl(v')), \qquad  v,v'\in
T_z\cF, \qquad z\in Z. \label{spr903}
\eeq
The formulas (\ref{spr904}) and (\ref{spr903}) establish the
equivalence between the Poisson structures on a manifold $Z$ and
its symplectic foliations.

Turn now to a group action on Poisson manifolds. By $G$ throughout
is meant a real connected Lie group, $\cG$ is its right Lie
algebra, and $\cG^*$ is the Lie coalgebra (see Section 7.5).

We start with the symplectic case. Let a Lie group $G$ act on a
symplectic manifold $(Z,\Om)$ on the left by symplectomorphisms.
Such an action of $G$ is called symplectic. Since $G$ is
connected, its action on a manifold $Z$ is symplectic iff the
homomorphism $\ve\to\xi_\ve$, $\ve\in\cG$, (\ref{spr941}) of a Lie
algebra $\cG$ to a Lie algebra $\cT_1(Z)$ of vector fields on $Z$
is carried out by canonical vector fields for a symplectic form
$\Om$ on $Z$. If all these vector fields are Hamiltonian, an
action of $G$ on $Z$ is called a Hamiltonian action. One can show
that, in this case, $\xi_\ve$, $\ve\in\cG$, are Hamiltonian vector
fields of functions on $Z$ of the following particular type.

\begin{proposition}\label{gen4}  \mar{gena4} An action of a Lie group $G$ on
a symplectic manifold $Z$ is Hamiltonian iff there exists a
mapping
\mar{spr943}\beq
\wh J: Z\to \cG^*, \label{spr943}
\eeq
called the momentum mapping,  such that
\mar{z210}\beq
\xi_\ve\rfloor\Om=- dJ_\ve, \qquad J_\ve(z)=\lng\wh J(z),\ve\rng,
\qquad  \ve\in \cG. \label{z210}
\eeq
\end{proposition}

The momentum mapping (\ref{spr943}) is defined up to a constant
map. Indeed, if $\wh J$ and $\wh J'$ are different momentum
mappings for the same symplectic action of $G$ on $Z$, then
\be
d(\lng \wh J(z)-\wh J'(z),\ve\rng)=0, \qquad  \ve\in\cG.
\ee
Given $g\in G$, let us us consider the difference
\mar{z220}\beq
\si(g) =\wh J(gz)-{\rm Ad}^*g(\wh J(z)), \label{z220}
\eeq
where ${\rm Ad}^*g$ is the coadjoint representation (\ref{z211})
on $\cG^*$. One can show that the difference (\ref{z220}) is
constant on a symplectic manifold $Z$ \cite{abr}. A momentum
mapping $\wh J$ is called equivariant if $\si(g)=0$, $g\in G$.

\begin{example}\label{imp} \mar{imp}
Let a symplectic form on $Z$ be exact, i.e., $\Om=d\thh$, and let
$\thh$ be $G$-invariant, i.e.,
\be
\bL_{\xi_\ve}\thh=d(\xi_\ve\rfloor\thh) +\xi_\ve\rfloor\Om=0,
\qquad
  \ve\in\cG.
\ee
Then the momentum mapping $\wh J$ (\ref{spr943}) can be given by
the relation
\be
\lng\wh J(z),\ve\rng=(\xi_\ve\rfloor\thh)(z).
\ee
It is equivariant. In accordance with the relation (\ref{z211}),
it suffices to show that
\be
J_\ve(gz)=J_{{\rm Ad}\,g^{-1}(\ve)}(z), \qquad (\xi_\ve\rfloor
\thh)(gz)=(\xi_{{\rm Ad}\,g^{-1}(\ve)}\rfloor\thh)(z).
\ee
This holds by virtue of the relation (\ref{z212}). For instance,
let $T^*Q$ be a symplectic manifold equipped with the canonical
symplectic form $\Om_T$ (\ref{m83}). Let a left action of a Lie
group $G$ on $Q$ have the infinitesimal generators
$\tau_m=\ve^i_m(q)\dr_i$. The canonical lift of this action onto
$T^*Q$ has the  infinitesimal generators (\ref{l28}):
\mar{z216}\beq
\xi_m=\wt\tau_m=ve^i_m\dr_i -p_j\dr_i\ve_m^j\dr^i, \label{z216}
\eeq
and preserves the canonical Liouville form $\Xi$ on $T^*Q$. The
$\xi_m$ (\ref{z216}) are Hamiltonian vector fields of the
functions $J_m=\ve_m^i(q)p_i$, determined by the equivariant
momentum mapping $\wh J=\ve^i_m(q)p_i\ve^m$.
\end{example}

\begin{theorem}
A momentum mapping $\wh J$ associated to a symplectic action of a
Lie group $G$ on a symplectic manifold $Z$ obeys the relation
\mar{z223}\beq
\{J_\ve,J_{\ve'}\}= J_{[\ve,\ve']} - \lng T_e\si(\ve'),\ve\rng.
\label{z223}
\eeq
\end{theorem}

In the case of an equivariant momentum mapping, the relation
(\ref{z223}) leads to a homomorphism
\mar{z224}\beq
\{J_\ve,J_{\ve'}\}= J_{[\ve,\ve']} \label{z224}
\eeq
of a Lie algebra $\cG$ to a Poisson algebra of smooth functions on
a symplectic manifold $Z$ (cf. Proposition \ref{spr948} below).

Now let a Lie group $G$ act on a Poisson manifold $(Z,w)$ on the
left by Poisson automorphism. This is a Poisson action. Since $G$
is connected, its action on a manifold $Z$ is a Poisson action iff
the homomorphism $\ve\to\xi_\ve$, $\ve\in\cG$, (\ref{spr941}) of a
Lie algebra $\cG$ to a Lie algebra $\cT_1(Z)$ of vector fields on
$Z$ is carried out by canonical vector fields for a Poisson
bivector field $w$, i.e., the condition (\ref{gm508}) holds. The
equivalent conditions are
\be
&&\xi_\ve(\{f,g\})=\{\xi_\ve(f),g\} +\{f, \xi_\ve(g)\}, \qquad f,g\in
C^\infty(Z), \\
&&\xi_\ve(\{f,g\})=[\xi_\ve,\vt_f](g)- [\xi_\ve,\vt_g](f),\\
&& [\xi_\ve,\vt_f] =\vt_{\xi_\ve(f)},
\ee
where $\vt_f$ is the Hamiltonian vector field (\ref{gm509}) of a
function $f$.

A Hamiltonian action of $G$ on a Poisson manifold $Z$ is defined
similarly to that on a symplectic manifold. Its infinitesimal
generators are tangent to leaves of the symplectic foliation of
$Z$, and there is a Hamiltonian action of $G$ on every symplectic
leaf. Proposition \ref{gen4} together with the notions of a
momentum mapping and an equivariant momentum mapping also are
extended to a Poisson action. However, the difference $\si$
(\ref{z220}) is constant only on leaves of the symplectic
foliation of $Z$ in general. At the same time, one can say
something more on an equivariant momentum mapping (that also is
valid for a symplectic action).

\begin{proposition} \label{spr948} \mar{spr948}
An equivariant momentum mapping $\wh J$  (\ref{spr943}) is a
Poisson morphism to the Lie coalgebra $\cG^*$, provided with the
Lie -- Poisson structure (\ref{gm510}).
\end{proposition}

\subsection{Poisson and symplectic Hamiltonian systems}

Given a Poisson manifold $(Z,w)$, a Poisson Hamiltonian system
$(w,\cH)$ on $Z$ for a Hamiltonian $\cH\in C^\infty(Z)$ with
respect to a Poisson structure $w$ is defined as a set
\mar{gm151}\beq
S_\cH=\op\bigcup_{z\in Z} \{v\in T_zZ:\, v-w^\sh(d\cH)(z)=0\}.
\label{gm151}
\eeq
By a solution of this Hamiltonian system is meant a vector field
$\vt$ on $Z$ which takes its values into $TN\cap S_\cH$. Clearly,
the Poisson Hamiltonian system (\ref{gm151}) has a unique solution
which is the Hamiltonian vector field
\mar{gm152}\beq
\vt_\cH=w^\sh(d\cH) \label{gm152}
\eeq
of $\cH$. Hence, $S_\cH$ (\ref{gm151}) is an autonomous first
order dynamic equation (see forthcoming Remark \ref{gena70}),
called the Hamilton equation for a Hamiltonian $\cH$ with respect
to a Poisson structure $w$.

\begin{remark}\label{gena70} \mar{gena70}
Let $u$ be a vector field $u$ on $Z$. A closed subbundle $u(Z)$ of
the tangent bundle $TZ$ given by the coordinate relations
\mar{gm150}\beq
\dot z^\la=u^\la(z) \label{gm150}
\eeq
is said to be a first order autonomous dynamic equation
 on a manifold $Z$ \cite{book98,book00}. By a solution  of the autonomous first order dynamic
equation (\ref{gm150}) is meant an integral curve of a vector
field $u$.
\end{remark}

Relative to local canonical coordinates $(z^\la,q^i,p_i)$
(\ref{m111}) for a Poisson structure $w$ on $Z$ and corresponding
holonomic coordinates $(z^\la,q^i,p_i, \dot z^\la,\dot q^i, \dot
p_i)$ on $TZ$, the Hamilton equation (\ref{gm151}) and the
Hamiltonian vector field (\ref{gm152}) take a form
\mar{051,2}\ben
&& \dot q^i = \dr^i\cH, \qquad \dot p_i =-\dr_i\cH, \qquad \dot
z^\la=0, \label{051}\\
&& \vt_\cH= \dr^i\cH\dr_i -\dr_i\cH\dr^i. \label{052}
\een
Solutions  of the Hamilton equation (\ref{051}) are integral
curves of the Hamiltonian vector field (\ref{052}).

Let $(Z,w,\cH)$ be a Poisson Hamiltonian system. Its integral of
motion is a smooth function $F$ on $Z$ whose Lie derivative
\mar{z72}\beq
\bL_{\vt_\cH}F=\{\cH,F\} \label{z72}
\eeq
along the Hamiltonian vector field $\vt_\cH$ (\ref{052}) vanishes
in accordance with the equality (\ref{0115}). The equality
(\ref{z72}) is called the evolution equation.

It is readily observed that the Poisson bracket $\{F,F'\}$ of any
two integrals of motion $F$ and $F'$ also is an integral of
motion. Consequently, the integrals of motion of a Poisson
Hamiltonian system constitute a real Lie algebra.

Since
\be
\vt_{\{\cH,F\}}=[\vt_\cH,\vt_F], \qquad \{\cH,F\}=-\bL_{\vt_F}\cH,
\ee
the Hamiltonian vector field $\vt_F$ of any integral of motion $F$
of a Poisson Hamiltonian system is a symmetry both of the Hamilton
equation (\ref{051}) (Proposition \ref{082}) and a Hamiltonian
$\cH$ (Definition \ref{084}).

Let $(Z,\Om)$ be a symplectic manifold. The notion of a symplectic
Hamiltonian system is a repetition of the Poisson one, but all
expressions are rewritten in terms of a symplectic form $\Om$ as
follows.

A symplectic Hamiltonian system $(\Om,\cH)$ on a manifold $Z$ for
a Hamiltonian $\cH$ with respect to a symplectic structure $\Om$
is a set
\mar{z90}\beq
S_\cH=\op\bigcup_{z\in Z} \{v\in T_zZ:\, v\rfloor \Om +
d\cH(z)=0\}. \label{z90}
\eeq
As in the general case of Poisson Hamiltonian systems, the
symplectic one $(\Om,\cH)$ has a unique solution which is the
Hamiltonian vector field
\mar{053}\beq
\vt_\cH\rfloor\Om=-d\cH \label{053}
\eeq
of $\cH$. Hence, $S_\cH$ (\ref{z90}) is an autonomous first order
dynamic equation, called the Hamilton equation for a Hamiltonian
$\cH$ with respect to a symplectic structure $\Om$. Relative to
the local canonical coordinates $(q^i,p_i)$ for a symplectic
structure $\Om$, the Hamilton equation (\ref{z90}) and the
Hamiltonian vector field (\ref{053}) read
\mar{z66,054}\ben
&& \dot q^i = \dr^i\cH, \qquad \dot p_i =-\dr_i\cH, \label{z66}\\
&& \vt_\cH= \dr^i\cH\dr_i -\dr_i\cH\dr^i. \label{054}
\een

Integrals of motion of a symplectic Hamiltonian system are defined
just as those of a Poisson Hamiltonian system.

\subsection{Partially integrable systems on a Poisson manifold}

Completely integrable and superintegrable systems are considered
with respect to a symplectic structure on a manifold which holds
fixed from the beginning. As was mentioned above, partially
integrable system admits different compatible Poisson structures
(see Theorem \ref{bi92} below). Treating partially integrable
systems, we therefore are based on a wider notion of the dynamical
algebra \cite{jmp03,book10}.

Let we have $m$ mutually commutative vector fields $\{\vt_\la\}$
on a connected smooth real manifold $Z$ which are independent
almost everywhere on $Z$, i.e., the set of points, where the
multivector field $\op\w^m\vt_\la$ vanishes, is nowhere dense. We
denote by $\cS\subset C^\infty(Z)$ the $\mathbb R$-subring of
smooth real functions $f$ on $Z$ whose derivations $\vt_\la\rfloor
df$ vanish for all $\vt_\la$. Let $\cA$ be an $m$-dimensional Lie
$\cS$-algebra generated by the vector fields $\{\vt_\la\}$. One
can think of one of its elements as being an autonomous first
order dynamic equation on $Z$ and of the other as being its
integrals of motion in accordance with Definition \ref{026}. By
virtue of this definition, elements of $\cS$ also are regarded as
integrals of motion. Therefore, we agree to call $\cA$ a dynamical
algebra.

Given a commutative dynamical algebra $\cA$ on a manifold $Z$, let
$G$ be the group of local diffeomorphisms of $Z$ generated by the
flows of these vector fields. The orbits of $G$ are maximal
invariant submanifolds of $\cA$ (we follow the terminology of
\cite{susm}). Tangent spaces to these submanifolds form a
(non-regular) distribution $\cV\subset TZ$ whose maximal integral
manifolds  coincide with orbits of $G$. Let $z\in Z$ be a regular
point of the distribution $\cV$, i.e., $\op\w^m\vt_\la(z)\neq 0$.
Since the group $G$ preserves $\op\w^m\vt_\la$, a maximal integral
manifold $M$ of $\cV$ through $z$ also is  regular (i.e., its
points are regular). Furthermore, there exists an open
neighborhood $U$ of $M$ such that, restricted to $U$, the
distribution $\cV$ is an $m$-dimensional regular distribution on
$U$. Being involutive, it yields a foliation $\gF$ of $U$. A
regular open neighborhood $U$ of an invariant submanifold of $M$
is called  saturated
 if any invariant submanifold
through a point of $U$ belongs to $U$. For instance, any compact
invariant submanifold has such an open neighborhood.

\begin{definition} \label{cc1} \mar{cc1} Let $\cA$ be
an $m$-dimensional dynamical algebra on a regular Poisson manifold
$(Z,w)$. It is said to be a  partially integrable system if:

(a) its generators $\vt_\la$ are Hamiltonian vector fields of some
functions $S_\la\in \cS$ which are independent almost everywhere
on $Z$, i.e., the set of points where the $m$-form $\op\w^m
dS_\la$ vanishes is nowhere dense;

(b) all elements of $\cS\subset C^\infty(Z)$ are mutually  in
involution, i.e., their Poisson brackets equal zero.
\end{definition}

It follows at once from this definition that the Poisson structure
$w$ is at least of rank $2m$, and that $\cS$ is a commutative
Poisson algebra. We call the functions $S_\la$ in item (a) of
Definition \ref{cc1} the  generating functions of a partially
integrable system, which is uniquely defined by a family
$(S_1,\ldots,S_m)$ of these functions.

\begin{remark} \label{000} \mar{000}
If $2m=\di Z$ in Definition \ref{cc1}, we have a  completely
integrable system on a symplectic manifold $Z$ (see Definition
\ref{cmp21} below).
\end{remark}

If $2m<\di Z$, there exist different Poisson structures on $Z$
which bring a dynamical algebra $\cA$ into a partially integrable
system. Forthcoming Theorems \ref{bi0} and \ref{bi92} describe all
these Poisson structures around a regular invariant submanifold
$M\subset Z$ of $\cA$ \cite{jmp03}.

\begin{theorem} \label{bi0} \mar{bi0}
Let $\cA$ be a dynamical algebra, $M$ its regular invariant
submanifold, and $U$ a saturated regular open neighborhood of $M$.
Let us suppose that:

(i) the vector fields $\vt_\la$ on $U$ are complete,

(ii) the foliation $\gF$ of $U$ admits a transversal manifold
$\Si$ and its holonomy pseudogroup on $\Si$ is trivial,

(iii) the leaves of this foliation are mutually diffeomorphic.

\noindent Then the following hold.

(I) The leaves of $\cF$ are diffeomorphic to a toroidal cylinder
\mar{g120}\beq
\mathbb R^{m-r}\times T^r, \qquad 0\leq r\leq m. \label{g120}
\eeq

(II) There exists an open saturated neighborhood of $M$, say $U$
again, which is the trivial principal bundle
\mar{z10'}\beq
U=N\times(\mathbb R^{m-r}\times T^r)\ar^\pi N \label{z10'}
\eeq
over a domain $N\subset \mathbb R^{\di Z-m}$ with the structure
group (\ref{g120}).

(III) If $2m\leq\di Z$, there exists a Poisson structure of rank
$2m$ on $U$ such that $\cA$ is a partially integrable system in
accordance with Definition \ref{cc1}.
\end{theorem}

\begin{proof} We follow the proof in \cite{cush,laz} generalized to
the case of non-compact invariant submanifolds
\cite{jmp03,book05,book10}.

(I). Since $m$-dimensional leaves of the foliation $\cF$ admit $m$
complete independent vector fields, they are locally affine
manifolds diffeomorphic to a toroidal cylinder (\ref{g120}).

(II). By virtue of the condition (ii), the foliation $\gF$ of $U$
is a fibred manifold \cite{mol}. Then one can always choose an
open fibred neighborhood of its fibre $M$, say $U$ again, over a
domain $N$ such that this fibred manifold
\mar{d20}\beq
\pi:U\to N \label{d20}
\eeq
admits a section $\si$. In accordance with the well-known theorem
\cite{onish,palais} complete Hamiltonian vector fields $\vt_\la$
define an action of a simply connected Lie group $G$ on $Z$.
Because vector fields $\vt_\la$ are mutually commutative, it is
the additive group $\mathbb R^m$ whose group space is coordinated
by parameters $s^\la$ of the flows with respect to the basis
$\{e_\la=\vt_\la\}$ for its Lie algebra. The orbits of the group
$\mathbb R^m$ in $U\subset Z$ coincide with the fibres of the
fibred manifold (\ref{d20}). Since vector fields $\vt_\la$ are
independent everywhere on $U$, the action of $\mathbb R^m$ on $U$
is locally free, i.e., isotropy groups of points of $U$ are
discrete subgroups of the group $\mathbb R^m$. Given a point $x\in
N$, the action of $\mathbb R^m$ on the fibre $M_x=\pi^{-1}(x)$
factorizes as
\mar{d4}\beq
\mathbb R^m\times M_x\to G_x\times M_x\to M_x \label{d4}
\eeq
through the free transitive action on $M_x$ of the factor group
$G_x=\mathbb R^m/K_x$, where $K_x$ is the isotropy group of an
arbitrary point of $M_x$. It is the same group for all points of
$M_x$ because $\mathbb R^m$ is a commutative group. Clearly, $M_x$
is diffeomorphic to the group space of $G_x$. Since the fibres
$M_x$ are mutually diffeomorphic, all isotropy groups $K_x$ are
isomorphic to the group $\mathbb Z_r$ for some fixed $0\leq r\leq
m$. Accordingly, the groups $G_x$ are isomorphic to the additive
group (\ref{g120}). Let us bring the fibred manifold (\ref{d20})
into a principal bundle with the structure group $G_0$, where we
denote $\{0\}=\pi(M)$. For this purpose, let us determine
isomorphisms $\rho_x: G_0\to G_x$ of the group $G_0$ to the groups
$G_x$, $x\in N$. Then a desired fibrewise action of $G_0$ on $U$
is defined by the law
\mar{d5}\beq
G_0\times M_x\to\rho_x(G_0)\times M_x\to M_x. \label{d5}
\eeq
Generators of each isotropy subgroup $K_x$ of $\mathbb R^m$ are
given by $r$ linearly independent vectors of the group space
$\mathbb R^m$. One can show that there exist ordered collections
of generators $(v_1(x),\ldots,v_r(x))$ of the groups $K_x$ such
that $x\to v_i(x)$ are smooth $\mathbb R^m$-valued fields on $N$.
Indeed, given a vector $v_i(0)$ and a section $\si$ of the fibred
manifold (\ref{d20}), each field $v_i(x)=(s^\al_i(x))$ is a unique
smooth solution of the equation
\be
g(s_i^\al)\si(x)=\si(x), \qquad  (s_i^\al(0))=v_i(0),
\ee
on an open neighborhood of $\{0\}$. Let us consider the
decomposition
\be
v_i(0)=B_i^a(0) e_a + C_i^j(0) e_j, \qquad a=1,\ldots,m-r, \qquad
j=1,\ldots, r,
\ee
where $C_i^j(0)$ is a non-degenerate matrix. Since the fields
$v_i(x)$ are smooth, there exists an open neighborhood of $\{0\}$,
say $N$ again, where the matrices $C_i^j(x)$ are non-degenerate.
Then
\mar{d6}\beq
A(x)=\left(
\begin{array}{ccc}
\id & \qquad & (B(x)-B(0))C^{-1}(0) \\
0 & & C(x)C^{-1}(0)
\end{array}
\right) \label{d6}
\eeq
is a unique linear endomorphism
\be
(e_a,e_i)\to (e_a,e_j)A(x)
\ee
of the vector space $\mathbb R^m$ which transforms the frame
$\{v_\la(0)\}=\{e_a,v_i(0)\}$ into the frame
$\{v_\la(x)\}=\{e_a,\vt_i(x)\}$, i.e.,
\be
v_i(x)=B_i^a(x) e_a + C_i^j(x) e_j=B_i^a(0) e_a + C_i^j(0)
[A_j^b(x)e_b +A_j^k(x)e_k].
\ee
Since $A(x)$ (\ref{d6}) also is an automorphism of the group
$\mathbb R^m$ sending $K_0$ onto $K_x$, we obtain a desired
isomorphism $\rho_x$ of the group $G_0$ to the group $G_x$. Let an
element $g$ of the group $G_0$ be the coset of an element
$g(s^\la)$ of the group $\mathbb R^m$. Then it acts on $M_x$ by
the rule (\ref{d5}) just as the element $g((A_x^{-1})^\la_\bt
s^\bt)$ of the group $\mathbb R^m$ does. Since entries of the
matrix $A$ (\ref{d6}) are smooth functions on $N$, this action of
the group $G_0$ on $U$ is smooth. It is free, and $U/G_0=N$. Then
the fibred manifold (\ref{d20}) is a trivial principal bundle with
the structure group $G_0$. Given a section $\si$ of this principal
bundle, its trivialization $U=N\times G_0$ is defined by assigning
the points $\rho^{-1}(g_x)$ of the group space $G_0$ to the points
$g_x\si(x)$, $g_x\in G_x$, of a fibre $M_x$. Let us endow $G_0$
with the standard coordinate atlas $(r^\la)=(t^a,\vf^i)$ of the
group (\ref{g120}). Then $U$ admits the trivialization
(\ref{z10'}) with respect to the bundle coordinates
$(x^A,t^a,\vf^i)$ where $x^A$, $A=1,\ldots,\di Z-m$, are
coordinates on a base $N$. The vector fields $\vt_\la$ on $U$
relative to these coordinates read
\mar{ww25}\beq
\vt_a=\dr_a, \qquad \vt_i=-(BC^{-1})^a_i(x)\dr_a +
(C^{-1})_i^k(x)\dr_k.\label{ww25}
\eeq
Accordingly, the subring $\cS$ restricted to $U$ is the pull-back
$\pi^*C^\infty(N)$ onto $U$ of the ring of smooth functions on
$N$.

(III). Let us split the coordinates $(x^A)$ on $N$ into some $m$
coordinates $(J_\la)$ and the rest $\di Z- 2m$ coordinates
$(z^A)$. Then we can provide the toroidal domain $U$ (\ref{z10'})
with the Poisson bivector field
\mar{kk500}\beq
w=\dr^\la\w\dr_\la \label{kk500}
\eeq
of rank $2m$. The independent complete vector fields $\dr_a$ and
$\dr_i$ are Hamiltonian vector fields of the functions $S_a=J_a$
and $S_i=J_i$ on $U$ which are in involution with respect to the
Poisson bracket
\mar{bi12}\beq
\{f,f'\}=\dr^\la f\dr_\la f'-\dr_\la f\dr^\la f' \label{bi12}
\eeq
defined by the bivector field $w$ (\ref{kk500}). By virtue of the
expression (\ref{ww25}), the Hamiltonian vector fields
$\{\dr_\la\}$ generate the $\cS$-algebra $\cA$. Therefore,
$(w,\cA)$ is a partially integrable system.
\end{proof}

\begin{remark} \label{kk501} \mar{kk501}
Condition (ii) of Theorem \ref{bi0} is equivalent to that $U\to
U/G$ is a fibred manifold \cite{mol}. It should be emphasized that
a fibration in invariant submanifolds is a standard property of
integrable systems \cite{arn,bog,broer,gaeta,acang2,nekh94}. If
fibres of such a fibred manifold are assumed to be compact then
this fibred manifold is a fibre bundle (Theorem \ref{11t2}) and
vertical vector fields on it (e.g., in condition (i) of Theorem
\ref{bi0}) are complete (Theorem \ref{10b3}).
\end{remark}

A Poisson structure in Theorem \ref{bi0} is by no means unique.
Given the toroidal domain $U$ (\ref{z10'}) provided with bundle
coordinates $(x^A,r^\la)$, it is readily observed that, if a
Poisson bivector field on $U$ satisfies Definition \ref{cc1}, it
takes the form
\mar{bi20}\beq
w=w_1+w_2=w^{A\la}(x^B)\dr_A\w\dr_\la +
w^{\m\nu}(x^B,r^\la)\dr_\m\w \dr_\nu. \label{bi20}
\eeq
The converse also holds as follows.

\begin{theorem} \label{bi92} \mar{bi92}
For any Poisson bivector field $w$ (\ref{bi20}) of rank $2m$ on
the toroidal domain $U$ (\ref{z10'}), there exists a toroidal
domain $U'\subset U$ such that a dynamical algebra $\cA$ in
Theorem \ref{bi0} is a partially integrable system on $U'$.
\end{theorem}

\begin{remark}
It is readily observed that any Poisson bivector field $w$
(\ref{bi20}) fulfills condition (b) in Definition \ref{cc1}, but
condition (a) imposes a restriction on the toroidal domain $U$.
The key point is that the characteristic foliation $\cF$ of $U$
yielded by the Poisson bivector fields $w$ (\ref{bi20}) is the
pull-back of an $m$-dimensional foliation $\cF_N$ of the base $N$,
which is defined by the first summand $w_1$ (\ref{bi20}) of $w$.
With respect to the adapted coordinates $(J_\la,z^A)$,
$\la=1,\ldots, m$, on the foliated manifold $(N,\cF_N)$, the
Poisson bivector field $w$ reads
\mar{bi42}\beq
w= w^\m_\n(J_\la,z^A)\dr^\n\w \dr_\m +
w^{\m\n}(J_\la,z^A,r^\la)\dr_\m\w \dr_\n. \label{bi42}
\eeq
Then condition (a) in Definition \ref{cc1} is satisfied if
$N'\subset N$ is a domain of a coordinate chart $(J_\la,z^A)$ of
the foliation $\cF_N$. In this case, the dynamical algebra $\cA$
on the toroidal domain $U'=\pi^{-1}(N')$ is generated by the
Hamiltonian vector fields
\mar{bi93}\beq
\vt_\la=-w\lfloor dJ_\la=w^\m_\la\dr_\m \label{bi93}
\eeq
of the $m$ independent functions $S_\la=J_\la$.
\end{remark}

\begin{proof}
The characteristic distribution of the Poisson bivector field $w$
(\ref{bi20}) is spanned by the Hamiltonian vector fields
\mar{bi21}\beq
v^A=-w\lfloor dx^A=w^{A\m}\dr_\m \label{bi21}
\eeq
and the vector fields
\be
w\lfloor dr^\la= w^{A\la}\dr_A + 2w^{\m\la}\dr_\m.
\ee
Since $w$ is of rank $2m$, the vector fields $\dr_\m$ can be
expressed in the vector fields $v^A$ (\ref{bi21}). Hence, the
characteristic distribution of $w$ is spanned by the Hamiltonian
vector fields $v^A$ (\ref{bi21}) and the vector fields
\mar{bi25}\beq
v^\la=w^{A\la}\dr_A. \label{bi25}
\eeq
The vector fields (\ref{bi25}) are projected onto $N$. Moreover,
one can derive from the relation $[w,w]=0$ that they generate a
Lie algebra and, consequently, span an involutive distribution
$\cV_N$ of rank $m$ on $N$. Let $\cF_N$ denote the corresponding
foliation of $N$. We consider the pull-back $\cF=\pi^*\cF_N$ of
this foliation onto $U$ by the trivial fibration $\pi$ \cite{mol}.
Its leaves are the inverse images $\pi^{-1}(F_N)$ of leaves $F_N$
of the foliation $\cF_N$, and so is its characteristic
distribution
\be
T\cF=(T\pi)^{-1}(\cV_N).
\ee
This distribution is spanned by the vector fields $v^\la$
(\ref{bi25}) on $U$ and the vertical vector fields on $U\to N$,
namely,  the vector fields $v^A$ (\ref{bi21}) generating the
algebra $\cA$. Hence, $T\cF$ is the characteristic distribution of
the Poisson bivector field $w$. Furthermore, since $U\to N$ is a
trivial bundle, each leaf $\pi^{-1}(F_N)$ of the pull-back
foliation $\cF$ is the manifold product of a leaf $F_N$ of $N$ and
the toroidal cylinder $\mathbb R^{k-m}\times T^m$. It follows that
the foliated manifold $(U,\cF)$ can be provided with an adapted
coordinate atlas
\be
\{(U_\iota,J_\la,z^A,r^\la)\}, \qquad \la=1,\ldots, k, \qquad
A=1,\ldots,\di Z-2m,
\ee
such that $(J_\la,z^A)$ are adapted coordinates on the foliated
manifold $(N,\cF_N)$. Relative to these coordinates, the Poisson
bivector field (\ref{bi20}) takes the form (\ref{bi42}). Let $N'$
be the domain of this coordinate chart. Then the  dynamical
algebra $\cA$ on the toroidal domain $U'=\pi^{-1}(N')$ is
generated by the  Hamiltonian vector fields $\vt_\la$ (\ref{bi93})
of functions $S_\la=J_\la$. {}
\end{proof}

\begin{remark} \label{kk508} \mar{kk508}
Let us note that the coefficients $w^{\m\nu}$ in the expressions
(\ref{bi20}) and (\ref{bi42}) are affine in coordinates $r^\la$
because of the relation $[w,w]=0$ and, consequently, they are
constant on tori.
\end{remark}

Now, let $w$ and $w'$ be two different Poisson structures
(\ref{bi20}) on the toroidal domain (\ref{z10'}) which make a
commutative dynamical algebra $\cA$ into different partially
integrable systems $(w,\cA)$ and $(w',\cA)$.

\begin{definition} \label{cc2} \mar{cc2}
We agree to call the triple $(w,w',\cA)$ a  bi-Hamiltonian
partially integrable system if any Hamiltonian vector field
$\vt\in\cA$ with respect to $w$ possesses the same Hamiltonian
representation
\mar{bi71}\beq
\vt=-w\lfloor df=-w'\lfloor df, \qquad f\in\cS, \label{bi71}
\eeq
relative to $w'$, and {\it vice versa}.
\end{definition}

Definition \ref{cc2} establishes a {\it sui generis} equivalence
between the partially integrable systems $(w,\cA)$ and $(w',\cA)$.
Theorem \ref{bi72} below states that the triple $(w,w',\cA)$ is a
bi-Hamiltonian partially integrable system in accordance with
Definition \ref{cc2} iff the Poisson bivector fields $w$ and $w'$
(\ref{bi20}) differ only in the second terms $w_2$ and $w'_2$.
Moreover, these Poisson bivector fields admit a recursion operator
as follows.

\begin{theorem} \label{bi72} \mar{bi72}
(I) The triple $(w,w',\cA)$ is a bi-Hamiltonian partially
integrable system in accordance with Definition \ref{cc2} iff the
Poisson bivector fields $w$ and $w'$ (\ref{bi20}) differ in the
second terms $w_2$ and $w'_2$. (II) These Poisson bivector fields
admit a recursion operator.
\end{theorem}

\begin{proof} (I). It is easily justified that, if Poisson bivector
fields $w$ (\ref{bi20}) fulfil Definition \ref{cc2}, they are
distinguished only by the second summand $w_2$. Conversely, as
follows from the proof of Theorem \ref{bi92}, the characteristic
distribution of a Poisson bivector field $w$ (\ref{bi20}) is
spanned by the vector fields (\ref{bi21}) and (\ref{bi25}). Hence,
all Poisson bivector fields $w$ (\ref{bi20}) distinguished only by
the second summand $w_2$ have the same characteristic
distribution, and they bring $\cA$ into a partially integrable
system on the same toroidal domain $U'$. Then the condition in
Definition \ref{cc2} is easily justified.   (II). The result
follows from forthcoming Lemma \ref{pp1}.
\end{proof}

Given  a smooth real manifold $X$, let $w$ and $w'$ be Poisson
bivector fields of rank $2m$ on $X$, and let $w^\sh$ and $w'^\sh$
be the corresponding bundle homomorphisms (\ref{m51}). A
tangent-valued one-form $R$ on $X$ yields bundle endomorphisms
\mar{bi90}\beq
R: TX\to TX, \qquad R^*: T^*X\to T^*X. \label{bi90}
\eeq
It is called a  recursion operator if
\mar{pp0}\beq
w'^\sh=R\circ w^\sh=w^\sh\circ R^*. \label{pp0}
\eeq
Given a Poisson bivector field $w$ and a tangent valued one-form
$R$ such that $R\circ w^\sh=w^\sh\circ R^*$, the well-known
sufficient condition for $R\circ w^\sh$ to be a Poisson bivector
field is that the Nijenhuis torsion (\ref{spr613}) of $R$, seen as
a tangent-valued one-form, and the Magri -- Morosi concomitant of
$R$ and $w$ vanish \cite{1,2}. However, as we will see later,
recursion operators between Poisson bivector fields in Theorem
\ref{bi72} need not satisfy these conditions.

\begin{lemma} \label{pp1} \mar{pp1}
A recursion operator between Poisson structures of the same rank
exists iff their characteristic distributions coincide.
\end{lemma}

\begin{proof} It follows from the equalities
(\ref{pp0}) that a recursion operator $R$ sends the characteristic
distribution of $w$ to that of $w'$, and these distributions
coincide if $w$ and $w'$ are of the same rank. Conversely, let
regular Poisson structures $w$ and $w'$ possess the same
characteristic distribution $T\cF\to TX$ tangent to a foliation
$\cF$ of $X$. We have the exact sequences (\ref{pp2}) --
(\ref{pp3}). The bundle homomorphisms $w^\sh$ and $w'^\sh$
(\ref{m51}) factorize in a unique fashion (\ref{lmp03}) through
the bundle isomorphisms $w_\cF^\sh$ and $w'^\sh_\cF$
(\ref{lmp03}). Let us consider the inverse isomorphisms
\mar{pp13}\beq
w_\cF^\fl : T\cF\to T\cF^*, \qquad w'^\fl_\cF : T\cF\to T\cF^*
\label{pp13}
\eeq
and the compositions
\mar{pp10}\beq
R_\cF= w'^\sh_\cF\circ w_\cF^\fl: T\cF\to T\cF, \qquad R_\cF^*=
w_\cF^\fl \circ w'^\sh_\cF: T\cF^*\to T\cF^*. \label{pp10}
\eeq
There is the obvious relation
\be
w'^\sh_\cF=R_\cF\circ w^\sh_\cF=  w^\sh_\cF\circ R^*_\cF.
\ee
In order to obtain a recursion operator (\ref{pp0}), it suffices
to extend the morphisms $R_\cF$ and $R_\cF^*$ (\ref{pp10}) onto
$TX$ and $T^*X$, respectively. For this purpose, let us consider a
splitting
\be
&& \zeta: TX\to T\cF, \\
&& TX=T\cF\oplus (\id-i_\cF\circ\zeta)TX=T\cF\oplus E,
\ee
of the exact sequence (\ref{pp2}) and the dual splitting
\be
&& \zeta^* :T\cF^*\to T^*X, \\
&& T^*X=\zeta^*(T\cF^*)\oplus (\id-\zeta^*\circ i^*_\cF)T^*X=
\zeta^*(T\cF^*)\oplus E',
\ee
of the exact sequence (\ref{pp3}). Then the desired extensions are
\be
R=R_\cF\times \id E, \qquad R^*=(\zeta^*\circ R^*_\cF)\times \id
E'.
\ee
This recursion operator is invertible, i.e., the morphisms
(\ref{bi90}) are bundle isomorphisms.
\end{proof}

For instance,  the Poisson bivector field $w$ (\ref{bi20}) and the
Poisson bivector field
\be
w_0=w^{A\la}\dr_A\w\dr_\la
\ee
admit a recursion operator $w^\sh_0=R\circ w^\sh$ whose entries
are given by the equalities
\mar{bi24}\beq
R^A_B=\dl^A_B, \qquad R^\m_\n=\dl^\m_\n, \qquad R^A_\la=0, \qquad
w^{\m\la}=R^\la_Bw^{B\m}. \label{bi24}
\eeq
Its Nijenhuis torsion (\ref{spr613}) fails to vanish, unless
coefficients $w^{\m\la}$ are independent of coordinates $r^\la$.

Given a partially integrable system $(w,\cA)$ in Theorem
\ref{bi92}, the bivector field $w$ (\ref{bi42}) can be brought
into the canonical form (\ref{kk500}) with respect to partial
action-angle coordinates in forthcoming Theorem \ref{bi100}. This
theorem extends the Liouville -- Arnold theorem to the case of a
Poisson structure and a non-compact invariant submanifold
\cite{jmp03,book10}.

\begin{theorem} \label{bi100} \mar{bi100}
Given a partially integrable system $(w,\cA)$ on a Poisson
manifold $(U,w)$, there exists a toroidal domain $U'\subset U$
equipped with partial  action-angle coordinates $(I_a,I_i,z^A,
\tau^a,\f^i)$ such that, restricted to $U'$, a Poisson bivector
field takes the canonical form
\mar{bi101}\beq
w=\dr^a\w \dr_a + \dr^i\w \dr_i, \label{bi101}
\eeq
while the dynamical algebra $\cA$ is generated by Hamiltonian
vector fields of the action coordinate functions $S_a=I_a$,
$S_i=I_i$.
\end{theorem}

\begin{proof}
First, let us employ Theorem \ref{bi92} and restrict $U$ to the
toroidal domain, say  $U$ again, equipped with coordinates
$(J_\la,z^A,r^\la)$ such that the Poisson bivector field $w$ takes
the form (\ref{bi42}) and the algebra $\cA$ is generated by the
Hamiltonian vector fields $\vt_\la$ (\ref{bi93}) of $m$
independent functions $S_\la=J_\la$ in involution. Let us choose
these vector fields as new generators of the group $G$ and return
to Theorem \ref{bi0}. In accordance with this theorem, there
exists a toroidal domain $U'\subset U$ provided with another
trivialization $U'\to N'\subset N$ in toroidal cylinders $\mathbb
R^{m-r}\times T^r$ and endowed with bundle coordinates
$(J_\la,z^A,r^\la)$ such that the vector fields $\vt_\la$
(\ref{bi93}) take the form (\ref{ww25}). For the sake of
simplicity, let $U'$, $N'$ and $y^\la$ be denoted $U$, $N$ and
$r^\la=(t^a,\vf^i)$ again. Herewith, the Poisson bivector field
$w$ is given by the expression (\ref{bi42}) with new coefficients.
Let $w^\sh: T^*U\to TU$ be the corresponding bundle homomorphism.
It factorizes in a unique fashion (\ref{lmp03}):
\be
w^\sh: T^*U\ar^{i^*_\cF} T\cF^*\ar^{w^\sh_\cF} T\cF\ar^{i_\cF} TU
\ee
through the bundle isomorphism
\be
w_\cF^\sh: T\cF^*\to T\cF,  \qquad w^\sh_\cF:\al\to -w(x)\lfloor
\al.
\ee
Then the inverse isomorphisms $w_\cF^\fl : T\cF\to T\cF^*$
provides the foliated manifold $(U,\cF)$ with the leafwise
symplectic form
\mar{bi102,'}\ben
&& \Om_\cF=\Om^{\m\n}(J_\la,z^A,t^a) \wt dJ_\m\w \wt dJ_\n +
\Om_\m^\n(J_\la,z^A) \wt dJ_\n\w \wt dr^\m, \label{bi102}\\
&& \Om_\m^\al w^\m_\bt=\dl^\al_\bt, \qquad
\Om^{\al\bt}=-\Om^\al_\m\Om^\bt_\n w^{\m\n}. \label{bi102'}
\een
Let us show that it is $\wt d$-exact. Let $F$ be a leaf of the
foliation $\cF$ of $U$. There is a homomorphism of the de Rham
cohomology $H^*_{\rm DR}(U)$ of $U$ to the de Rham cohomology
$H^*_{\rm DR}(F)$ of $F$, and it factorizes through the leafwise
cohomology $H^*_\cF(U)$. Since $N$ is a domain of an adapted
coordinate chart of the foliation $\cF_N$, the foliation $\cF_N$
of $N$ is a trivial fibre bundle
\be
N=V\times W\to W.
\ee
Since $\cF$ is the pull-back onto $U$ of the foliation $\cF_N$ of
$N$, it also is a trivial fibre bundle
\mar{bi103}\beq
U=V\times W\times (\mathbb R^{k-m}\times T^m) \to W \label{bi103}
\eeq
over a domain $W\subset \mathbb R^{\di Z-2m}$. It follows that
\be
H^*_{\rm DR}(U)=H^*_{\rm DR}(T^r)=H^*_\cF(U).
\ee
Then the closed leafwise two-form $\Om_\cF$ (\ref{bi102}) is exact
due to the absence of the term $\Om_{\m\n}dr^\m\w dr^\nu$.
Moreover, $\Om_\cF=\wt d\Xi$ where $\Xi$ reads
\be
\Xi=\Xi^\al(J_\la,z^A,r^\la)\wt dJ_\al + \Xi_i(J_\la,z^A)\wt
d\vf^i
\ee
up to a $\wt d$-exact leafwise form. The Hamiltonian vector fields
$\vt_\la=\vt_\la^\m\dr_\m$ (\ref{ww25}) obey the relation
\mar{ww22'}\beq
\vt_\la\rfloor\Om_\cF=-\wt dJ_\la, \qquad \Om^\al_\bt
\vt^\bt_\la=\dl^\al_\la, \label{ww22'}
\eeq
which falls into the following conditions
\mar{bi110,1}\ben
&& \Om^\la_i=\dr^\la\Xi_i-\dr_i\Xi^\la, \label{bi110} \\
&& \Om^\la_a=-\dr_a\Xi^\la=\dl^\la_a. \label{bi111}
\een
The first of the relations (\ref{bi102'})  shows that
$\Om^\al_\bt$ is a non-degenerate matrix  independent of
coordinates $r^\la$. Then the condition (\ref{bi110}) implies that
$\dr_i\Xi^\la$ are  independent of $\vf^i$, and so are $\Xi^\la$
since $\vf^i$ are cyclic coordinates. Hence,
\mar{bi112,3}\ben
&&\Om^\la_i=\dr^\la\Xi_i, \label{bi112}\\
&& \dr_i\rfloor\Om_\cF=-\wt d\Xi_i. \label{bi113}
\een
Let us introduce new coordinates $I_a=J_a$, $I_i=\Xi_i(J_\la)$. By
virtue of the equalities (\ref{bi111}) and (\ref{bi112}), the
Jacobian of this coordinate transformation is regular. The
relation (\ref{bi113}) shows that $\dr_i$ are Hamiltonian vector
fields of the functions $S_i=I_i$. Consequently, we can choose
vector fields $\dr_\la$ as generators of the algebra $\cA$. One
obtains from the equality (\ref{bi111}) that
\be
\Xi^a=-t^a+E^a(J_\la,z^A)
\ee
and $\Xi^i$ are independent of $t^a$. Then the leafwise Liouville
form $\Xi$ reads
\be
\Xi=(-t^a+E^a(I_\la,z^A))\wt dI_a + E^i(I_\la, z^A)\wt dI_i + I_i
\wt d\vf^i.
\ee
The coordinate shifts
\be
\tau^a=-t^a+E^a(I_\la,z^A), \qquad \f^i=\vf^i-E^i(I_\la,z^A)
\ee
bring the leafwise form $\Om_\cF$ (\ref{bi102}) into the canonical
form
\be
\Om_\cF= \wt dI_a\w \wt d \tau^a + \wt dI_i\w \wt d\f^i
\ee
which ensures the canonical form (\ref{bi101}) of a Poisson
bivector field $w$.
\end{proof}

\subsection{Partially integrable system on a symplectic manifold}

Let $\cA$ be a commutative dynamical algebra on a $2n$-dimensional
connected symplectic manifold $(Z,\Om)$. Let it obey condition (a)
in Definition \ref{cc1}. However, condition (b) is not necessarily
satisfied, unless $m=n$, i.e., a system is completely integrable.
Therefore, we modify a definition of partially integrable systems
on a symplectic manifold.

\begin{definition} \label{kk506} \mar{kk506} A
collection $\{S_1,\ldots, S_m\}$ of $m\leq n$ independent smooth
real functions in involution on a symplectic manifold $(Z,\Om)$ is
called a  partially integrable system.
\end{definition}

\begin{remark} \label{kk507} \mar{507} By analogy with Definition \ref{cc1},
one can require that functions $S_\la$ in Definition \ref{kk506}
are independent almost everywhere on $Z$. However, all theorems
that we have proved above are concerned with partially integrable
systems restricted to some open submanifold $Z'\subset Z$ of
regular points of $Z$. Therefore, let us restrict functions
$S_\la$ to an open submanifold $Z'\subset Z$ where they are
independent, and we obtain a partially integrable system on a
symplectic manifold $(Z',\Om)$ which obeys Definition \ref{kk506}.
However, it may happen that $Z'$ is not connected. In this case,
we have different partially integrable systems on different
components of $Z'$.
\end{remark}

Given a partially integrable system $(S_\la)$ in Definition
\ref{kk506}, let us consider the map
\mar{g106}\beq
S:Z\to W\subset\mathbb R^m. \label{g106}
\eeq
Since functions $S_\la$ are everywhere independent, this map is a
submersion onto a domain $W\subset \mathbb R^m$, i.e., $S$
(\ref{g106}) is a fibred manifold of fibre dimension $2n-m$.
Hamiltonian vector fields $\vt_\la$ of functions $S_\la$ are
mutually commutative and independent. Consequently, they span an
$m$-dimensional involutive distribution on $Z$ whose maximal
integral manifolds constitute an isotropic foliation $\cF$ of $Z$.
Because functions $S_\la$ are constant on leaves of this
foliation, each fibre of a fibred manifold $Z\to W$ (\ref{g106})
is foliated by the leaves of the foliation $\cF$.

If $m=n$, we are in the case of a completely integrable system,
and leaves of $\cF$ are connected components of fibres of the
fibred manifold (\ref{g106}).

The Poincar\'e -- Lyapounov -- Nekhoroshev theorem
\cite{gaeta,nekh94} generalizes the Liouville -- Arnold one to a
partially integrable system if leaves of the foliation $\cF$ are
compact. It imposes a sufficient condition which Hamiltonian
vector fields $v_\la$ must satisfy in order that the foliation
$\cF$ is a fibred manifold \cite{gaeta,gaeta03}. Extending the
Poincar\'e -- Lyapounov -- Nekhoroshev theorem to the case of
non-compact invariant submanifolds, we in fact assume from the
beginning that these submanifolds form a fibred manifold
\cite{jmp03,book10}.

\begin{theorem} \label{nc6} \mar{nc6}
Let a partially integrable system $\{S_1,\ldots,S_m\}$ on a
symplectic manifold $(Z,\Om)$ satisfy the following conditions.

(i) The Hamiltonian vector fields $\vt_\la$ of $S_\la$ are
complete.

(ii) The foliation $\cF$ is a fibred manifold
\mar{d20'}\beq
\pi:Z\to N \label{d20'}
\eeq
whose fibres are mutually diffeomorphic.

\noindent Then the following hold.

(I) The fibres of $\cF$ are diffeomorphic to the toroidal cylinder
(\ref{g120}).

(II) Given a fibre $M$ of $\cF$, there exists its open saturated
neighborhood $U$ whose fibration (\ref{d20'}) is a trivial
principal bundle with the structure group (\ref{g120}).

(III) The neighborhood $U$ is provided with the bundle (partial
action-angle) coordinates
\be
(I_\la,p_s,q^s,y^\la)\to (I_\la,p_s,q^s), \qquad \la=1,\ldots,m,
\quad s=1,\ldots n-m,
\ee
such that: (i) the action coordinates $(I_\la)$ (\ref{cmp25}) are
expressed in the values of the functions $(S_\la)$, (ii) the angle
coordinates $(y^\la)$ (\ref{dd11}) are coordinates on a toroidal
cylinder, and (iii) the symplectic form $\Om$ on $U$ reads
\mar{cmp6'}\beq
\Om= dI_\la\w dy^\la + dp_s\w dq^s. \label{cmp6'}
\eeq
\end{theorem}

\begin{proof} (I) The proof of parts (I) and (II) repeats exactly that of parts (I) and (II)
of Theorem \ref{bi0}. As a result, let
\mar{d20a}\beq
\pi:U\to \pi(U)\subset N \label{d20a}
\eeq
be a trivial principal bundle with the structure group $\mathbb
R^{m-r}\times T^r$, endowed with the standard coordinate atlas
$(r^\la)=(t^a,\vf^i)$. Then $U$ (\ref{d20a}) admits a
trivialization
\mar{z10}\beq
U=\pi(U)\times(\mathbb R^{m-r}\times T^r)\to\pi(U) \label{z10}
\eeq
with respect to the fibre coordinates $(t^a,\vf^i)$. The
Hamiltonian vector fields $\vt_\la$ on $U$ relative to these
coordinates read (\ref{ww25}):
\mar{ww25a}\beq
\vt_a=\dr_a, \qquad \vt_i=-(BC^{-1})^a_i(x)\dr_a +
(C^{-1})_i^k(x)\dr_k. \label{ww25a}
\eeq
In order to specify coordinates on the base $\pi(U)$ of the
trivial bundle (\ref{z10}), let us consider the fibred manifold
$S$ (\ref{g106}). It factorizes as
\be
S: U\ar^\pi \pi(U)\ar^{\pi'} S(U)
\ee
through the fibre bundle $\pi$. The map $\pi'$ also is a fibred
manifold. One can always restrict the domain $\pi(U)$ to a chart
of the fibred manifold $\pi'$, say $\pi(U)$ again. Then $\pi(U)\to
S(U)$ is a trivial bundle $\pi(U)=S(U)\times V$, and so is $U\to
S(U)$. Thus, we have the composite bundle
\mar{z10a}\beq
U=S(U)\times V\times (\mathbb R^{m-r}\times T^r)\to S(U)\times
V\to S(U). \label{z10a}
\eeq
Let us provide its base $S(U)$ with the coordinates $(J_\la)$ such
that
\mar{cmp23}\beq
J_\la\circ S=S_\la. \label{cmp23}
\eeq
Then $\pi(U)$ can be equipped with the bundle coordinates $(J_\la,
x^A)$, $A=1,\ldots, 2(n-m)$, and $(J_\la, x^A, t^a,\vf^i)$ are
coordinates on $U$ (\ref{z10'}). Since fibres of $U\to \pi(U)$ are
isotropic, a symplectic form $\Om$ on $U$ relative to the
coordinates $(J_\la,x^A,r^\la)$ reads
\mar{d23}\ben
&&\Om=\Om^{\al\bt}dJ_\al\w dJ_\bt + \Om^\al_\bt dJ_\al\w dr^\bt +
\label{d23}\\
&& \qquad \Om_{AB}dx^A\w dx^B +\Om_A^\la dJ_\la\w dx^A +
  \Om_{A\bt} dx^A\w dr^\bt. \nonumber
\een
The Hamiltonian vector fields $\vt_\la=\vt_\la^\m\dr_\m$
(\ref{ww25a}) obey the relations $\vt_\la\rfloor\Om=-dJ_\la$ which
result in the coordinate conditions
\mar{ww22}\beq
  \Om^\al_\bt \vt^\bt_\la=\dl^\al_\la, \qquad \Om_{A\bt}\vt^\bt_\la=0.
\label{ww22}
\eeq
The first of them shows that $\Om^\al_\bt$ is a non-degenerate
matrix independent of coordinates $r^\la$. Then the second one
implies that $\Om_{A\bt}=0$. By virtue of the well-known K\"unneth
formula for the de Rham cohomology of manifold products, the
closed form $\Om$ (\ref{d23}) is exact, i.e., $\Om=d\Xi$ where the
Liouville form $\Xi$ is
\be
\Xi=\Xi^\al(J_\la,x^B,r^\la)dJ_\al + \Xi_i(J_\la,x^B) d\vf^i
+\Xi_A(J_\la,x^B,r^\la)dx^A.
\ee
Since $\Xi_a=0$ and $\Xi_i$ are independent of $\vf^i$, it follows
from the relations
\be
\Om_{A\bt}=\dr_A\Xi_\bt-\dr_\bt\Xi_A=0
\ee
that $\Xi_A$ are independent of coordinates $t^a$ and are at most
affine in $\vf^i$. Since $\vf^i$ are cyclic coordinates, $\Xi_A$
are independent of $\vf^i$. Hence, $\Xi_i$ are  independent of
coordinates $x^A$, and the Liouville form reads
\mar{ac2}\beq
\Xi=\Xi^\al(J_\la,x^B,r^\la)dJ_\al + \Xi_i(J_\la) d\vf^i
+\Xi_A(J_\la,x^B)dx^A. \label{ac2}
\eeq
Because entries $\Om^\al_\bt$ of $d\Xi=\Om$ are independent of
$r^\la$, we obtain the following.

(i) $\Om^\la_i=\dr^\la\Xi_i-\dr_i\Xi^\la$. Consequently,
$\dr_i\Xi^\la$ are independent of $\vf^i$, and so are $\Xi^\la$
since $\vf^i$ are cyclic coordinates. Hence,
$\Om^\la_i=\dr^\la\Xi_i$ and $\dr_i\rfloor\Om=-d\Xi_i$. A glance
at the last equality shows that $\dr_i$ are Hamiltonian vector
fields. It follows that, from the beginning, one can separate $m$
generating functions on $U$, say $S_i$ again, whose Hamiltonian
vector fields are tangent to invariant tori. In this case, the
matrix $B$ in the expressions (\ref{d6}) and (\ref{ww25a})
vanishes, and the Hamiltonian vector fields $\vt_\la$
(\ref{ww25a}) read
\mar{ww25'}\beq
\vt_a=\dr_a, \qquad \vt_i=(C^{-1})_i^k\dr_k. \label{ww25'}
\eeq
Moreover, the coordinates $t^a$ are exactly the flow parameters
$s^a$. Substituting the expressions (\ref{ww25'}) into the first
condition (\ref{ww22}), we obtain
\be
&& \Om=\Om^{\al\bt}dJ_\al\w dJ_\bt +dJ_a\w ds^a + C^i_k dJ_i\w
d\vf^k + \\
&& \qquad \Om_{AB}dx^A\w dx^B +\Om_A^\la dJ_\la\w dx^A.
\ee
It follows that $\Xi_i$ are independent of $J_a$, and so are
$C^k_i=\dr^k\Xi_i$.

(ii) $\Om^\la_a=-\dr_a\Xi^\la=\dl^\la_a$. Hence,
$\Xi^a=-s^a+E^a(J_\la)$ and $\Xi^i=E^i(J_\la,x^B)$ are independent
of $s^a$.

In view of items (i) -- (ii), the Liouville form $\Xi$ (\ref{ac2})
reads
\be
&& \Xi=(-s^a+E^a(J_\la,x^B))dJ_a + E^i(J_\la,x^B) dJ_i + \\
&& \qquad \Xi_i(J_j) d\vf^i + \Xi_A(J_\la,x^B)dx^A.
\ee
Since the matrix $\dr^k\Xi_i$ is non-degenerate, we can perform
the coordinate transformations
\mar{cmp25}\ben
&& I_a=J_a, \qquad I_i=\Xi_i(J_j), \label{cmp25}\\
&& r'^a=-s^a+E^a(J_\la,x^B), \quad
r'^i=\vf^i-E^j(J_\la,x^B)\frac{\dr J_j}{\dr I_i}. \nonumber
\een
These transformations bring $\Om$ into the form
\mar{d26}\beq
\Om= dI_\la\w d r'^\la +\Om_{AB}(I_\m,x^C) dx^A\w dx^B +
\Om_A^\la(I_\m,x^C) dI_\la\w dx^A. \label{d26}
\eeq
Since functions $I_\la$ are in involution and their Hamiltonian
vector fields $\dr_\la$ mutually commute, a point $z\in M$ has an
open neighborhood
\be
U_z=\pi(U_z)\times O_z, \qquad O_z\subset \mathbb R^{m-r}\times
T^r,
\ee
endowed with local Darboux coordinates $(I_\la,p_s,q^s, y^\la)$,
$s=1,\ldots,n-m$, such that the symplectic form $\Om$ (\ref{d26})
is given by the expression
\mar{dd12}\beq
\Om= dI_\la\w d y^\la + dp_s\w dq^s. \label{dd12}
\eeq
Here, $y^\la(I_\la,x^A,r'^\al)$ are local functions
\mar{dd11}\beq
y^\la=r'^\la + f^\la(I_\la,x^A) \label{dd11}
\eeq
on $U_z$. With the above-mentioned group $G$ of flows of
Hamiltonian vector fields $\vt_\la$, one can extend these
functions to an open neighborhood
\be
\pi(U_z)\times \mathbb R^{k-m}\times T^m
\ee
of $M$, say $U$ again, by the law
\be
y^\la(I_\la,x^A,G(z)^\al)= G(z)^\la + f^\la(I_\la,x^A).
\ee
Substituting the functions (\ref{dd11}) on $U$ into the expression
(\ref{d26}), one brings the symplectic form $\Om$ into the
canonical form (\ref{cmp6'}) on $U$.
\end{proof}

\begin{remark}
If one supposes from the beginning that leaves of the foliation
$\cF$ are compact, the conditions of Theorem \ref{nc6} can be
replaced with that $\cF$ is a fibred manifold (see Theorems
\ref{11t2} and \ref{10b3}).
\end{remark}

\subsection{Global partially integrable systems}

As was mentioned above, there is a topological obstruction to the
existence of global action-angle coordinates. Forthcoming Theorem
\ref{cmp35} is a global generalization of Theorem \ref{nc6}
\cite{jmp07a,book10,ijgmmp09a}.

\begin{theorem} \label{cmp35} \mar{cmp35}
Let a partially integrable system $\{S_1,\ldots,S_m\}$ on a
symplectic manifold $(Z,\Om)$ satisfy the following conditions.

(i) The Hamiltonian vector fields $\vt_\la$ of $S_\la$ are
complete.

(ii) The foliation $\cF$ is a fibre bundle
\mar{cmp40}\beq
\pi:Z\to N. \label{cmp40}
\eeq

(iii) Its base $N$ is simply connected and the cohomology
$H^2(N;\mathbb Z)$ of $N$ with coefficients in the constant sheaf
$\mathbb Z$ is trivial.

\noindent Then the following hold.

(I) The fibre bundle (\ref{cmp40}) is a trivial principal bundle
with the structure group (\ref{g120}), and we have a composite
fibred manifold
\mar{g107}\beq
S=\zeta\circ\pi: Z\ar N\ar W, \label{g107}
\eeq
where $N\to W$ however need not be a fibre bundle.

(II) The fibred manifold (\ref{g107}) is provided with the global
fibred action-angle coordinates
\be
(I_\la,x^A,y^\la)\to (I_\la,x^A)\to (I_\la), \qquad
\la=1,\ldots,m, \quad A=1,\ldots 2(n-m),
\ee
such that: (i) the action coordinates $(I_\la)$ (\ref{g142}) are
expressed in the values of the functions $(S_\la)$ and they
possess identity transition functions, (ii) the angle coordinates
$(y^\la)$ (\ref{g142}) are coordinates on a toroidal cylinder,
(iii) the symplectic form $\Om$ on $U$ reads
\mar{nc3'}\beq
\Om= dI_\la\w dy^\la + \Om_A^\la dI_\la\w dx^A+ \Om_{AB} dx^A\w
dx^B. \label{nc3'}
\eeq
\end{theorem}

\begin{proof} Following part (I) of the proof of Theorems \ref{bi0} and \ref{nc6}, one can show
that a typical fibre of the fibre bundle (\ref{cmp40}) is the
toroidal cylinder (\ref{g120}). Let us bring this fibre bundle
into a principal bundle with the structure group (\ref{g120}).
Generators of each isotropy subgroup $K_x$ of $\mathbb R^m$ are
given by $r$ linearly independent vectors $u_i(x)$ of a group
space $\mathbb R^m$. These vectors are assembled into an $r$-fold
covering $K\to N$. This is a subbundle of the trivial bundle
\mar{g101}\beq
N\times \mathbb R^m\to N \label{g101}
\eeq
whose local sections are local smooth sections of the fibre bundle
(\ref{g101}). Such a section over an open neighborhood of a point
$x\in N$ is given by a unique local solution $s^\la(x')e_\la$,
$e_\la=\vt_\la$, of the equation
\be
g(s^\la)\si(x')=\exp(s^\la e_\la)\si(x')=\si(x'), \qquad
s^\la(x)e_\la=u_i(x),
\ee
where $\si$ is an arbitrary local section of the fibre bundle
$Z\to N$ over an open neighborhood of $x$. Since $N$ is simply
connected, the covering $K\to N$ admits $r$ everywhere different
global sections $u_i$ which are global smooth sections
$u_i(x)=u^\la_i(x)e_\la$ of the fibre bundle (\ref{g101}). Let us
fix a point of $N$ further denoted by $\{0\}$. One can determine
linear combinations of the functions $S_\la$, say again $S_\la$,
such that $u_i(0)=e_i$, $i=m-r,\ldots,m$, and the group $G_0$ is
identified to the group $\mathbb R^{m-r}\times T^r$. Let $E_x$
denote an $r$-dimensional subspace of $\mathbb R^m$ passing
through the points $u_1(x),\ldots,u_r(x)$. The spaces $E_x$, $x\in
N$, constitute an $r$-dimensional subbundle $E\to N$ of the
trivial bundle (\ref{g101}). Moreover, the latter is split into
the Whitney sum of vector bundles $E\oplus E'$, where
$E'_x=\mathbb R^m/E_x$ \cite{hir}. Then there is a global smooth
section $\g$ of the trivial principal bundle $N\times GL(m,\mathbb
R)\to N$ such that $\g(x)$ is a morphism of $E_0$ onto $E_x$,
where
\be
u_i(x)=\g(x)(e_i)=\g_i^\la e_\la.
\ee
This morphism also is an automorphism of the group $\mathbb R^m$
sending $K_0$ onto $K_x$. Therefore, it provides a group
isomorphism $\rho_x: G_0\to G_x$. With these isomorphisms, one can
define the fibrewise action of the group $G_0$ on $Z$ given by the
law
\mar{d5'}\beq
G_0\times M_x\to\rho_x(G_0)\times M_x\to M_x. \label{d5'}
\eeq
Namely, let an element of the group $G_0$ be the coset
$g(s^\la)/K_0$ of an element $g(s^\la)$ of the group $\mathbb
R^m$. Then it acts on $M_x$ by the rule (\ref{d5'}) just as the
coset $g((\g(x)^{-1})^\la_\bt s^\bt)/K_x$ of an element
$g((\g(x)^{-1})^\la_\bt s^\bt)$ of $\mathbb R^m$ does. Since
entries of the matrix $\g$ are smooth functions on $N$, the action
(\ref{d5'}) of the group $G_0$ on $Z$ is smooth. It is free, and
$Z/G_0=N$. Thus, $Z\to N$ (\ref{cmp40}) is a principal bundle with
the structure group $G_0=\mathbb R^{m-r}\times T^r$.

Furthermore, this principal bundle over a paracompact smooth
manifold $N$ is trivial as follows. In accordance with the
well-known theorem \cite{hir}, its structure group $G_0$
(\ref{g120}) is reducible to the maximal compact subgroup $T^r$,
which also is the maximal compact subgroup of the group product
$\op\times^rGL(1,\mathbb C)$. Therefore, the equivalence classes
of $T^r$-principal bundles $\xi$ are defined as
\be
c(\xi)=c(\xi_1\oplus\cdots\oplus \xi_r)=(1+c_1(\xi_1))\cdots
(1+c_1(\xi_r))
\ee
by the Chern classes $c_1(\xi_i)\in H^2(N;\mathbb Z)$ of
$U(1)$-principal bundles $\xi_i$ over $N$ \cite{hir}. Since the
cohomology group $H^2(N;\mathbb Z)$ of $N$ is trivial, all Chern
classes $c_1$ are trivial, and the principal bundle $Z\to N$ over
a contractible base also is trivial. This principal bundle can be
provided with the following coordinate atlas.

Let us consider the fibred manifold $S:Z\to W$ (\ref{g106}).
Because functions $S_\la$ are constant on fibres of the fibre
bundle $Z\to N$ (\ref{cmp40}), the fibred manifold (\ref{g106})
factorizes through the fibre bundle (\ref{cmp40}), and we have the
composite fibred manifold (\ref{g107}). Let us provide the
principal bundle $Z\to N$ with a trivialization
\mar{g110}\beq
Z=N\times \mathbb R^{m-r}\times T^r\to N, \label{g110}
\eeq
whose fibres are endowed with the standard coordinates
$(r^\la)=(t^a,\vf^i)$ on the toroidal cylinder (\ref{g120}). Then
the composite fibred manifold (\ref{g107}) is provided with the
fibred coordinates
\mar{g108}\ben
&& (J_\la,x^A,t^a,\vf^i), \label{g108}\\
&& \la=1,\ldots, m, \quad A=1, \ldots, 2(n-m), \quad a=1, \ldots,
m-r, \quad i=1,\ldots, r, \nonumber
\een
where $J_\la$ (\ref{cmp23}) are coordinates on the base $W$
induced by Cartesian coordinates on $\mathbb R^m$, and $(J_\la,
x^A)$ are fibred coordinates on the fibred manifold $\zeta:N\to
W$. The coordinates $J_\la$ on $W\subset \mathbb R^m$ and the
coordinates $(t^a,\vf^i)$ on the trivial bundle (\ref{g110})
possess the identity transition functions, while the transition
function of coordinates $(x^A)$ depends on the coordinates
$(J_\la)$ in general.

The Hamiltonian vector fields $\vt_\la$ on $Z$ relative to the
coordinates (\ref{g108}) take the form
\mar{ww25b}\beq
\vt_\la=\vt_\la^a(x)\dr_a + \vt^i_\la(x)\dr_i. \label{ww25b}
\eeq
Since these vector fields commute (i.e., fibres of $Z\to N$ are
isotropic), the symplectic form $\Om$ on $Z$ reads
\mar{g103}\ben
&& \Om=\Om^\al_\bt dJ_\al\w dr^\bt + \Om_{\al A}dr^\al\w dx^A +
\Om^{\al\bt} dJ_\al\w dJ_\bt + \label{g103}\\
&& \qquad \Om^\al_A d J_\al\w dx^A +\Om_{AB} dx^A\w dx^B.
\nonumber
\een
This form is exact (see Lemma \ref{g144} below). Thus, we can
write
\mar{g113}\ben
&& \Om=d\Xi, \qquad \Xi=\Xi^\la(J_\al,x^B,r^\al) dJ_\la + \Xi_\la(J_\al,x^B)
dr^\la + \label{g113} \\
&& \qquad \Xi_A(J_\al,x^B,r^\al) dx^A. \nonumber
\een
Up to an exact summand, the Liouville form $\Xi$ (\ref{g113}) is
brought into the form
\be
\Xi=\Xi^\la(J_\al,x^B,r^\al) dJ_\la + \Xi_i(J_\al,x^B) d\vf^i
+\Xi_A(J_\al,x^B,r^\al) dx^A,
\ee
i.e., it does not contain the term $\Xi_a dt^a$.

The Hamiltonian vector fields $\vt_\la$ (\ref{ww25b}) obey the
relations $\vt_\la\rfloor\Om=-dJ_\la$, which result in the
coordinate conditions (\ref{ww22}). Then following the proof of
Theorem \ref{nc6}, we can show that a symplectic form $\Om$ on $Z$
is given by the expression (\ref{nc3'}) with respect to the
coordinates
\mar{g142}\ben
&& I_a=J_a, \qquad I_i=\Xi_i(J_j), \label{g142}\\
&& y^a = -\Xi^a=t^a-E^a(J_\la,x^B), \quad y^i
=\vf^i-\Xi^j(J_\la,x^B)\frac{\dr J_j}{\dr I_i}. \nonumber
\een
\end{proof}

\begin{lemma} \label{g144} \mar{g144}
The symplectic form $\Om$ (\ref{g103}) is exact.
\end{lemma}

\begin{proof} In accordance with the well-known K\"unneth formula,
the de Rham cohomology group of the product (\ref{g110}) reads
\be
H^2_{\rm DR}(Z)=H^2_{\rm DR}(N)\oplus H^1_{\rm DR}(N)\ot H^1_{\rm
DR}(T^r) \oplus H^2_{\rm DR}(T^r).
\ee
By the de Rham theorem \cite{hir}, the de Rham cohomology
$H^2_{\rm DR}(N)$ is isomorphic to the cohomology $H^2(N;\mathbb
R)$ of $N$ with coefficients in the constant sheaf $\mathbb R$. It
is trivial since
\be
H^2(N;\mathbb R)=H^2(N;\mathbb Z)\ot\mathbb R
\ee
where $H^2(N;\mathbb Z)$ is trivial. The first cohomology group
$H^1_{\rm DR}(N)$ of $N$ is trivial because $N$ is simply
connected. Consequently, $H^2_{\rm DR}(Z)=H^2_{\rm DR}(T^r)$. Then
the closed form $\Om$ (\ref{g103}) is exact since it does not
contain the term $\Om_{ij}d\vf^i\w d\vf^j$.
\end{proof}

\section{Superintegrable systems}

In comparison with partially integrable and completely integrable
systems integrals of motion of a superintegrable system need not
be in involution. We consider superintegrable systems on a
symplectic manifold. A key point is that invariant submanifolds of
any superintegrable system are maximal integral manifolds of a
certain partially integrable system (Proposition \ref{nc8}).
Completely integrable systems are particular superintegrable
systems (see Definition \ref{cmp21}).

Our goal are Theorem \ref{nc0'} for superintegrable systems,
Theorem \ref{cmp20} for completely integrable systems, Theorem
\ref{cmp34} for globally superintegrable systems, and Theorem
\ref{nc6'a} for globally completely integrable systems.

\begin{definition} \label{i0} \mar{i0}
Let $(Z,\Om)$ be a $2n$-dimensional connected symplectic manifold,
and let $(C^\infty(Z), \{,\})$ be the Poisson algebra of smooth
real functions on $Z$. A subset
\mar{i00}\beq
F=(F_1,\ldots,F_k), \qquad n\leq k<2n, \label{i00}
\eeq
of the Poisson algebra $C^\infty(Z)$ is called a superintegrable
system if the following conditions hold.

(i) All the functions $F_i$ (called the  generating functions of a
superintegrable system) are independent, i.e., the $k$-form
$\op\w^kdF_i$ nowhere vanishes on $Z$. It follows that the map
$F:Z\to \mathbb R^k$ is a submersion, i.e.,
\mar{nc4}\beq
F:Z\to N=F(Z) \label{nc4}
\eeq
is a fibred manifold over a domain (i.e., contractible open
subset) $N\subset\mathbb R^k$ endowed with the coordinates $(x_i)$
such that $x_i\circ F=F_i$.

(ii) There exist smooth real functions $s_{ij}$ on $N$ such that
\mar{nc1}\beq
\{F_i,F_j\}= s_{ij}\circ F, \qquad i,j=1,\ldots, k. \label{nc1}
\eeq

(iii) The matrix function $\bs$ with the entries $s_{ij}$
(\ref{nc1}) is of constant corank $m=2n-k$ at all points of $N$.
\end{definition}

\begin{remark} We restrict our consideration to the case of
generating functions which are independent everywhere on a
symplectic manifold $Z$ (see Remarks \ref{kk507} and \ref{cmp8}).
\end{remark}

If $k=n$, then $\bs=0$, and we are in the case of completely
integrable systems as follows.

\begin{definition} \label{cmp21} \mar{cmp21} The subset $F$, $k=n$, (\ref{i00})
of the Poisson algebra $C^\infty(Z)$ on a symplectic manifold
$(Z,\Om)$ is called a completely integrable system if $F_i$ are
independent functions in involution.
\end{definition}

If $k>n$, the matrix $\bs$ is necessarily non-zero. Therefore,
superintegrable systems also are called  non-commutative
completely integrable systems. If $k=2n-1$, a superintegrable
system is called  maximally superintegrable.

The following two assertions clarify the structure of
superintegrable systems \cite{fasso05,fior2,book10}.

\begin{proposition} \label{nc7} \mar{nc7} Given a symplectic manifold $(Z,\Om)$,
let $F:Z\to N$ be a fibred manifold such that, for any two
functions $f$, $f'$ constant on fibres of $F$, their Poisson
bracket $\{f,f'\}$ is so. By virtue of Theorem \ref{p11.3}, $N$ is
provided with an unique coinduced Poisson structure $\{,\}_N$ such
that $F$ is a Poisson morphism.
\end{proposition}

Since any function constant on fibres of $F$ is a pull-back of
some function on $N$, the superintegrable system (\ref{i00})
satisfies the condition of Proposition \ref{nc7} due to item (ii)
of Definition \ref{i0}. Thus, the base $N$ of the fibration
(\ref{nc4}) is endowed with a coinduced Poisson structure of
corank $m$. With respect to coordinates $x_i$ in item (i) of
Definition \ref{i0} its bivector field reads
\mar{cmp1}\beq
w=s_{ij}(x_k)\dr^i\w\dr^j. \label{cmp1}
\eeq

\begin{proposition} \label{nc8} \mar{nc8} Given a fibred manifold $F:Z\to N$ in
Proposition \ref{nc7}, the following conditions are equivalent
\cite{fasso05,libe}:

(i) the rank of the coinduced Poisson structure $\{,\}_N$ on $N$
equals $2\di N-\di Z$,

(ii) the fibres of $F$ are isotropic,

(iii) the fibres of $F$ are  maximal integral manifolds of the
involutive distribution spanned by the Hamiltonian vector fields
of the pull-back $F^*C$ of Casimir functions $C$ of the coinduced
Poisson structure (\ref{cmp1}) on $N$.
\end{proposition}

It is readily observed that the fibred manifold $F$ (\ref{nc4})
obeys condition (i) of Proposition \ref{nc8} due to item (iii) of
Definition \ref{i0}, namely, $k-m= 2(k-n)$.

Fibres of the fibred manifold $F$ (\ref{nc4}) are called the
invariant submanifolds.

\begin{remark} \label{cmp8} \mar{cmp8} In many physical models, condition (i) of Definition
\ref{i0} fails to hold. Just as in the case of partially
integrable systems, it can be replaced with that a subset
$Z_R\subset Z$ of regular points (where $\op\w^kdF_i\neq 0$) is
open and dense. Let $M$ be an invariant submanifold through a
regular point $z\in Z_R\subset Z$. Then it is regular, i.e.,
$M\subset Z_R$. Let $M$ admit a regular open saturated
neighborhood $U_M$ (i.e., a fibre of $F$ through a point of $U_M$
belongs to $U_M$). For instance, any compact invariant submanifold
$M$ has such a neighborhood $U_M$. The restriction of functions
$F_i$ to $U_M$ defines a superintegrable system on $U_M$ which
obeys Definition \ref{i0}. In this case, one says that a
superintegrable system is considered around its invariant
submanifold $M$.
\end{remark}

Let $(Z,\Om)$ be a $2n$-dimensional connected symplectic manifold.
Given the superintegrable system $(F_i)$ (\ref{i00}) on $(Z,\Om)$,
the well known Mishchenko -- Fomenko theorem (Theorem \ref{nc0})
states the existence of (semi-local) generalized action-angle
coordinates around its connected compact invariant submanifold
\cite{bols03,fasso05,mishc}. The Mishchenko -- Fomenko theorem is
extended to superintegrable systems with non-compact invariant
submanifolds (Theorem \ref{nc0'})
\cite{fior2,f2,book10,ijgmmp09a}. These submanifolds are
diffeomorphic to a toroidal cylinder
\mar{g120'}\beq
\mathbb R^{m-r}\times T^r, \qquad m=2n-k, \qquad 0\leq r\leq m.
\label{g120'}
\eeq

Note that the Mishchenko -- Fomenko theorem is mainly applied to
superintegrable systems whose integrals of motion form a compact
Lie algebra. The group generated by flows of their Hamiltonian
vector fields is compact. Since a fibration of a compact manifold
possesses compact fibres, invariant submanifolds of such a
superintegrable system are compact. With Theorem \ref{nc0'}, one
can describe superintegrable Hamiltonian system with an arbitrary
Lie algebra of integrals of motion.

Given a superintegrable system in accordance with Definition
\ref{i0}, the above mentioned generalization of the Mishchenko --
Fomenko theorem to non-compact invariant submanifolds states the
following.

\begin{theorem} \label{nc0'} \mar{nc0'} Let the Hamiltonian vector fields $\vt_i$ of the
functions $F_i$ be complete, and let the fibres of the fibred
manifold $F$ (\ref{nc4}) be connected and mutually diffeomorphic.
Then the following hold.

(I) The fibres of $F$ (\ref{nc4}) are diffeomorphic to the
toroidal cylinder (\ref{g120'}).

(II) Given a fibre $M$ of $F$ (\ref{nc4}), there exists its open
saturated neighborhood $U_M$ which is a trivial principal bundle
\mar{kk600}\beq
U_M=N_M\times \mathbb R^{m-r}\times T^r\ar^F N_M \label{kk600}
\eeq
with the structure group (\ref{g120'}).

(III) The neighborhood $U_M$ is provided with the bundle (
generalized action-angle) coordinates $(I_\la,p_s,q^s, y^\la)$,
$\la=1,\ldots, m$, $s=1,\ldots,n-m$, such that: (i) the
generalized angle coordinates $(y^\la)$ are coordinates on a
toroidal cylinder, i.e., fibre coordinates on the fibre bundle
(\ref{kk600}), (ii) $(I_\la,p_s,q^s)$ are coordinates on its base
$N_M$ where the action coordinates $(I_\la)$ are values of Casimir
functions of the coinduced Poisson structure $\{,\}_N$ on $N_M$,
and (iii) the symplectic form $\Om$ on $U_M$ reads
\mar{cmp6}\beq
\Om= dI_\la\w dy^\la + dp_s\w dq^s. \label{cmp6}
\eeq
\end{theorem}

\begin{proof}
It follows from item (iii) of Proposition \ref{nc8} that every
fibre $M$ of the fibred manifold (\ref{nc4}) is a maximal integral
manifolds of the involutive distribution spanned by the
Hamiltonian vector fields $\up_\la$ of the pull-back $F^*C_\la$ of
$m$ independent Casimir functions $\{C_1,\ldots, C_m\}$ of the
Poisson structure $\{,\}_N$ (\ref{cmp1}) on an open neighborhood
$N_M$ of a point $F(M)\in N$. Let us put $U_M=F^{-1}(N_M)$. It is
an open saturated neighborhood of $M$. Consequently, invariant
submanifolds of a superintegrable system (\ref{i00}) on $U_M$ are
maximal integral manifolds of the partially integrable system
\mar{cmp4}\beq
C^*=(F^*C_1, \ldots, F^*C_m), \qquad 0<m\leq n, \label{cmp4}
\eeq
on a symplectic manifold $(U_M,\Om)$. Therefore, statements (I) --
(III) of Theorem \ref{nc0'} are the corollaries of Theorem
\ref{nc6}. Its condition (i) is satisfied as follows. Let $M'$ be
an arbitrary fibre of the fibred manifold $F:U_M\to N_M$
(\ref{nc4}). Since
\be
F^*C_\la(z)= (C_\la\circ F)(z)= C_\la(F_i(z)),\qquad z\in M',
\ee
the Hamiltonian vector fields $\up_\la$ on $M'$ are $\mathbb
R$-linear combinations of Hamiltonian vector fields $\vt_i$ of the
functions $F_i$ It follows that $\up_\la$ are elements of a
finite-dimensional real Lie algebra of vector fields on $M'$
generated by the vector fields $\vt_i$. Since vector fields
$\vt_i$ are complete, the vector fields $\up_\la$ on $M'$ also are
complete (see forthcoming Remark \ref{zz95}). Consequently, these
vector fields are complete on $U_M$ because they are vertical
vector fields on $U_M\to N$. The proof of Theorem \ref{nc6} shows
that the action coordinates $(I_\la)$ are values of Casimir
functions expressed in the original ones $C_\la$.
\end{proof}

\begin{remark} \label{zz95} \mar{zz95} If complete vector fields on a
smooth manifold constitute a basis for a finite-dimensional real
Lie algebra, any element of this Lie algebra is complete
\cite{palais}.
\end{remark}

\begin{remark} Since an open neighborhood $U_M$ (\ref{kk600}) in item (II) of
Theorem \ref{nc0'} is not contractible, unless $r=0$, the
generalized action-angle coordinates on $U$ sometimes are called
semi-local.
\end{remark}

\begin{remark} \label{zz90} The condition of the completeness of Hamiltonian
vector fields of the generating functions $F_i$ in Theorem
\ref{nc0'} is rather restrictive (see the Kepler system in Section
3). One can replace this condition with that the Hamiltonian
vector fields of the pull-back onto $Z$ of Casimir functions on
$N$ are complete.
\end{remark}

If the conditions of Theorem \ref{nc0'} are replaced with that the
fibres of the fibred manifold $F$ (\ref{nc4}) are compact and
connected, this theorem restarts the Mishchenko -- Fomenko one as
follows.

\begin{theorem} \label{nc0} \mar{nc0}
Let the fibres of the fibred manifold $F$ (\ref{nc4}) be connected
and compact. Then they are diffeomorphic to a torus $T^m$, and
statements (II) -- (III) of Theorem \ref{nc0'} hold.
\end{theorem}

\begin{remark}
In Theorem \ref{nc0}, the Hamiltonian vector fields $\up_\la$ are
complete because fibres of the fibred manifold $F$ (\ref{nc4}) are
compact. As well known, any vector field on a compact manifold is
complete.
\end{remark}

If $F$ (\ref{i00}) is a completely integrable system, the
coinduced Poisson structure on $N$ equals zero, and the generating
functions $F_i$ are the pull-back of $n$ independent functions on
$N$. Then Theorems \ref{nc0} and \ref{nc0'} come to the Liouville
-- Arnold theorem \cite{arn,laz} and its generalization (Theorem
\ref{cmp20}) to the case of non-compact invariant submanifolds
\cite{fior,book05}, respectively. In this case, the partially
integrable system $C^*$ (\ref{cmp4}) is exactly the original
completely integrable system $F$.

\begin{theorem} \label{cmp20} \mar{cmp20} Given a completely integrable system,
$F$ in accordance with Definition \ref{cmp21}, let the Hamiltonian
vector fields $\vt_i$ of the functions $F_i$ be complete, and let
the fibres of the fibred manifold $F$ (\ref{nc4}) be connected and
mutually diffeomorphic. Then items (I) and (II) of Theorem
\ref{nc0'} hold, and its item (III) is replaced with the following
one.

(III') The neighborhood $U_M$ (\ref{kk600}) where $m=n$ is
provided with the bundle (generalized action-angle) coordinates
$(I_\la,y^\la)$, $\la=1,\ldots, n$, such that the angle
coordinates $(y^\la)$ are coordinates on a toroidal cylinder, and
the symplectic form $\Om$ on $U_M$ reads
\mar{cmp66}\beq
\Om= dI_\la\w dy^\la. \label{cmp66}
\eeq
\end{theorem}

To study a superintegrable system, one conventionally considers it
with respect to generalized action-angle coordinates. A problem is
that, restricted to an action-angle coordinate chart on an open
subbundle $U$ of the fibred manifold $Z\to N$ (\ref{nc4}), a
superintegrable system becomes different from the original one
since there is no morphism of the Poisson algebra $C^\infty(U)$ on
$(U,\Om)$ to that $C^\infty(Z)$ on $(Z,\Om)$. Moreover, a
superintegrable system on $U$ need not satisfy the conditions of
Theorem \ref{nc0'} because it may happen that the Hamiltonian
vector fields of the generating functions on $U$ are not complete.
To describe superintegrable systems in terms of generalized
action-angle coordinates, we therefore follow the notion of a
globally superintegrable system \cite{book10,ijgmmp09a}.

\begin{definition} \mar{cmp30a} \label{cmp30a} A superintegrable system
$F$ (\ref{i00}) on a symplectic manifold $(Z,\Om)$ in Definition
\ref{i0} is called  globally superintegrable if there exist global
generalized action-angle coordinates
\mar{cmp31a}\beq
(I_\la, x^A, y^\la), \qquad \la=1,\ldots,m, \qquad A=1,\ldots,
2(n-m), \label{cmp31a}
\eeq
such that: (i) the action coordinates $(I_\la)$ are expressed in
the values of some Casimir functions $C_\la$ on the Poisson
manifold $(N,\{,\}_N)$, (ii) the angle coordinates $(y^\la)$ are
coordinates on the toroidal cylinder (\ref{g120}), and (iii) the
symplectic form $\Om$ on $Z$ reads
\mar{cmp32}\beq
\Om= dI_\la\w d y^\la +\Om_{AB}(I_\m,x^C) dx^A\w dx^B.
\label{cmp32}
\eeq
\end{definition}

It is readily observed that the semi-local generalized
action-angle coordinates on $U$ in Theorem \ref{nc0'} are global
on $U$ in accordance with Definition \ref{cmp30a}.

Forthcoming Theorem \ref{cmp34} provides the sufficient conditions
of the existence of global generalized action-angle coordinates of
a superintegrable system on a symplectic manifold $(Z,\Om)$
\cite{jmp07a,book10,ijgmmp09a}. It generalizes the well-known
result for the case of compact invariant submanifolds
\cite{daz,fasso05}.

\begin{theorem} \label{cmp34} \mar{cmp34} A superintegrable system $F$ on
a symplectic manifold $(Z,\Om)$ is globally superintegrable if the
following conditions hold.

(i) Hamiltonian vector fields $\vt_i$ of the generating functions
$F_i$ are complete.

(ii) The fibred manifold $F$ (\ref{nc4}) is a fibre bundle with
connected fibres.

(iii) Its base $N$ is simply connected and the cohomology
$H^2(V;\mathbb Z)$ is trivial

(iv) The coinduced Poisson structure $\{,\}_N$ on a base $N$
admits $m$ independent Casimir functions $C_\la$.
\end{theorem}

\begin{proof} Theorem \ref{cmp34} is a corollary of
Theorem \ref{cmp35}. In accordance with Theorem \ref{cmp35}, we
have a composite fibred manifold
\mar{g150}\beq
Z\ar^F N\ar^C W, \label{g150}
\eeq
where $C:N\to W$ is a fibred manifold of level surfaces of the
Casimir functions $C_\la$ (which coincides with the symplectic
foliation of a Poisson manifold $N$). The composite fibred
manifold (\ref{g150}) is provided with the adapted fibred
coordinates $(J_\la, x^A, r^\la)$ (\ref{g108}), where $J_\la$ are
values of independent Casimir functions and $(r^\la)=(t^a,\vf^i)$
are coordinates on a toroidal cylinder. Since $C_\la=J_\la$ are
Casimir functions on $N$, the symplectic form $\Om$ (\ref{g103})
on $Z$ reads
\mar{g141}\beq
\Om=\Om^\al_\bt dJ_\al\w r^\bt + \Om_{\al A}dy^\al\w dx^A +
\Om_{AB} dx^A\w dx^B. \label{g141}
\eeq
In particular, it follows that transition functions of coordinates
$x^A$ on $N$ are independent of coordinates $J_\la$, i.e., $C:V\to
W$ is a trivial bundle. By virtue of Lemma \ref{g144}, the
symplectic form (\ref{g141}) is exact, i.e., $\Om=d\Xi$, where the
Liouville form $\Xi$ (\ref{g113}) is
\be
\Xi=\Xi^\la(J_\al,y^\m)dJ_\la + \Xi_i(J_\al) d\vf^i
+\Xi_A(x^B)dx^A.
\ee
Then the coordinate transformations (\ref{g142}):
\mar{g151}\ben
&& I_a=J_a, \qquad I_i=\Xi_i(J_j), \label{g151}\\
&&  y^a = -\Xi^a=t^a-E^a(J_\la), \qquad y^i
=\vf^i-\Xi^j(J_\la)\frac{\dr J_j}{\dr I_i}, \nonumber
\een
bring $\Om$ (\ref{g141}) into the form (\ref{cmp32}). In
comparison with the general case (\ref{g142}), the coordinate
transformations (\ref{g151}) are independent of coordinates $x^A$.
Therefore, the angle coordinates $y^i$ possess identity transition
functions on $N$.
\end{proof}

Theorem \ref{cmp34} restarts Theorem \ref{nc0'} if one considers
an open subset $V$ of $N$ admitting the Darboux coordinates $x^A$
on the symplectic leaves of $U$.

Note that, if invariant submanifolds of a superintegrable system
are assumed to be connected and compact, condition (i) of Theorem
\ref{cmp34} is unnecessary since vector fields $\vt_\la$ on
compact fibres of $F$ are complete. Condition (ii) also holds by
virtue of Theorem \ref{11t2}. In this case, Theorem \ref{cmp34}
reproduces the well known result in \cite{daz}.

If  $F$ in Theorem \ref{cmp34} is a completely integrable system,
the coinduced Poisson structure on $N$ equals zero, the generating
functions $F_i$ are the pull-back of $n$ independent functions on
$N$, and Theorem \ref{cmp34} takes the following form
\cite{jmp07a,book10}.

\begin{theorem} \label{nc6'a} \mar{nc6'a}
Let a completely integrable system $\{F_1,\ldots, F_n\}$ on a
symplectic manifold $(Z,\Om)$ satisfy the following conditions.

(i) The Hamiltonian vector fields $\vt_i$ of $F_i$ are complete.

(ii) The fibred manifold $F$ (\ref{nc4}) is a fibre bundle with
connected fibres over a simply connected base $N$ whose cohomology
$H^2(N,\mathbb Z)$ is trivial.

\noindent Then the following hold.

(I) The fibre bundle $F$ (\ref{nc4}) is a trivial principal bundle
with the structure group $\mathbb R^{2n-r}\times T^r$.

(II) The symplectic manifold $Z$ is provided with the global
Darboux coordinates $(I_\la,y^\la)$ such that $\Om= dI_\la\w
dy^\la$.
\end{theorem}

It follows from the proof of Theorem \ref{cmp35} that its
condition (iii) and, accordingly, condition (iii) of Theorem
\ref{cmp34} guarantee that fibre bundles $F$ in conditions (ii) of
these theorems are trivial. Therefore, Theorem \ref{cmp34} can be
reformulated as follows.

\begin{theorem} \mar{cmp36} \label{cmp36} A superintegrable system $F$ on
a symplectic manifold $(Z,\Om)$ is globally superintegrable iff
the following conditions hold.

(i) The fibred manifold $F$ (\ref{nc4}) is a trivial fibre bundle.

(ii) The coinduced Poisson structure $\{,\}_N$ on a base $N$
admits $m$ independent Casimir functions $C_\la$ such that
Hamiltonian vector fields of their pull-back $F^*C_\la$ are
complete.
\end{theorem}

\begin{remark} \label{zz96} \mar{zz96}
It follows from Remark \ref{zz95} and condition (ii) of Theorem
\ref{cmp36} that a Hamiltonian vector field of the the pull-back
$F^*C$ of any Casimir function $C$ on a Poisson manifold $N$ is
complete.
\end{remark}

In autonomous Hamiltonian mechanics, one considers superintegrable
systems whose generating functions are integrals of motion, i.e.,
they are in involution with a Hamiltonian $\cH$, and the functions
$(\cH,F_1,\ldots,F_k)$ are nowhere independent, i.e.,
\mar{cmp11,'}\ben
&&\{\cH, F_i\}=0, \label{cmp11}\\
&& d\cH\w(\op\w^kdF_i)=0. \label{cmp11'}
\een.

In order that an evolution of a Hamiltonian system can be defined
at any instant $t\in\mathbb R$, one supposes that the Hamiltonian
vector field of its Hamiltonian is complete. By virtue of Remark
\ref{zz96} and forthcoming Proposition \ref{cmp12b}, a Hamiltonian
of a superintegrable system always satisfies this condition.

\begin{proposition} \label{cmp12b} \mar{cmp12b}
It follows from the equality (\ref{cmp11'}) that a Hamiltonian
$\cH$ is constant on the invariant submanifolds. Therefore, it is
the pull-back of a function on $N$ which is a Casimir function of
the Poisson structure (\ref{cmp1}) because of the conditions
(\ref{cmp11}).
\end{proposition}

Proposition \ref{cmp12b} leads to the following.

\begin{proposition} \label{zz1} \mar{zz1}
Let $\cH$ be a Hamiltonian of a globally superintegrable system
provided with the generalized action-angle coordinates $(I_\la,
x^A, y^\la)$ (\ref{cmp31a}). Then a Hamiltonian $\cH$ depends only
on the action coordinates $I_\la$. Consequently, the Hamilton
equation of a globally superintegrable system take the form
\be \dot
y^\la=\frac{\dr \cH}{\dr I_\la}, \qquad I_\la={\rm const.}, \qquad
x^A={\rm const.}
\ee
\end{proposition}

Following the original Mishchenko -- Fomenko theorem, let us
mention superintegrable systems whose generating functions
$\{F_1,\ldots,F_k\}$ form a $k$-dimensional real Lie algebra $\cG$
of corank $m$ with the commutation relations
\mar{zz60}\beq
\{F_i,F_j\}= c_{ij}^h F_h, \qquad c_{ij}^h={\rm const.}
\label{zz60}
\eeq
Then $F$ (\ref{nc4}) is a momentum mapping of $Z$ to the Lie
coalgebra $\cG^*$ provided with the coordinates $x_i$ in item (i)
of Definition \ref{i0} \cite{book10,guil}. In this case, the
coinduced Poisson structure $\{,\}_N$ coincides with the canonical
Lie -- Poisson structure on $\cG^*$ given by the Poisson bivector
field
\be
w=\frac12 c_{ij}^h x_h\dr^i\w\dr^j.
\ee
Let $V$ be an open subset of $\cG^*$ such that conditions (i) and
(ii) of Theorem \ref{cmp36} are satisfied. Then an open subset
$F^{-1}(V)\subset Z$ is provided with the generalized action-angle
coordinates.

\begin{remark}
Let Hamiltonian vector fields $\vt_i$ of the generating functions
$F_i$ which form a Lie algebra $\cG$ be complete. Then they define
a locally free Hamiltonian action on $Z$ of some simply connected
Lie group $G$ whose Lie algebra is isomorphic to $\cG$
\cite{onish,palais}. Orbits of $G$ coincide with $k$-dimensional
maximal integral manifolds of the regular distribution $\cV$ on
$Z$ spanned by Hamiltonian vector fields $\vt_i$ \cite{susm}.
Furthermore, Casimir functions of the Lie -- Poisson structure on
$\cG^*$ are exactly the coadjoint invariant functions on $\cG^*$.
They are constant on orbits of the coadjoint action of $G$ on
$\cG^*$ which coincide with leaves of the symplectic foliation of
$\cG^*$.
\end{remark}

\begin{theorem} \label{zz34} \mar{zz34} Let a globally superintegrable Hamiltonian system
on a symplectic manifold $Z$ obey the following conditions.

(i) It is maximally superintegrable.

(ii) Its Hamiltonian $\cH$ is regular, i.e, $d\cH$ nowhere
vanishes.

(iii) Its generating functions $F_i$ constitute a finite
dimensional real Lie algebra and their Hamiltonian vector fields
are complete.

\noindent Then any integral of motion of this Hamiltonian system
is the pull-back of a function on a base $N$ of the fibration $F$
(\ref{nc4}). In other words, it is expressed in the integrals of
motion $F_i$.
\end{theorem}

\begin{proof}
The proof is based on the following. A Hamiltonian vector field of
a function $f$ on $Z$ lives in the one-codimensional regular
distribution $\cV$ on $Z$ spanned by Hamiltonian vector fields
$\vt_i$ iff $f$ is the pull-back of a function on a base $N$ of
the fibration $F$ (\ref{nc4}). A Hamiltonian $\cH$ brings $Z$ into
a fibred manifold of its level surfaces whose vertical tangent
bundle coincide with $\cV$. Therefore, a Hamiltonian vector field
of any integral of motion of $\cH$ lives in $\cV$.
\end{proof}

It may happen that, given a Hamiltonian $\cH$ of a Hamiltonian
system on a symplectic manifold $Z$, we have different
superintegrable Hamiltonian systems on different open subsets of
$Z$. For instance, this is the case of the Kepler system.

\section{Global Kepler system}

We consider the Kepler system on a plane $\mathbb R^2$
\cite{book10,ijgmmp09a}.

Its phase space is $T^*\mathbb R^2=\mathbb R^4$ provided with the
Cartesian coordinates $(q_i,p_i)$, $i=1,2$, and the canonical
symplectic form
\mar{zz43}\beq
\Om_T=\op\sum_idp_i\w dq_i. \label{zz43}
\eeq
Let us denote
\be
p=\left(\op\sum_i(p_i)^2\right)^{1/2}, \qquad
r=\left(\op\sum_i(q^i)^2\right)^{1/2}, \qquad
(p,q)=\op\sum_ip_iq_i.
\ee
An autonomous Hamiltonian of the Kepler system reads
\mar{zz41}\beq
\cH=\frac12p^2-\frac1r. \label{zz41}
\eeq
The Kepler system is a Hamiltonian system on a symplectic manifold
\mar{zz42}\beq
Z=\mathbb R^4\setminus \{0\} \label{zz42}
\eeq
endowed with the symplectic form $\Om_T$ (\ref{zz43}).

Let us consider the functions
\mar{zz44,5}\ben
&& M_{12}=-M_{21}=q_1p_2-q_2p_1,  \label{zz44}\\
&& A_i=\op\sum_j M_{ij}p_j -\frac{q_i}{r}=q_ip^2 -p_i(p,q)-\frac{q_i}{r}, \qquad i=1,2,
\label{zz45}
\een
on the symplectic manifold $Z$ (\ref{zz42}). It is readily
observed that they are integrals of motion of the Hamiltonian
$\cH$ (\ref{zz41}) where $M_{12}$ is an angular momentum and
$(A_i)$ is a Rung -- Lenz vector. Let us denote
\mar{zz50}\beq
M^2=(M_{12})^2, \qquad  A^2=(A_1)^2 + (A_a)^2=2M^2\cH+1.
\label{zz50}
\eeq

Let $Z_0\subset Z$ be a closed subset of points where $M_{12}=0$.
A direct computation shows that the functions $(M_{12},A_i)$
(\ref{zz44}) -- (\ref{zz45}) are independent of an open
submanifold
\mar{zz47}\beq
U=Z\setminus Z_0 \label{zz47}
\eeq
of $Z$. At the same time, the functions $(\cH,M_{12},A_i)$ are
independent nowhere on $U$ because it follows from the expression
(\ref{zz50}) that
\mar{zz52}\beq
\cH=\frac{A^2-1}{2M^2} \label{zz52}
\eeq
on $U$ (\ref{zz47}). The well known dynamics of the Kepler system
shows that the Hamiltonian vector field of its Hamiltonian is
complete on $U$ (but not on Z).

The Poisson bracket of integrals of motion $M_{12}$ (\ref{zz44})
and $A_i$ (\ref{zz45}) obeys the relations
\mar{zz56,7}\ben
&& \{M_{12},A_i\}=\eta_{2i}A_1 -\eta_{1i}A_2, \label{zz56}\\
&& \{A_1,A_2\}=2\cH M_{12}=\frac{A^2-1}{M_{12}}, \label{zz57}
\een
where $\eta_{ij}$ is an Euclidean metric on $\mathbb R^2$. It is
readily observed that these relations take the form (\ref{nc1}).
However, the matrix function $\bs$ of the relations (\ref{zz56})
-- (\ref{zz57}) fails to be of constant rank at points where
$\cH=0$. Therefore, let us consider the open submanifolds
$U_-\subset U$ where $\cH<0$ and $U_+$ where $\cH>0$. Then we
observe that the Kepler system with the Hamiltonian $\cH$
(\ref{zz41}) and the integrals of motion $(M_{ij},A_i)$
(\ref{zz44}) -- (\ref{zz45}) on $U_-$ and the Kepler system with
the Hamiltonian $\cH$ (\ref{zz41}) and the integrals of motion
$(M_{ij},A_i)$ (\ref{zz44}) -- (\ref{zz45}) on $U_+$ are
superintegrable Hamiltonian systems. Moreover, these
superintegrable systems can be brought into the form (\ref{zz60})
as follows.

Let us replace the integrals of motions $A_i$ with the integrals
of motion
\mar{zz61}\beq
L_i=\frac{A_i}{\sqrt{-2\cH}} \label{zz61}
\eeq
on $U_-$, and with the integrals of motion
\mar{zz62}\beq
K_i=\frac{A_i}{\sqrt{2\cH}} \label{zz62}
\eeq
on $U_+$.

The superintegrable system $(M_{12},L_i)$ on $U_-$ obeys the
relations
\mar{zz66,7}\ben
&& \{M_{12},L_i\}=\eta_{2i}L_1 -\eta_{1i}L_2, \label{zz66}\\
&& \{L_1,L_2\}=-M_{12}. \label{zz67}
\een
Let us denote $M_{i3}=-L_i$ and put the indexes
$\m,\nu,\al,\bt=1,2,3$. Then the relations (\ref{zz66}) --
(\ref{zz67}) are brought into the form
\mar{zz68}\beq
\{M_{\m\nu},M_{\al\bt}\}=\eta_{\m\bt}M_{\nu\al} +
\eta_{\nu\al}M_{\m\bt} -
\eta_{\m\al}M_{\nu\bt}-\eta_{\nu\bt}M_{\m\al} \label{zz68}
\eeq
where $\eta_{\m\nu}$ is an Euclidean metric on $\mathbb R^3$. A
glance at the expression (\ref{zz68}) shows that the integrals of
motion $M_{12}$ (\ref{zz44}) and $L_i$ (\ref{zz61}) constitute the
Lie algebra $\cG=so(3)$. Its corank equals 1. Therefore the
superintegrable system $(M_{12}, L_i)$ on $U_-$ is maximally
superintegrable. The equality (\ref{zz52}) takes the form
\mar{zz100}\beq
M^2 +L^2=-\frac1{2\cH}. \label{zz100}
\eeq

The superintegrable system $(M_{12},K_i)$ on $U_+$ obeys the
relations
\mar{zz76,7}\ben
&& \{M_{12},K_i\}=\eta_{2i}K_1 -\eta_{1i}K_2, \label{zz76}\\
&& \{K_1,K_2\}=M_{12}. \label{zz77}
\een
Let us denote $M_{i3}=-K_i$ and put the indexes
$\m,\nu,\al,\bt=1,2,3$. Then the relations (\ref{zz76}) --
(\ref{zz77}) are brought into the form
\mar{zz78}\beq
\{M_{\m\nu},M_{\al\bt}\}=\rho_{\m\bt}M_{\nu\al} +
\rho_{\nu\al}M_{\m\bt} -
\rho_{\m\al}M_{\nu\bt}-\rho_{\nu\bt}M_{\m\al} \label{zz78}
\eeq
where $\rho_{\m\nu}$ is a pseudo-Euclidean metric of signature
$(+,+,-)$ on $\mathbb R^3$. A glance at the expression
(\ref{zz78}) shows that the integrals of motion $M_{12}$
(\ref{zz44}) and $K_i$ (\ref{zz62}) constitute the Lie algebra
$so(2,1)$. Its corank equals 1. Therefore the superintegrable
system $(M_{12}, K_i)$ on $U_+$ is maximally superintegrable. The
equality (\ref{zz52}) takes the form
\mar{zz101}\beq
K^2 -M^2=\frac1{2\cH}. \label{zz101}
\eeq

Thus, the Kepler system on a phase space $\mathbb R^4$ falls into
two different maximally superintegrable systems on open
submanifolds $U_-$ and $U_+$ of $\mathbb R^4$. We agree to call
them the Kepler superintegrable systems on $U_-$ and $U_+$,
respectively.

Let us study the first one and put
\mar{zz102}\ben
&& F_1=-L_1, \qquad F_2=-L_2, \qquad F_3=-M_{12}, \label{zz102}\\
&& \{F_1,F_2\}=F_3, \qquad \{F_2,F_3\}=F_1, \qquad
\{F_3,F_1\}=F_2.\nonumber
\een
We have a fibred manifold
\mar{zz103}\beq
F: U_-\to N\subset\cG^*, \label{zz103}
\eeq
which is the momentum mapping to the Lie coalgebra
$\cG^*=so(3)^*$, endowed with the coordinates $(x_i)$ such that
integrals of motion $F_i$ on $\cG^*$ read $F_i=x_i$. A base $N$ of
the fibred manifold (\ref{zz103}) is an open submanifold of
$\cG^*$ given by the coordinate condition $x_3\neq 0$. It is a
union of two contractible components defined by the conditions
$x_3>0$ and $x_3<0$. The coinduced Lie -- Poisson structure on $N$
takes the form
\mar{j51b}\beq
w= x_2\dr^3\w\dr^1 + x_3\dr^1\w\dr^2 + x_1\dr^2\w\dr^3.
\label{j51b}
\eeq

The coadjoint action of $so(3)$ on $N$ reads
\mar{kk605}\beq
\ve_1=x_3\dr^2-x_2\dr^3, \quad \ve_2=x_1\dr^3-x_3\dr^1, \quad
\ve_3=x_2\dr^1-x_1\dr^2. \label{kk605}
\eeq
The orbits of this coadjoint action are given by the equation
\mar{kk606}\beq
x_1^2 + x_2^2 + x_3^2={\rm const}. \label{kk606}
\eeq
They are the level surfaces of the Casimir function
\be
C=x_1^2 + x_2^2 + x_3^2
\ee
and, consequently, the Casimir function
\mar{zz120}\beq
h=-\frac12(x_1^2 + x_2^2 + x_3^2)^{-1}. \label{zz120}
\eeq
A glance at the expression (\ref{zz100}) shows that the pull-back
$F^*h$ of this Casimir function (\ref{zz120}) onto $U_-$ is the
Hamiltonian $\cH$ (\ref{zz41}) of the Kepler system on $U_-$.

As was mentioned above, the Hamiltonian vector field of $F^*h$ is
complete. Furthermore, it is known that invariant submanifolds of
the superintegrable Kepler system on $U_-$ are compact. Therefore,
the fibred manifold $F$ (\ref{zz103}) is a fibre bundle in
accordance with Theorem \ref{11t2}. Moreover, this fibre bundle is
trivial because $N$ is a disjoint union of two contractible
manifolds. Consequently, it follows from Theorem \ref{cmp36} that
the Kepler superintegrable system on $U_-$ is globally
superintegrable, i.e., it admits global generalized action-angle
coordinates as follows.

The Poisson manifold $N$ (\ref{zz103}) can be endowed with the
coordinates
\mar{j52x}\beq
(I,x_1,\g), \qquad I<0, \qquad \g\neq\frac{\pi}2,\frac{3\pi}2,
\label{j52x}
\eeq
defined by the equalities
\mar{j52}\ben
&& I=-\frac12(x_1^2 + x_2^2 + x_3^2)^{-1}, \label{j52}\\
&& x_2=\left(-\frac1{2I}-x_1^2\right)^{1/2}\sin\g, \qquad
x_3=\left(-\frac1{2I}-x_1^2\right)^{1/2}\cos\g. \nonumber
\een
It is readily observed that the coordinates (\ref{j52x}) are
Darboux coordinates of the Lie -- Poisson structure (\ref{j51b})
on $U_-$, namely,
\mar{j53}\beq
w=\frac{\dr}{\dr x_1}\w \frac{\dr}{\dr \g}. \label{j53}
\eeq

Let $\vt_I$ be the Hamiltonian vector field of the Casimir
function $I$ (\ref{j52}). By virtue of Proposition \ref{nc8}, its
flows are invariant submanifolds of the Kepler superintegrable
system on $U_-$. Let $\al$ be a parameter along the flow of this
vector field, i.e.,
\mar{zz121}\beq
\vt_I= \frac{\dr}{\dr \al}. \label{zz121}
\eeq
Then $U_-$ is provided with the generalized action-angle
coordinates $(I,x_1,\g,\al)$ such that the Poisson bivector
associated to the symplectic form $\Om_T$ on $U_-$ reads
\mar{j54}\beq
W= \frac{\dr}{\dr I}\w \frac{\dr}{\dr \al} + \frac{\dr}{\dr x_1}\w
\frac{\dr}{\dr \g}. \label{j54}
\eeq
Accordingly, Hamiltonian vector fields of integrals of motion
$F_i$ (\ref{zz102}) take the form
\be
&& \vt_1= \frac{\dr}{\dr \g}, \\
&& \vt_2= \frac1{4I^2}\left(-\frac1{2I}-x_1^2\right)^{-1/2}\sin\g
\frac{\dr}{\dr \al} - x_1
\left(-\frac1{2I}-x_1^2\right)^{-1/2}\sin\g
\frac{\dr}{\dr \g} - \\
&& \qquad \left(-\frac1{2I}-x_1^2\right)^{1/2}\cos\g
\frac{\dr}{\dr x_1}, \\
&& \vt_3= \frac1{4I^2}\left(-\frac1{2I}-x_1^2\right)^{-1/2}\cos\g
\frac{\dr}{\dr \al} - x_1
\left(-\frac1{2I}-x_1^2\right)^{-1/2}\cos\g
\frac{\dr}{\dr \g} + \\
&&\qquad \left(-\frac1{2I}-x_1^2\right)^{1/2}\sin\g \frac{\dr}{\dr x_1}.
\ee
A glance at these expressions shows that the vector fields $\vt_1$
and $\vt_2$ fail to be complete on $U_-$ (see Remark \ref{zz90}).

One can say something more about the angle coordinate $\al$. The
vector field $\vt_I$ (\ref{zz121}) reads
\be
\frac{\dr}{\dr\al}= \op\sum_i\left(\frac{\dr \cH}{\dr
p_i}\frac{\dr}{\dr q_i}-\frac{\dr \cH}{\dr q_i}\frac{\dr}{\dr
p_i}\right).
\ee
This equality leads to the relations
\be
\frac{\dr q_i}{\dr \al}=\frac{\dr \cH}{\dr p_i}, \qquad \frac{\dr
p_i}{\dr \al}=-\frac{\dr \cH}{\dr q_i},
\ee
which take the form of the Hamilton equation. Therefore, the
coordinate $\al$ is a cyclic time $\al=t\,{\rm mod}2\pi$ given by
the well-known expression
\be
&&\al=\f-a^{3/2}e\sin(a^{-3/2}\f),\qquad
r=a(1-e\cos(a^{-3/2}\f)),\\
&& a=-\frac1{2I}, \qquad e=(1+2IM^2)^{1/2}.
\ee

Now let us turn to the Kepler superintegrable system on $U_+$. It
is a globally superintegrable system with non-compact invariant
submanifolds as follows.

Let us put
\mar{zz102a}\ben
&& S_1=-K_1, \qquad S_2=-K_2, \qquad S_3=-M_{12}, \label{zz102a}\\
&& \{S_1,S_2\}=-S_3, \qquad \{S_2,S_3\}=S_1, \qquad \{S_3,S_1\}=S_2.
\nonumber
\een
We have a fibred manifold
\mar{zz103a}\beq
S: U_+\to N\subset\cG^*, \label{zz103a}
\eeq
which is the momentum mapping to the Lie coalgebra
$\cG^*=so(2,1)^*$, endowed with the coordinates $(x_i)$ such that
integrals of motion $S_i$ on $\cG^*$ read $S_i=x_i$. A base $N$ of
the fibred manifold (\ref{zz103a}) is an open submanifold of
$\cG^*$ given by the coordinate condition $x_3\neq 0$. It is a
union of two contractible components defined by the conditions
$x_3>0$ and $x_3<0$. The coinduced Lie -- Poisson structure on $N$
takes the form
\mar{j51a}\beq
w= x_2\dr^3\w\dr^1 - x_3\dr^1\w\dr^2 + x_1\dr^2\w\dr^3.
\label{j51a}
\eeq

The coadjoint action of $so(2,1)$ on $N$ reads
\be
\ve_1=-x_3\dr^2-x_2\dr^3, \qquad \ve_2=x_1\dr^3+x_3\dr^1, \qquad
\ve_3=x_2\dr^1-x_1\dr^2.
\ee
The orbits of this coadjoint action are given by the equation
\be
x_1^2 + x_2^2 - x_3^2={\rm const}.
\ee
They are the level surfaces of the Casimir function
\be
C=x_1^2 + x_2^2 - x_3^2
\ee
and, consequently, the Casimir function
\mar{zz120a}\beq
h=\frac12(x_1^2 + x_2^2 - x_3^2)^{-1}. \label{zz120a}
\eeq
A glance at the expression (\ref{zz101}) shows that the pull-back
$S^*h$ of this Casimir function (\ref{zz120a}) onto $U_+$ is the
Hamiltonian $\cH$ (\ref{zz41}) of the Kepler system on $U_+$.

As was mentioned above, the Hamiltonian vector field of $S^*h$ is
complete. Furthermore, it is known that invariant submanifolds of
the superintegrable Kepler system on $U_+$ are diffeomorphic to
$\mathbb R$. Therefore, the fibred manifold $S$ (\ref{zz103a}) is
a fibre bundle in accordance with Theorem \ref{11t2}. Moreover,
this fibre bundle is trivial because $N$ is a disjoint union of
two contractible manifolds. Consequently, it follows from Theorem
\ref{cmp36} that the Kepler superintegrable system on $U_+$ is
globally superintegrable, i.e., it admits global generalized
action-angle coordinates as follows.

The Poisson manifold $N$ (\ref{zz103a}) can be endowed with the
coordinates
\be
(I,x_1,\la), \qquad I>0, \qquad \la\neq 0,
\ee
defined by the equalities
\be
&& I=\frac12(x_1^2 + x_2^2 - x_3^2)^{-1}, \\
&& x_2=\left(\frac1{2I}-x_1^2\right)^{1/2}\cosh\la, \qquad
x_3=\left(\frac1{2I}-x_1^2\right)^{1/2}\sinh\la.
\ee
These coordinates are Darboux coordinates of the Lie -- Poisson
structure (\ref{j51a}) on $N$, namely,
\mar{j53a}\beq
w=\frac{\dr}{\dr \la}\w \frac{\dr}{\dr x_1}. \label{j53a}
\eeq

Let $\vt_I$ be the Hamiltonian vector field of the Casimir
function $I$ (\ref{j52}). By virtue of Proposition \ref{nc8}, its
flows are invariant submanifolds of the Kepler superintegrable
system on $U_+$. Let $\tau$ be a parameter along the flows of this
vector field, i.e.,
\mar{zz121a}\beq
\vt_I= \frac{\dr}{\dr \tau}. \label{zz121a}
\eeq
Then $U_+$ (\ref{zz103a}) is provided with the generalized
action-angle coordinates $(I,x_1,\la,\tau)$ such that the Poisson
bivector associated to the symplectic form $\Om_T$ on $U_+$ reads
\mar{j54a}\beq
W= \frac{\dr}{\dr I}\w \frac{\dr}{\dr \tau} + \frac{\dr}{\dr
\la}\w \frac{\dr}{\dr x_1}. \label{j54a}
\eeq
Accordingly, Hamiltonian vector fields of integrals of motion
$S_i$ (\ref{zz102a}) take the form
\be
&& \vt_1= -\frac{\dr}{\dr \la}, \\
&& \vt_2= \frac1{4I^2}\left(\frac1{2I}-x_1^2\right)^{-1/2}\cosh\la
\frac{\dr}{\dr \tau} + x_1
\left(\frac1{2I}-x_1^2\right)^{-1/2}\cosh\la
\frac{\dr}{\dr \la} + \\
&& \qquad \left(\frac1{2I}-x_1^2\right)^{1/2}\sinh\la
\frac{\dr}{\dr x_1}, \\
&& \vt_3= \frac1{4I^2}\left(\frac1{2I}-x_1^2\right)^{-1/2}\sinh\la
\frac{\dr}{\dr \tau} + x_1
\left(\frac1{2I}-x_1^2\right)^{-1/2}\sinh\la
\frac{\dr}{\dr \la} + \\
&&\qquad \left(\frac1{2I}-x_1^2\right)^{1/2}\cosh\la \frac{\dr}{\dr x_1}.
\ee

Similarly to the angle coordinate $\al$ (\ref{zz121}), the
generalized angle coordinate $\tau$ (\ref{zz121a}) obeys the
Hamilton equation
\be
\frac{\dr q_i}{\dr \tau}=\frac{\dr \cH}{\dr p_i}, \qquad \frac{\dr
p_i}{\dr \tau}=-\frac{\dr \cH}{\dr q_i}.
\ee
Therefore, it is the time $\tau=t$ given by the well-known
expression
\be
&& \tau=s-a^{3/2}e\sinh (a^{-3/2}s),\qquad r=a(e\cosh
(a^{-3/2}s)-1),\\
&& a=\frac1{2I}, \qquad e=(1+2IM^2)^{1/2}.
\ee

\section{Non-autonomous integrable systems}

The generalization of Liouville -- Arnold and Mishchenko --
Fomenko theorems to the case of non-compact invariant submanifolds
(Theorems \ref{nc0'} and \ref{cmp20}) enables one to analyze
completely integrable and superintegrable non-autonomous
Hamiltonian systems whose invariant submanifolds are necessarily
non-compact \cite{acang2,book10,ijgmmp12}.

A non-autonomous classical mechanics is described on a
configuration space $Q$ which is a fibre bundle $Q\to \mathbb R$
over the time axis $\mathbb R$. Its phase space is the vertical
cotangent bundle $V^*Q$ of $Q\to \mathbb R$ provided with the
canonical Poisson structure (\ref{m72}). However, non-autonomous
mechanics fails to be a familiar Poisson Hamiltonian system on
$V^*Q$. At the same time, it is equivalent to an autonomous
Hamiltonian system on the cotangent bundle $T^*Q$ provided with
the canonical symplectic form (\ref{m91'}).

This formulation of non-relativistic mechanics is similar to that
of classical field theory on fibre bundles over a base of
dimension $>1$ \cite{jpa99,book09,sard08}. A difference between
mechanics and field theory however lies in the fact that
connections on bundles over $\mathbb R$ are flat, and they fail to
be dynamic variables, but describe reference frames.

\subsection{Geometry of fibre bundle over $\mathbb R$}

This Section summarizes some peculiarities of fibre bundles over
$\mathbb R$.

Let
\mar{gm360}\beq
\pi:Q\to \mathbb R \label{gm360}
\eeq
be a fibred manifold whose base is treated as a time axis.
Throughout the Lectures, the time axis $\mathbb R$ is
parameterized by the Cartesian coordinate $t$ with the transition
functions $t'=t+$const. Relative to the Cartesian coordinate $t$,
the time axis $\mathbb R$ is provided with the  standard vector
field $\dr_t$ and the  standard one-form $dt$ which also is the
volume element on $\mathbb R$. The symbol $dt$ also stands for any
pull-back of the standard one-form $dt$ onto a fibre bundle over
$\mathbb R$.

\begin{remark} \mar{010} \label{010}
Point out one-to-one correspondence between the vector fields
$f\dr_t$, the densities $fdt$ and the real functions $f$ on
$\mathbb R$. Roughly speaking, we can neglect the contribution of
$T\mathbb R$ and $T^*\mathbb R$ to some expressions.
\end{remark}

In order that the dynamics of a mechanical system can be defined
at any instant $t\in\mathbb R$, we further assume that a fibred
manifold $Q\to \mathbb R$ is a fibre bundle with a typical fibre
$M$.

\begin{remark} \label{047} \mar{047}
In accordance with Remark \ref{Ehresmann}, a fibred manifold
$Q\to\mathbb R$ is a fibre bundle iff it admits an Ehresmann
connection $\G$, i.e., the horizontal lift $\G\dr_t$ onto $Q$ of
the standard vector field $\dr_t$ on $\mathbb R$ is complete. By
virtue of Theorem \ref{11t3}, any fibre bundle $Q\to \mathbb R$ is
trivial. Its different trivializations
\mar{gm219}\beq
\psi: Q=  \mathbb R\times M \label{gm219}
\eeq
differ from each other in fibrations $Q\to M$.
\end{remark}

Given bundle coordinates $(t,q^i)$ on the fibre bundle
$Q\to\mathbb R$ (\ref{gm360}), the first order jet manifold $J^1Q$
of $Q\to\mathbb R$  is provided with the adapted coordinates
$(t,q^i,q^i_t)$ possessing transition functions (\ref{50}) which
read
\be
q'^i_t=(\dr_t + q^j_t\dr_j)q'^i.
\ee

Note that, if $Q=\mathbb R\times M$ coordinated by $(t,\ol q^i)$,
there is the canonical isomorphism
\mar{gm220}\beq
J^1(\mathbb R\times M)=\mathbb R\times TM, \qquad \ol q^i_t=
\dot{\ol q}^i, \label{gm220}
\eeq
that one can justify by inspection of the transition functions of
the coordinates $\ol q^i_t$ and  $\dot{\ol q}^i$ when transition
functions of $q^i$ are time-independent. Due to the isomorphism
(\ref{gm220}), every trivialization (\ref{gm219}) yields the
corresponding trivialization of the jet manifold
\mar{jp2}\beq
J^1Q= \mathbb R\times TM. \label{jp2}
\eeq

The canonical imbedding (\ref{18}) of $J^1Q$ takes the form
\mar{z260}\ben
&& \la_{(1)}: J^1Q\ni (t,q^i,q^i_t)\to (t,q^i,\dot t=1, \dot q^i=q^i_t) \in TQ, \label{z260}\\
&& \la_{(1)}=d_t=\dr_t +q^i_t\dr_i, \nonumber
\een
where  by $d_t$  is meant the total derivative. From now on, a jet
manifold $J^1Q$ is identified with its image in $TQ$.

In view of the morphism $\la_{(1)}$ (\ref{z260}), any connection
\mar{z270}\beq
\G=dt\ot (\dr_t +\G^i\dr_i) \label{z270}
\eeq
on a fibre bundle $Q\to\mathbb R$  can be identified with a
nowhere vanishing   horizontal vector field
\mar{a1.10}\beq
\G = \dr_t + \G^i \dr_i \label{a1.10}
\eeq
on $Q$ which is the horizontal lift $\G\dr_t$ (\ref{b1.85}) of the
standard vector field $\dr_t$ on $\mathbb R$ by means of the
connection (\ref{z270}). Conversely, any vector field $\G$ on $Q$
such that $dt\rfloor\G =1$ defines a connection on $Q\to\mathbb
R$. Therefore, the connections (\ref{z270}) further are identified
with the vector fields (\ref{a1.10}). The integral curves of the
vector field (\ref{a1.10}) coincide with the integral sections for
the connection (\ref{z270}).

Connections on a fibre bundle $Q\to\mathbb R$ constitute an affine
space modelled over the vector space of vertical vector fields on
$Q\to\mathbb R$. Accordingly, the covariant differential
(\ref{2116}), associated with a connection $\G$ on $Q\to\mathbb
R$, takes its values into the vertical tangent bundle $VQ$ of
$Q\to\mathbb R$:
\mar{z279}\beq
D^\G: J^1Q\op\to_Q VQ, \qquad \dot q^i\circ D^\G =q^i_t-\G^i.
\label{z279}
\eeq
Its kernel, given by the coordinate equation
\mar{ggm155}\beq
q_t^i= \G^i(t,q^i), \label{ggm155}
\eeq
is a closed subbundle of the jet bundle $J^1Q\to\mathbb R$. This
is a first order dynamic differential equation on a fibre bundle
$Q\to \mathbb R$ \cite{book10,sard98}.

A connection $\G$ on a fibre bundle $Q\to\mathbb R$ is obviously
flat. It yields a horizontal distribution on $Q$. The integral
manifolds of this distribution are integral curves  of the vector
field (\ref{a1.10}) which are transversal to fibres of a fibre
bundle $Q\to\mathbb R$.

\begin{theorem} \label{gn1} \mar{gn1}
By virtue of Theorem \ref{gena113}, every connection $\G$ on a
fibre bundle $Q\to\mathbb R$ defines an atlas of local constant
trivializations of $Q\to\mathbb R$ such that the associated bundle
coordinates $(t,q^i)$ on $Q$ possess the transition functions
$q^i\to q'^i(q^j)$ independent of $t$, and
\mar{z271}\beq
\G=\dr_t \label{z271}
\eeq
with respect to these coordinates. Conversely, every atlas of
local constant trivializations of the fibre bundle $Q\to\mathbb R$
determines a connection on  $Q\to\mathbb R$ which is equal to
(\ref{z271}) relative to this atlas.
\end{theorem}

A connection $\G$ on a fibre bundle $Q\to \mathbb R$ is said to be
 complete  if the horizontal vector field (\ref{a1.10}) is
complete. In accordance with Remark \ref{Ehresmann}, a connection
on a fibre bundle $Q\to \mathbb R$ is complete iff it is an
Ehresmann connection. The following holds \cite{book98}.

\begin{theorem}\label{compl} \mar{compl}
Every trivialization of a fibre bundle $Q\to \mathbb R$ yields a
complete connection on this fibre bundle. Conversely, every
complete connection $\G$ on $Q\to\mathbb R$ defines its
trivialization (\ref{gm219}) such that the horizontal vector field
(\ref{a1.10}) equals $\dr_t$ relative to the bundle coordinates
associated with this trivialization.
\end{theorem}

\subsection{Non-autonomous Hamiltonian systems}

In non-autonomous mechanics on a configuration space $Q\to\mathbb
R$, the jet manifold $J^1Q$ plays a role of the velocity space. To
describe non-autonomous mechanics, let us restrict our
consideration to first order Lagrangian theory on a fibre bundle
$Q\to \mathbb R$ \cite{arxiv09,book10}. A  first order Lagrangian
is defined as a density
\mar{23f2}\beq
L=\cL dt, \qquad \cL: J^1Q\to \mathbb R, \label{23f2}
\eeq
on a velocity space $J^1Q$. The corresponding  second-order
Lagrange operator reads
\mar{305}\beq
\dl L= (\dr_i\cL- d_t\dr^t_i\cL) \thh^i\w dt. \label{305}
\eeq
Let us further use the  notation
\mar{03}\beq
\pi_i=\dr^t_i\cL, \qquad \pi_{ji}=\dr_j^t\dr_i^t\cL. \label{03}
\eeq
The kernel $\Ker\dl L\subset J^2Q$ of the Lagrange operator
defines the   second order Lagrange equation
\mar{b327}\beq
(\dr_i- d_t\dr^t_i)\cL=0. \label{b327}
\eeq
Its solutions  are (local) sections $c$ of the fibre bundle
$Q\to\mathbb R$ whose second order jet prolongations $\ddot c$
live in (\ref{b327}). They obey the equations
\beq
\dr_i\cL\circ \dot c- \frac{d}{dt}(\pi_i\circ \dot c)=0.
\label{z333}
\eeq

As was mentioned above, a  phase space of non-relativistic
mechanics on a configuration space $Q\to\mathbb R$ is the vertical
cotangent bundle
\be
V^*Q\ar^{\pi_\Pi} Q \ar^\pi \mathbb R,
\ee
of $Q\to\mathbb R$ equipped with the holonomic coordinates $(t,
q^i,p_i=\dot q_i)$ with respect to the fibre bases $\{\ol dq^i\}$
for the bundle $V^*Q\to Q$.

The cotangent bundle $T^*Q$ of the configuration space $Q$ is
endowed  with the holonomic coordinates $(t,q^i,p_0,p_i)$,
possessing the transition functions
\mar{2.3'}\beq
{p'}_i = \frac{\dr q^j}{\dr{q'}^i}p_j, \qquad p'_0=
\left(p_0+\frac{\dr q^j}{\dr t}p_j\right). \label{2.3'}
\eeq
It admits the Liouville form
\mar{N43}\beq
\Xi= p_0dt + p_i dq^i, \label{N43}
\eeq
the symplectic form
\mar{m91'}\beq
\Om_T=d\Xi=dp_0\w dt +dp_i\w dq^i, \label{m91'}
\eeq
and the corresponding Poisson bracket
\mar{m116}\beq
\{f,g\}_T =\dr^0f\dr_tg - \dr^0g\dr_tf +\dr^if\dr_ig-\dr^ig\dr_if,
\quad f,g\in C^\infty(T^*Q). \label{m116}
\eeq
Provided with the structures (\ref{m91'}) -- (\ref{m116}), the
cotangent bundle $T^*Q$ of $Q$ plays a role of the homogeneous
phase space of Hamiltonian non-relativistic mechanics.

There is the canonical one-dimensional affine bundle
\mar{z11'}\beq
\zeta:T^*Q\to V^*Q. \label{z11'}
\eeq
A glance at the transformation law (\ref{2.3'}) shows that it is a
trivial affine bundle. Indeed, given a global section $h$ of
$\zeta$, one can equip $T^*Q$ with the global fibre coordinate
\mar{09151}\beq
I_0=p_0-h, \qquad I_0\circ h=0, \label{09151}
\eeq
possessing the identity transition functions. With respect to the
coordinates
\mar{09150}\beq
(t,q^i,I_0,p_i), \qquad i=1,\ldots,m, \label{09150}
\eeq
the fibration (\ref{z11'}) reads
\mar{z11}\beq
\zeta: \mathbb R\times V^*Q \ni (t,q^i,I_0,p_i)\to (t,q^i,p_i)\in
V^*Q. \label{z11}
\eeq

Let us consider the subring of $C^\infty(T^*Q)$ which comprises
the pull-back $\zeta^*f$ onto $T^*Q$ of functions $f$ on the
vertical cotangent bundle $V^*Q$ by the fibration $\zeta$
(\ref{z11'}).  This subring is closed under the Poisson bracket
(\ref{m116}). Then by virtue of Theorem \ref{p11.3}, there exists
the degenerate coinduced Poisson structure
\mar{m72}\beq
\{f,g\}_V = \dr^if\dr_ig-\dr^ig\dr_if, \qquad f,g\in
C^\infty(V^*Q), \label{m72}
\eeq
on a phase space $V^*Q$ such that
\mar{m72'}\beq
\zeta^*\{f,g\}_V=\{\zeta^*f,\zeta^*g\}_T.\label{m72'}
\eeq
The holonomic coordinates on $V^*Q$ are canonical for the Poisson
structure (\ref{m72}).

With respect to the Poisson bracket (\ref{m72}), the Hamiltonian
vector fields of functions on $V^*Q$ read
\mar{m73,94}\ben
&& \vt_f = \dr^if\dr_i- \dr_if\dr^i, \qquad f\in C^\infty(V^*Q),
\label{m73}\\
&& [\vt_f,\vt_{f'}]=\vt_{\{f,f'\}_V}. \label{094}
\een
They are vertical vector fields on $V^*Q\to \mathbb R$.
Accordingly, the characteristic distribution of the Poisson
structure (\ref{m72}) is the vertical tangent bundle $VV^*Q\subset
TV^*Q$ of a fibre bundle $V^*Q\to \mathbb R$. The corresponding
symplectic foliation on the phase space $V^*Q$ coincides with the
fibration $V^*Q\to \mathbb R$.

It is readily observed that the ring $\cC(V^*Q)$ of Casimir
functions on a Poisson manifold $V^*Q$ consists of the pull-back
onto $V^*Q$ of functions on $\mathbb R$. Therefore, the Poisson
algebra $C^\infty(V^*Q)$ is a Lie $C^\infty(\mathbb R)$-algebra.

\begin{remark} \label{ws529} \mar{ws529}
The Poisson structure (\ref{m72}) can be introduced in a different
way \cite{book98,sard98}. Given  any section $h$ of the fibre
bundle (\ref{z11'}), let us consider the pull-back forms
\mar{z401}\ben
&& \bth=h^*(\Xi\w dt)=p_idq^i\w dt, \nonumber\\
&& \bom=h^*(d\Xi\w dt)=dp_i\w dq^i\w dt \label{z401}
\een
on $V^*Q$. They are independent of the choice of $h$. With $\bom$
(\ref{z401}), the Hamiltonian vector field $\vt_f$ (\ref{m73}) for
a function $f$ on $V^*Q$ is given by the relation
\be
\vt_f\rfloor\bom = -df\w dt,
\ee
while the Poisson bracket (\ref{m72}) is written as
\be
\{f,g\}_Vdt=\vt_g\rfloor\vt_f\rfloor\bom.
\ee
Moreover, one can show that a projectable vector field $\vt$ on
$V^*Q$ such that $\vt\rfloor dt=$const. is a canonical vector
field for the Poisson structure (\ref{m72}) iff
\mar{0100}\beq
\bL_\vt\bom=d(\vt\rfloor\bom)=0. \label{0100}
\eeq
\end{remark}

In contrast with autonomous Hamiltonian mechanics, the Poisson
structure (\ref{m72}) fails to provide any dynamic equation on a
fibre bundle $V^*Q\to\mathbb R$ because Hamiltonian vector fields
(\ref{m73}) of functions on $V^*Q$ are vertical vector fields, but
not connections on $V^*Q\to\mathbb R$. Hamiltonian dynamics on
$V^*Q$ is described as a particular Hamiltonian dynamics on fibre
bundles \cite{jpa99,book09,sard08}.

A  Hamiltonian on a phase space $V^*Q\to\mathbb R$ of
non-relativistic mechanics is defined as a global section
\mar{ws513}\beq
h:V^*Q\to T^*Q, \qquad p_0\circ h=\cH(t,q^j,p_j), \label{ws513}
\eeq
of the affine bundle $\zeta$ (\ref{z11'}). Given the Liouville
form $\Xi$ (\ref{N43}) on $T^*Q$, this section yields the
pull-back  Hamiltonian form
\mar{b4210}\beq
H=(-h)^*\Xi= p_k dq^k -\cH dt  \label{b4210}
\eeq
on $V^*Q$. This is the well-known  invariant of Poincar\'e --
Cartan \cite{arn}.

It should be emphasized that, in contrast with a Hamiltonian in
autonomous mechanics, the Hamiltonian $\cH$ (\ref{ws513}) is not a
function on $V^*Q$, but it obeys the transformation law
\mar{0144}\beq
\cH'(t,q'^i,p'_i)=\cH(t,q^i,p_i)+ p'_i\dr_t q'^i. \label{0144}
\eeq

\begin{remark} \label{ws512} \mar{ws512}
Any connection $\G$ (\ref{a1.10}) on a configuration bundle
$Q\to\mathbb R$ defines the global section $h_\G=p_i\G^i$
(\ref{ws513}) of the affine bundle $\zeta$ (\ref{z11'}) and the
corresponding Hamiltonian form
\mar{ws515}\beq
H_\G= p_k dq^k -\cH_\G dt= p_k dq^k -p_i\G^i dt. \label{ws515}
\eeq
Furthermore, given a connection $\G$, any Hamiltonian form
(\ref{b4210}) admits the splitting
\mar{m46'}\beq
H= H_\G -\cE_\G dt, \label{m46'}
\eeq
where
\mar{xx60}\beq
\cE_\G=\cH-\cH_\G=\cH- p_i\G^i \label{xx60}
\eeq
is a function on $V^*Q$. It is called the  Hamiltonian function
relative to a reference frame $\G$. With respect to the
coordinates adapted to a reference frame $\G$, we have
$\cE_\G=\cH$. Given different reference frames $\G$ and $\G'$, the
decomposition (\ref{m46'}) leads at once to the relation
\mar{0200}\beq
\cE_{\G'}=\cE_\G + \cH_\G -\cH_{\G'}=\cE_\G + (\G^i -\G'^i)p_i
\label{0200}
\eeq
between the Hamiltonian functions with respect to different
reference frames.
\end{remark}

Given a Hamiltonian form $H$ (\ref{b4210}), there exists a unique
horizontal vector field (\ref{a1.10}):
\be
\g_H=\dr_t -\g^i\dr_i -\g_i\dr^i,
\ee
on $V^*Q$ (i.e., a connection on $V^*Q\to \mathbb R$) such that
\mar{w255}\beq
\g_H\rfloor dH=0. \label{w255}
\eeq
This  vector field, called the  Hamilton vector field, reads
\mar{z3}\beq
\g_H=\dr_t + \dr^k\cH\dr_k- \dr_k\cH\dr^k. \label{z3}
\eeq
In a different way (Remark \ref{ws529}), the Hamilton vector field
$\g_H$ is defined by the relation
\be
\g_H\rfloor\bom=dH.
\ee
Consequently, it is canonical for the Poisson structure $\{,\}_V$
(\ref{m72}). This vector field  yields the first order dynamic
 Hamilton equation
\mar{z20a,b}\ben
&& q^k_t=\dr^k\cH, \label{z20a}\\
&&  p_{tk}=-\dr_k\cH \label{z20b}
\een
on $V^*Q\to\mathbb R$, where $(t,q^k,p_k,q^k_t,\dot p_{tk})$ are
the adapted coordinates on the first order jet manifold $J^1V^*Q$
of $V^*Q\to\mathbb R$.

Due to the canonical imbedding $J^1V^*Q\to TV^*Q$ (\ref{z260}),
the Hamilton equation (\ref{z20a}) -- (\ref{z20b}) is equivalent
to the autonomous first order dynamic equation
\mar{z20}\beq
\dot t=1, \qquad \dot q^i=\dr^i\cH, \qquad \dot p_i=-\dr_i\cH
\label{z20}
\eeq
on a manifold $V^*Q$ (Remark \ref{gena70}).

A  solution  of the Hamilton equation (\ref{z20a}) -- (\ref{z20b})
is an integral section $r$ for the connection $\g_H$.

We agree to call $(V^*Q,H)$ the  Hamiltonian system of $k=\di Q-1$
degrees of freedom.

In order to describe evolution of a Hamiltonian system at any
instant, the Hamilton vector field $\g_H$ (\ref{z3}) is assumed to
be complete, i.e., it is an Ehresmann connection (Remark
\ref{047}). In this case, the Hamilton equation (\ref{z20a}) --
(\ref{z20b}) admits a unique global solution through each point of
the phase space $V^*Q$. By virtue of Theorem \ref{compl}, there
exists a trivialization of a fibre bundle $V^*Q\to \mathbb R$ (not
necessarily compatible with its fibration $V^*Q\to Q$) such that
\mar{0102}\beq
\g_H=\dr_t, \qquad H=\ol p_id\ol q^i \label{0102}
\eeq
with respect to the associated coordinates $(t,\ol q^i, \ol p_i)$.
A direct computation shows that the Hamilton vector field $\g_H$
(\ref{z3}) satisfies the relation (\ref{0100}) and, consequently,
it is an infinitesimal generator of a one-parameter group of
automorphisms of the Poisson manifold $(V^*Q,\{,\}_V)$. Then one
can show that $(t,\ol q^i,\ol p_i)$ are canonical coordinates for
the Poisson manifold $(V^*Q,\{,\}_V)$ \cite{book98}, i.e.,
\be
w=\frac{\dr}{\dr \ol p_i}\w \frac{\dr}{\dr \ol q^i}.
\ee
Since $\cH=0$, the Hamilton equation (\ref{z20a}) -- (\ref{z20b})
in these coordinates takes the form
\be
\ol q^i_t=0, \qquad \ol p_{ti}=0,
\ee
i.e., $(t,\ol q^i,\ol p_i)$ are the  initial data coordinates.

As was mentioned above,  one can associate to any Hamiltonian
system on a phase space $V^*Q$ an equivalent autonomous symplectic
Hamiltonian system on the cotangent bundle $T^*Q$ (Theorem
\ref{09121}).

Given a Hamiltonian system $(V^*Q,H)$, its Hamiltonian $\cH$
(\ref{ws513}) defines the function
\mar{mm16}\beq
\cH^*=\dr_t\rfloor(\Xi-\zeta^* (-h)^*\Xi))=p_0+h=p_0+\cH
\label{mm16}
\eeq
on  $T^*Q$. Let us regard $\cH^*$ (\ref{mm16}) as a Hamiltonian of
an autonomous Hamiltonian system on the symplectic manifold
$(T^*Q,\Om_T)$. The corresponding autonomous Hamilton equation on
$T^*Q$ takes the form
\mar{z20'}\beq
\dot t=1, \qquad \dot p_0=-\dr_t\cH, \qquad \dot q^i=\dr^i\cH,
\qquad \dot p_i=-\dr_i\cH. \label{z20'}
\eeq

\begin{remark} \label{0170} \mar{0170}
Let us note that the splitting $\cH^*=p_0+\cH$ (\ref{mm16}) is ill
defined. At the same time, any reference frame $\G$ yields the
decomposition
\mar{j3}\beq
\cH^*=(p_0+\cH_\G) + (\cH-\cH_\G) = \cH^*_\G +\cE_\G, \label{j3}
\eeq
where $\cH_\G$ is the Hamiltonian (\ref{ws515}) and $\cE_\G$
(\ref{xx60}) is the Hamiltonian function relative to a reference
frame $\G$.
\end{remark}

The Hamiltonian vector field $\vt_{\cH^*}$ of $\cH^*$ (\ref{mm16})
on $T^*Q$ is
\mar{z5}\beq
\vt_{\cH^*}=\dr_t -\dr_t\cH\dr^0+ \dr^i\cH\dr_i- \dr_i\cH\dr^i.
\label{z5}
\eeq
Written relative to the coordinates (\ref{09150}), this vector
field reads
\mar{z5'}\beq
\vt_{\cH^*}=\dr_t + \dr^i\cH\dr_i- \dr_i\cH\dr^i. \label{z5'}
\eeq
It is identically projected onto the Hamilton vector field $\g_H$
(\ref{z3}) on $V^*Q$ such that
\mar{ws525}\beq
\zeta^*(\bL_{\g_H}f)=\{\cH^*,\zeta^*f\}_T, \qquad f\in
C^\infty(V^*Q). \label{ws525}
\eeq
Therefore, the Hamilton equation (\ref{z20a}) -- (\ref{z20b}) is
equivalent to the autonomous Hamilton equation (\ref{z20'}).

Obviously, the Hamiltonian vector field $\vt_{\cH^*}$ (\ref{z5'})
is complete if the Hamilton vector field $\g_H$ (\ref{z3}) is
complete.

Thus, the following has been proved \cite{dew,book10,mang00}.

\begin{theorem} \label{09121} \mar{09121} A Hamiltonian system $(V^*Q,H)$
of $k$ degrees of freedom is equivalent to an autonomous
Hamiltonian system $(T^*Q,\cH^*)$ of $k+1$ degrees of freedom on a
symplectic manifold $(T^*Q,\Om_T)$ whose Hamiltonian is the
function $\cH^*$ (\ref{mm16}).
\end{theorem}

We agree to call $(T^*Q,\cH^*)$ the  homogeneous Hamiltonian
system  and $\cH^*$ (\ref{mm16}) the  homogeneous Hamiltonian.

It is readily observed that the Hamiltonian form $H$ (\ref{b4210})
is the Poincar\'e -- Cartan form of the Lagrangian
\mar{Q33}\beq
L_H=h_0(H) = (p_iq^i_t - \cH)dt \label{Q33}
\eeq
on the jet manifold $J^1V^*Q$ of $V^*Q\to\mathbb R$
\cite{book09,jmp07,sard08}.

The Lagrange operator (\ref{305}) associated to the Lagrangian
$L_H$ reads
\mar{3.9}\beq
\cE_H=\dl L_H=[(q^i_t-\dr^i\cH) dp_i -(p_{ti}+\dr_i\cH) dq^i]\w
dt. \label{3.9}
\eeq
The corresponding Lagrange equation (\ref{b327}) is of first
order, and it coincides with the Hamilton equation (\ref{z20a}) --
(\ref{z20b}) on $J^1V^*Q$.

Due to this fact, the Lagrangian $L_H$ (\ref{Q33}) plays a
prominent role in Hamiltonian non-relativistic mechanics.

In particular, let
\be
u=u^t\dr_t + u^i\dr_i, \qquad u^t=0,1,
\ee
be a projectable vector field  on a configuration space $Q$. Its
functorial lift (\ref{l27'}) onto the cotangent bundle $T^*Q$ is
\mar{gm513}\beq
\wt u=u^t\dr_t + u^i\dr_i - p_j\dr_i u^j \dr^i. \label{gm513}
\eeq
This vector field is identically projected onto a vector field,
also given by the expression (\ref{gm513}), on the phase space
$V^*Q$ as a base of the trivial fibre bundle (\ref{z11'}). Then we
have the equality
\mar{mm24}\beq
\bL_{\wt u}H= \bL_{J^1\wt u}L_H= (-u^t\dr_t\cH+p_i\dr_tu^i
-u^i\dr_i\cH +p_i\dr_j u^i\dr^j\cH)dt. \label{mm24}
\eeq
This equality enables us to study conservation laws in Hamiltonian
mechanics similarly to those in Lagrangian mechanics.

Let an equation of motion of a mechanical system on a fibre bundle
$Y\to\mathbb R$ be described by an $r$-order differential equation
$\gE$ given by a closed subbundle of the jet bundle
$J^rY\to\mathbb R$ \cite{book09,sard08}.

\begin{definition} \label{026} \mar{026}
An integral of motion  of this mechanical system is defined as a
$(k<r)$-order differential operator $\Phi$ on $Y$ such that $\gE$
belongs to the kernel of an $r$-order jet prolongation of the
differential operator $d_t\Phi$, i.e.,
\mar{021}\beq
J^{r-k-1}(d_t\Phi)|_{\gE}=J^{r-k}\Phi|_{\gE}=0. \label{021}
\eeq
\end{definition}

It follows that an integral of motion $\Phi$ is constant on
solutions $s$ of a differential equation $\gE$, i.e., there is the
  differential
conservation law
\mar{020}\beq
(J^ks)^*\Phi={\rm const}., \qquad (J^{k+1}s)^*d_t\Phi=0.
\label{020}
\eeq

We agree to write the condition (\ref{021}) as the  weak equality
\mar{022}\beq
J^{r-k-1}(d_t\Phi)\ap 0, \label{022}
\eeq
which holds  on-shell, i.e.,  on solutions of a differential
equation $\gE$ by the formula (\ref{020}).

In non-relativistic mechanics, we can restrict our consideration
to integrals of motion $\Phi$ which are functions on $J^kY$. As
was mentioned above, equations of motion of non-relativistic
mechanics mainly are of first or second order. Accordingly, their
integrals of motion are functions on $Y$ or $J^kY$. In this case,
the corresponding weak equality (\ref{021}) takes the form
\mar{027}\beq
d_t\Phi\ap 0 \label{027}
\eeq
of a  weak conservation law  or,  simply, a conservation law.

Different integrals of motion need not be independent. Let
integrals of motion $\Phi_1,\ldots, \Phi_m$ of a mechanical system
on $Y$ be functions on $J^kY$. They are called  independent if
\mar{024}\beq
d\Phi_1\w\cdots\w d\Phi_m\neq 0 \label{024}
\eeq
everywhere on $J^kY$. In this case, any motion $J^ks$ of this
mechanical system lies in the common level surfaces of functions
$\Phi_1,\ldots, \Phi_m$ which bring $J^kY$ into a fibred manifold.

Integrals of motion can come from symmetries. This is the case of
Lagrangian and Hamiltonian mechanics.

\begin{definition} \label{025} \mar{025}
Let an equation of motion of a mechanical system be an $r$-order
differential equation $\gE\subset J^rY$. Its  infinitesimal
symmetry  (or, simply, a symmetry)  is defined as a vector field
on $J^rY$ whose restriction to $\gE$ is tangent to $\gE$.
\end{definition}

For instance, let us consider first order dynamic equations.

\begin{proposition} \label{080} \mar{080}
Let $\gE$ be the autonomous first order dynamic equation
(\ref{gm150}) given by a vector field $u$ on a manifold $Z$. A
vector field $\vt$ on $Z$ is its symmetry iff $[u,\vt]\ap 0$.
\end{proposition}

One can show that a smooth real function $F$ on a manifold $Z$ is
an integral of motion of the autonomous first order dynamic
equation (\ref{gm150}) (i.e., it is constant on solutions of this
equation) iff its Lie derivative along $u$ vanishes:
\mar{0115}\beq
\bL_u F=u^\la\dr_\la \Phi=0. \label{0115}
\eeq

\begin{proposition} \label{081} \mar{081}
Let $\gE$ be the first order dynamic equation (\ref{ggm155}) given
by a connection $\G$ (\ref{a1.10}) on a fibre bundle $Y\to\mathbb
R$. Then a vector field $\vt$ on $Y$ is its symmetry iff
$[\G,\vt]\ap 0$.
\end{proposition}

A smooth real function $\Phi$ on $Y$ is an integral of motion of
the first order dynamic equation (\ref{ggm155}) in accordance with
the equality (\ref{027}) iff
\mar{0116}\beq
\bL_\G \Phi=(\dr_t +\G^i\dr_i)\Phi=0. \label{0116}
\eeq

Following Definition \ref{025}, let us introduce the notion of a
symmetry of differential operators in the following relevant case.
Let us consider an $r$-order differential operator on a fibre
bundle $Y\to\mathbb R$ which is represented by an exterior form
$\cE$ on $J^rY$. Let its kernel $\Ker\cE$ be an $r$-order
differential equation on $Y\to\mathbb R$.

\begin{proposition} \label{082} \mar{082}
It is readily justified that a vector field $\vt$ on $J^rY$ is a
symmetry of the equation $\Ker\cE$ in accordance with Definition
\ref{025} iff
\mar{083}\beq
\bL_\vt \cE\ap 0. \label{083}
\eeq
\end{proposition}

Motivated by Proposition \ref{082}, we come to the following.

\begin{definition} \label{084} \mar{084} Let $\cE$ be the above
mentioned differential operator. A vector field $\vt$ on $J^rY$ is
called a symmetry  of a differential operator $\cE$ if the Lie
derivative $\bL_\vt \cE$ vanishes.
\end{definition}

By virtue of Proposition \ref{082}, a symmetry of a differential
operator $\cE$ also is a symmetry of the differential equation
$\Ker\cE$.

\subsection{Non-autonomous integrable systems}

Let us consider a non-autonomous mechanical system on a
configuration space $Q\to \mathbb R$ in Section 4.2. Its phase
space is the vertical cotangent bundle $V^*Q\to Q$ of $Q\to\mathbb
R$ endowed with the Poisson structure $\{,\}_V$ (\ref{m72}).A
Hamiltonian of a non-autonomous mechanical system is a section $h$
(\ref{ws513}) of the one-dimensional fibre bundle (\ref{z11'}) --
(\ref{z11}):
\mar{z11x}\beq
\zeta:T^*Q\to V^*Q, \label{z11x}
\eeq
where $T^*Q$ is the cotangent bundle of $Q$ endowed with the
canonical symplectic form $\Om_T$ (\ref{m91'}). The Hamiltonian
$h$ (\ref{ws513}) yields the pull-back Hamiltonian form $H$
(\ref{b4210}) on $V^*Q$ and defines the Hamilton vector field
$\g_H$ (\ref{z3}) on $V^*Q$. A smooth real function $F$ on $V^*Q$
is an integral of motion of a Hamiltonian system $(V^*Q,H)$ if its
Lie derivative $\bL_{\g_H} F$ vanishes.

\begin{definition} \label{i0x} \mar{i0x}
A non-autonomous Hamiltonian system $(V^*Q,H)$ of $n=\di Q-1$
degrees of freedom is called superintegrable if it admits $n\leq
k<2n$ integrals of motion $\F_1,\ldots,\F_k$, obeying the
following conditions.

(i) All the functions $\F_\al$ are independent, i.e., the $k$-form
$d\F_1\w\cdots\w d\F_k$ nowhere vanishes on $V^*Q$. It follows
that the map
\mar{nc4x}\beq
\F:V^*Q\to N=(\F_1(V^*Q),\ldots,\F_k(V^*Q))\subset \mathbb R^k
\label{nc4x}
\eeq
is a fibred manifold over a connected open subset $N\subset\mathbb
R^k$.

(ii) There exist smooth real functions $s_{\al\bt}$ on $N$ such
that
\mar{nc1z}\beq
\{\F_\al,\F_\bt\}_V= s_{\al\bt}\circ \F, \qquad \al,\bt=1,\ldots,
k. \label{nc1z}
\eeq

(iii) The matrix function with the entries $s_{\al\bt}$
(\ref{nc1z}) is of constant corank $m=2n-k$ at all points of $N$.
\end{definition}

In order to describe this non-autonomous superintegrable
Hamiltonian system, we use the fact that there exists an
equivalent autonomous Hamiltonian system $(T^*Q,\cH^*)$ of $n+1$
degrees of freedom on a symplectic manifold $(T^*Q,\Om_T)$ whose
Hamiltonian is the function $\cH^*$ (\ref{mm16}) (Theorem
\ref{09121}), and that this Hamiltonian system is superintegrable
(Theorem \ref{09141}). Our goal is the following.

\begin{theorem} \label{nc0'x} \mar{nc0'x}
Let Hamiltonian vector fields of the functions $\F_\al$ be
complete, and let fibres of the fibred manifold $\F$ (\ref{nc4x})
be connected and mutually diffeomorphic. Then there exists an open
neighborhood $U_M$ of a fibre $M$ of $\F$ (\ref{nc4x}) which is a
trivial principal bundle with the structure group
\mar{g120z}\beq
\mathbb R^{1+m-r}\times T^r \label{g120z}
\eeq
whose bundle coordinates are the generalized action-angle
coordinates
\mar{09135}\beq
(p_A,q^A,I_\la,t,y^\la), \qquad A=1,\ldots,k-n, \qquad
\la=1,\ldots, m,\label{09135}
\eeq
such that:

(i) $(t,y^\la)$ are coordinates on the toroidal cylinder
(\ref{g120z}),

(ii) the Poisson bracket $\{,\}_V$ on $U_M$ reads
\be
\{f,g\}_V = \dr^Af\dr_Ag-\dr^Ag\dr_Af + \dr^\la f\dr_\la g-\dr^\la
g\dr_\la f,
\ee

(iii) a Hamiltonian $\cH$ depends only on the action coordinates
$I_\la$,

(iv) the integrals of motion $\F_1, \ldots \F_k$ are independent
of coordinates $(t,y^\la)$.
\end{theorem}

Let us start with the case $k=n$  of a completely integrable
non-autonomous Hamiltonian system (Theorem \ref{z13}).

\begin{definition} \label{09122} \mar{09122}
A non-autonomous Hamiltonian system $(V^*Q,H)$ of $n$ degrees of
freedom is said to be completely integrable if it admits $n$
independent integrals of motion $F_1,\ldots,F_n$ which are in
involution with respect to the Poisson bracket $\{,\}_V$
(\ref{m72}).
\end{definition}

By virtue of the relations (\ref{094}), the vector fields
\mar{095}\beq
(\g_H,\vt_{F_1},\ldots,\vt_{F_n}), \qquad \vt_{F_\al} =
\dr^iF_\al\dr_i- \dr_iF_\al\dr^i, \label{095}
\eeq
mutually commute and, therefore, they span an $(n+1)$-dimensional
involutive distribution $\cV$ on $V^*Q$. Let $G$ be the group of
local diffeomorphisms of $V^*Q$ generated by the flows of vector
fields (\ref{095}). Maximal integral manifolds of $\cV$ are the
orbits of $G$ and invariant submanifolds of vector fields
(\ref{095}). They yield a foliation $\cF$ of $V^*Q$.

Let $(V^*Q,H)$ be a non-autonomous Hamiltonian system and
$(T^*Q,\cH^*)$ an equivalent autonomous Hamiltonian system on
$T^*Q$. An immediate consequence of the relations (\ref{m72'}) and
(\ref{ws525}) is the following.

\begin{theorem} \label{z6} \mar{z6}
Given a non-autonomous completely integrable Hamiltonian system
\mar{097'}\beq
(\g_H,F_1,\ldots,F_n) \label{097'}
\eeq
of $n$ degrees of freedom on $V^*Q$, the associated autonomous
Hamiltonian system
\mar{097}\beq
(\cH^*,\zeta^*F_1,\ldots,\zeta^*F_n) \label{097}
\eeq
of $n+1$ degrees of freedom on $T^*Q$ is completely integrable.
\end{theorem}

The Hamiltonian vector fields
\mar{099}\beq
(u_{\cH^*},u_{\zeta^*F_1},\ldots,u_{\zeta^*F_m}), \qquad
u_{\zeta^*F_\al} = \dr^iF_\al\dr_i- \dr_iF_\al\dr^i, \label{099}
\eeq
of the autonomous integrals of motion (\ref{097}) span an
$(n+1)$-dimensional involutive distribution $\cV_T$ on $T^*Q$ such
that
\mar{098}\beq
T\zeta(\cV_T)=\cV, \qquad Th(\cV)=\cV_T|_{h(V^*Q)=I_0=0},
\label{098}
\eeq
where
\be
&& Th: TV^*Q\ni(t,q^i,p_i,\dot t,\dot q^i,\dot
p_i)\to\\
&& \qquad (t,q^i,p_i,I_0=0,\dot t,\dot q^i,\dot p_i,\dot I_0=0)\in TT^*Q.
\ee
It follows that, if $M$ is an invariant submanifold of the
non-autonomous completely integrable Hamiltonian system
(\ref{097'}), then $h(M)$ is an invariant submanifold of the
autonomous completely integrable Hamiltonian system (\ref{097}).

In order do introduce generalized action-angle coordinates around
an invariant submanifold $M$ of the non-autonomous completely
integrable Hamiltonian system (\ref{097'}), let us suppose that
the vector fields (\ref{095}) on $M$ are complete. It follows that
$M$ is a locally affine manifold diffeomorphic to a toroidal
cylinder
\mar{0111}\beq
\mathbb R^{1+n-r}\times T^r. \label{0111}
\eeq
Moreover, let assume that there exists an open neighborhood $U_M$
of $M$ such that the foliation $\cF$ of $U_M$ is a fibred manifold
$\f: U_M\to N$ over a domain $N\subset \mathbb R^n$ whose fibres
are mutually diffeomorphic.

Because the morphism $Th$ (\ref{098}) is a bundle isomorphism, the
Hamiltonian vector fields (\ref{099}) on the invariant submanifold
$h(M)$ of the autonomous completely integrable Hamiltonian system
are complete. Since the affine bundle $\zeta$ (\ref{z11x}) is
trivial, the open neighborhood $\zeta^{-1}(U_M)$ of the invariant
submanifold $h(M)$ is a fibred manifold
\be
\wt\f: \zeta^{-1}(U_M)= \mathbb R\times U_M \ar^{(\id\mathbb
R,\f)} \mathbb R\times N = N'
\ee
over a domain $N'\subset \mathbb R^{n+1}$ whose fibres are
diffeomorphic to the toroidal cylinder (\ref{0111}). In accordance
with Theorem \ref{cmp20}, the open neighborhood $\zeta^{-1}(U_M)$
of $h(M)$ is a trivial principal bundle
\mar{0910}\beq
\zeta^{-1}(U_M)=N'\times (\mathbb R^{1+n-r}\times T^r)\to N'
\label{0910}
\eeq
with the structure group (\ref{0111}) whose bundle coordinates are
the generalized action-angle coordinates
\mar{0911}\beq
(I_0,I_1,\ldots,I_n, t,z^1,\ldots,z^n) \label{0911}
\eeq
such that:

(i) $(t,z^a)$ are coordinates on the toroidal cylinder
(\ref{0111}),

(ii) the symplectic form $\Om_T$ on $\zeta^{-1}(U)$ reads
\be
\Om_T=dI_0\w dt + dI_a\w dz^a,
\ee

(iii) $\cH^*=I_0$,

(iv) the integrals of motion $\zeta^*F_1,\ldots,\zeta^*F_n$ depend
only on the action coordinates $I_1,\ldots,I_n$.

Provided with the coordinates (\ref{0911}),
\be
\zeta^{-1}(U_M)= U_M\times\mathbb R
\ee
is a trivial bundle possessing the fibre coordinate $I_0$
(\ref{09151}). Consequently, the non-autonomous open neighborhood
$U_M$ of an invariant submanifold $M$ of the completely integrable
Hamiltonian system (\ref{095}) is diffeomorphic to the Poisson
annulus
\mar{0915}\beq
U_M=N\times (\mathbb R^{1+n-r}\times T^r) \label{0915}
\eeq
endowed with the generalized action-angle coordinates
\mar{0916}\beq
(I_1,\ldots,I_n, t,z^1,\ldots,z^n) \label{0916}
\eeq
such that:

(i) the Poisson structure (\ref{m72}) on $U_M$ takes the form
\be
\{f,g\}_V = \dr^af\dr_ag-\dr^ag\dr_af,
\ee

(ii) the Hamiltonian (\ref{ws513}) reads $\cH=0$,

(iii) the integrals of motion $F_1,\ldots,F_n$ depend only on the
action coordinates $I_1,\ldots,I_n$.

The Hamilton equation (\ref{z20a}) -- (\ref{z20b}) relative to the
generalized action-angle coordinates (\ref{0916}) takes the form
\be
 z^a_t=0, \qquad I_{ta}=0.
\ee
It follows that the generalized action-angle coordinates
(\ref{0916}) are the initial date coordinates.

Note that the generalized action-angle coordinates (\ref{0916}) by
no means are unique. Given a smooth function $\cH'$ on $\mathbb
R^n$, one can provide $\zeta^{-1}(U_M)$ with the generalized
action-angle coordinates
\mar{ww26}\beq
t, \qquad z'^a=z^a- t\dr^a\cH', \qquad I'_0=I_0+\cH'(I_b), \qquad
I'_a=I_a. \label{ww26}
\eeq
With respect to these coordinates, a Hamiltonian of the autonomous
Hamiltonian system on $\zeta^{-1}(U_M)$ reads $\cH'^*=I'_0-\cH'$.
A Hamiltonian of the non-autonomous Hamiltonian system on $U$
endowed with the generalized action-angle coordinates $(I_a,t,
z'^a)$ is $\cH'$.

Thus, the following has been proved.

\begin{theorem} \label{z13} \mar{z13}
Let $(\g_H,F_1,\ldots,F_n)$ be a non-autonomous completely
integrable Hamiltonian system. Let $M$ be its invariant
submanifold such that the vector fields (\ref{095}) on $M$ are
complete and that there exists an open neighborhood $U_M$ of $M$
which is a fibred manifold in mutually diffeomorphic invariant
submanifolds. Then $U_M$ is diffeomorphic to the Poisson annulus
(\ref{0915}), and it can be provided with the generalized
action-angle coordinates (\ref{0916}) such that the integrals of
motion $(F_1,\ldots,F_n)$ and the Hamiltonian $\cH$ depend only on
the action coordinates $I_1,\ldots,I_n$.
\end{theorem}

Let now $(\g_H,\F_1,\ldots,\F_k)$ be a non-autonomous
superintegrable Hamiltonian system in accordance with Definition
\ref{i0x}. The associated autonomous Hamiltonian system on $T^*Q$
possesses $k+1$ integrals of motion
\mar{09136}\beq
(\cH^*,\zeta^*\F_1,\ldots,\zeta^*\F_k) \label{09136}
\eeq
with the following properties.

(i) The functions (\ref{09136}) are mutually independent, and the
map
\mar{09140}\ben
&& \wt\F:T^*Q\to
(\cH^*(T^*Q),\zeta^*\F_1(T^*Q),\ldots,\zeta^*\F_k(T^*Q))= \label{09140}\\
&& \qquad (I_0,\F_1(V^*Q),\ldots,\F_k(V^*Q))= \mathbb R\times
N=N'\nonumber
\een
is a fibred manifold.

(ii) The functions (\ref{09136}) obey the relations
\be
\{\zeta^*\F_\al,\zeta^*\F_\bt\}= s_{\al\bt}\circ \zeta^*\F,\qquad
\{\cH^*,\zeta^*\F_\al\}=s_{0\al}=0
\ee
so that the matrix function with the entries
$(s_{0\al},s_{\al\bt})$ on $N'$ is of constant corank $2n+1-k$.

Refereing to Definition \ref{i0} of an autonomous superintegrable
system, we come to the following.

\begin{theorem} \label{09141} \mar{09141} Given a
non-autonomous superintegrable Hamiltonian system $(\g_H,\F_\al)$
on $V^*Q$, the associated autonomous Hamiltonian system
(\ref{09136}) on $T^*Q$ is superintegrable.
\end{theorem}

There is the commutative diagram
\be
\begin{array}{rcccl}
&T^*Q &\ar^\zeta & V^*Q&\\
_{\wt \F}& \put(0,10){\vector(0,-1){20}} & & \put(0,10){\vector(0,-1){20}}&_\F\\
& N' &\ar^\xi & N&
\end{array}
\ee
where $\zeta$ (\ref{z11x}) and
\be
\xi:N'=\mathbb R\times N\to N
\ee
are trivial bundles. It follows that the fibred manifold
(\ref{09140}) is the pull-back $\wt \F=\xi^* \F$ of the fibred
manifold $\F$ (\ref{nc4x}) onto $N'$.

Let the conditions of Theorem \ref{nc0'} hold. If the Hamiltonian
vector fields
\be
(\g_H,\vt_{\F_1},\ldots,\vt_{\F_k}),\qquad \vt_{\F_\al}=
\dr^i\F_\al\dr_i- \dr_i\F_\al\dr^i,
\ee
of integrals of motion $\F_\al$ on $V^*Q$ are complete, the
Hamiltonian vector fields
\be
(u_{\cH^*},u_{\zeta^*\F_1},\ldots,u_{\zeta^*\F_k}), \qquad
u_{\zeta^*\F_\al} = \dr^i\F_\al\dr_i- \dr_i\F_\al\dr^i,
\ee
on $T^*Q$ are complete. If fibres of the fibred manifold $\F$
(\ref{nc4x}) are connected and mutually diffeomorphic, the fibres
of the fibred manifold $\wt\F$ (\ref{09140}) also are well.

Let $M$ be a fibre of $\F$ (\ref{nc4x}) and $h(M)$ the
corresponding fibre of $\wt\F$ (\ref{09140}). In accordance
Theorem \ref{nc0'}, there exists an open neighborhood $U'$ of
$h(M)$ which is a trivial principal bundle with the structure
group (\ref{g120z}) whose bundle coordinates are the generalized
action-angle coordinates
\mar{09135'}\beq
(I_0,I_\la,t,y^\la,p_A,q^A), \qquad A=1,\ldots,n-m, \qquad
\la=1,\ldots, k,\label{09135'}
\eeq
such that:

(i) $(t,y^\la)$ are coordinates on the toroidal cylinder
(\ref{g120z}),

(ii) the symplectic form $\Om_T$ on $U'$ reads
\be
\Om_T= dI_0\w dt + dI_\al\w dy^\al + dp_A\w dq^A,
\ee

(iii) the action coordinates $(I_0,I_\al)$ are expressed in the
values of the Casimir functions $C_0=I_0$, $C_\al$ of the
coinduced Poisson structure
\be
w=\dr^A\w\dr_A
\ee
on $N'$,

(iv) a homogeneous Hamiltonian $\cH^*$ depends on the action
coordinates, namely, $\cH^*=I_0$,

(iv) the integrals of motion $\zeta^*\F_1, \ldots \zeta^*\F_k$ are
independent of the coordinates $(t,y^\la)$.

Provided with the generalized action-angle coordinates
(\ref{09135'}), the above mentioned neighborhood $U'$ is a trivial
bundle $U'=\mathbb R\times U_M$ where $U_M=\zeta(U')$ is an open
neighborhood of the fibre $M$ of the fibre bundle $\F$
(\ref{nc4x}). As a result, we come to Theorem \ref{nc0'x}.

\section{Quantum superintegrable systems}

To quantize classical Hamiltonian systems, one usually follows
canonical quantization which replaces the Poisson bracket
$\{f,f'\}$ of smooth functions with the bracket $[\wh f,\wh f']$
of Hermitian operators in a Hilbert space such that Dirac's
condition
\mar{qm514}\beq
[\wh f,\wh f']=-i\wh{\{f,f'\}} \label{qm514}
\eeq
holds. Canonical quantization of Hamiltonian non-relativistic
mechanics on a configuration space $Q\to\mathbb R$ is geometric
quantization \cite{jmp02a,book05,book10}. In the case of
integrable Hamiltonian systems, there is a reason that, since a
Hamiltonian of an integrable system depends only on action
variables (Proposition \ref{zz1}), it seems natural to provide the
Schr\"odinger representation of action variables by first order
differential operators on functions of angle coordinates

For the sake of simplicity, symplectic and Poisson manifolds
throughout this Section are assumed to be simple connected (see
Remark \ref{kk105}). Geometric quantization of toroidal cylinders
possessing a non-trivial first homotopy group is considered in
Section 6.4.

\subsection{Geometric quantization of symplectic manifolds}

We start with the basic geometric quantization of symplectic
manifolds \cite{eche98,book05,book10,sni}. It falls into the
following three steps: prequantization, polarization and
metaplectic correction.

Let $(Z,\Om)$ be a $2m$-dimensional simply connected symplectic
manifold. Let $C\to Z$ be  a complex line bundle whose typical
fibre is $\mathbb C$. It is coordinated by $(z^\la,c)$ where $c$
is a complex coordinate.

\begin{proposition} \label{0220} \mar{0220}
By virtue of the well-known theorems \cite{hir,book00}, the
structure group of a complex line bundle $C\to Z$ is reducible to
$U(1)$ such that:

$\bullet$ given a bundle atlas of $C\to Z$ with $U(1)$-valued
transition functions and associated bundle coordinates
$(z^\la,c)$, there exists a Hermitian fibre metric
\mar{qm513}\beq
g(c,c)=c\ol c \label{qm513}
\eeq
in $C$;

$\bullet$ for any Hermitian fibre metric $g$ in $C\to Z$, there
exists a bundle atlas of $C\to Z$ with $U(1)$-valued transition
functions such that $g$ takes the form (\ref{qm513}) with respect
to the associated bundle coordinates.
\end{proposition}

Let $\cK$ be a linear connection on a fibre bundles $C\to Z$. It
reads
\mar{qm508}\beq
\cK=dz^\la\ot(\dr_\la +\cK_\la c\dr_c), \label{qm508}
\eeq
where $\cK_\la$ are local complex functions on $Z$. The
corresponding covariant differential $D^\cK$ (\ref{2116}) takes
the form
\mar{0214}\beq
D^\cK=(c_\la -\cK_\la c)dz^\la\ot\dr_c. \label{0214}
\eeq
The curvature two-form (\ref{mos4}) of the connection $K$
(\ref{qm508}) reads
\mar{0221}\beq
R=\frac12(\dr_\nu \cK_\m-\dr_\m \cK_\nu)cdz^\nu\w dz^\m\ot\dr_c.
\label{0221}
\eeq

\begin{proposition} \label{0216} \mar{0216}
A connection $A$ on a complex line bundle $C\to Z$ is a
$U(1)$-principal connection iff there exists an $A$-invariant
Hermitian fibre metric $g$ in $C$, i.e.,
\be
d_H(g(c,c))=g(D^Ac,c) + g(c,D^Ac).
\ee
With respect to the bundle coordinates $(z^\la,c)$ in Proposition
\ref{0220}, this connection reads
\mar{qm508a}\beq
A=dz^\la\ot(\dr_\la +iA_\la c\dr_c), \label{qm508a}
\eeq
where $A_\la$ are local real functions on $Z$.
\end{proposition}

The curvature $R$ (\ref{0221}) of the $U(1)$-principal connection
$A$ (\ref{qm508a}) defines the  first Chern characteristic form
\mar{0222,3}\ben
&& c_1(A)=-\frac{1}{4\pi}(\dr_\nu A_\m-\dr_\m A_\nu)cdz^\nu\w
dz^\m, \label{0222}\\
&& R=-2\pi i c_1\ot u_C, \label{0223}
\een
where
\mar{0215}\beq
u_C=c\dr_c \label{0215}
\eeq
is the Liouville vector field (\ref{z112'}) on $C$. The Chern form
(\ref{0222}) is closed, but it need not be exact because $A_\m
dz^\m$ is not a one-form on $Z$ in general.

\begin{definition} \label{+308} \mar{+308}
A complex line bundle $C\to Z$ over a symplectic manifold
$(Z,\Om)$ is called a prequantization bundle if a form
$(2\pi)^{-1}\Om$ on $Z$ belongs to the first Chern characteristic
class of $C$.
\end{definition}

A prequantization bundle, by definition, admits a $U(1)$-principal
connection $A$, called an admissible connection, whose curvature
$R$ (\ref{0221}) obeys the relation
\mar{gm503}\beq
R=- i\Om\ot u_C,  \label{qm503}
\eeq
called the  admissible condition.

\begin{remark} \label{kk105} \mar{kk105}
Let $A$ be the admissible connection (\ref{qm508a}) and $B=B_\m
dz^\m$ a closed one-form on $Z$. Then
\mar{kk106}\beq
A'=A+ icB\ot\dr_c \label{kk106}
\eeq
also is an admissible connection. Since a manifold $Z$ is assumed
to be simply connected, a closed one-form $B$ is exact. In this
case, connections $A$ and $A'$ (\ref{kk106}) are gauge conjugate
\cite{book00}.
\end{remark}

Given an admissible connection $A$, one can assign to each
function $f\in C^\infty(Z)$ the $C$-valued first order
differential operator $\wh f$ on a fibre bundle $C\to Z$ in
accordance with  Kostant -- Souriau formula
\mar{qm506}\beq
\wh f= -i\vt_f\rfloor D^A -fu_C=-[i\vt_f^\la(c_\la -iA_\la c) +
fc]\dr_c, \label{qm506}
\eeq
where $D^A$ is the covariant differential (\ref{0214}) and $\vt_f$
is the Hamiltonian vector field of $f$. It is easily justified
that the operators (\ref{qm506}) obey Dirac's condition
(\ref{qm514}) for all elements $f$ of the Poisson algebra
$C^\infty(Z)$.

The Kostant -- Souriau formula (\ref{qm506}) is called the
prequantization because, in order to obtain Hermitian operators
$\wh f$ (\ref{qm506}) acting on a Hilbert space, one should
restrict both a class of functions $f\in C^\infty(Z)$ and a class
of sections of $C\to Z$ in consideration as follows.

Given a symplectic manifold $(Z,\Om)$, by its  polarization is
meant a maximal involutive distribution ${\bf T}\subset TZ$ such
that
\be
\Om(\vt,\up)=0, \qquad \vt,\up\in\bT_z, \qquad z\in Z.
\ee
This term also stands for the algebra $\cT_\Om$ of sections of the
distribution ${\bf T}$. We denote by $\cA_\cT$ the subalgebra of
the Poisson algebra $C^\infty(Z)$ which consists of the functions
$f$ such that
\be
[\vt_f,\cT_\Om]\subset \cT_\Om.
\ee
It is called the  quantum algebra of a symplectic manifold
$(Z,\Om)$. Elements of this algebra only are quantized.

In order to obtain the carrier space of the algebra $\cA_\cT$, let
us assume that $Z$ is oriented and that its cohomology
$H^2(Z;\mathbb Z_2)$ with coefficients in the constant sheaf
$\mathbb Z_2$ vanishes. In this case, one can consider the
metalinear complex line bundle $\cD_{1/2}[Z]\to Z$ characterized
by a bundle atlas $\{(U;z^\la,r)\}$ with the transition functions
\mar{0232}\beq
r'=Jr, \qquad J\ol J=\left|\det\left(\frac{\dr z^\m}{\dr
z'^\nu}\right)\right|. \label{0232}
\eeq
Global sections $\rho$ of this bundle are called the
half-densities on $Z$ \cite{eche98,book10}. Note that the
metalinear bundle $\cD_{1/2}[Z]\to Z$ admits the canonical lift of
any vector field $u$ on $Z$ such that the corresponding Lie
derivative of its sections reads
\mar{0231}\beq
\bL_u=u^\la\dr_\la+\frac12\dr_\la u^\la. \label{0231}
\eeq

Given an admissible connection $A$, the prequantization formula
(\ref{qm506}) is extended to sections $s\ot\rho$ of the fibre
bundle
\mar{0260}\beq
C\op\ot_Z\cD_{1/2}[Z]\to Z \label{0260}
\eeq
as follows:
\mar{qm515}\ben
&& \wh f(s\ot\rho)= (-i\nabla_{\vt_f}- f)(s\ot\rho)=
(\wh fs)\ot\rho +s\ot\bL_{\vt_f}\rho, \label{qm515}\\
&& \nabla_{\vt_f}(s\ot\rho)=(\nabla_{\vt_f}^As)\ot\rho
+s\ot\bL_{\vt_f}\rho, \nonumber
\een
where $\bL_{\vt_f}\rho$ is the Lie derivative (\ref{0231}) acting
on half-densities. This extension is said to be the metaplectic
correction, and the tensor product (\ref{0260}) is called the
quantization bundle. One can think of its sections $\vr$ as being
$C$-valued  half-forms. It is readily observed that the operators
(\ref{qm515}) obey Dirac's condition (\ref{qm514}). Let us denote
by $\gE_Z$  a complex vector space of sections $\vr$ of the fibre
bundle $C\ot\cD_{1/2}[Z]\to Z$ of compact support such that
\mar{kk120}\ben
&&\nabla_\up \vr=0, \qquad  \up\in \cT_\Om, \label{kk120}\\
&& \nabla_\up \vr=\nabla_\up(s\ot\rho)= (\nabla_\up^As)\ot\rho
+s\ot\bL_\up\rho. \nonumber
\een

\begin{lemma} \label{+310} \mar{+310}
For any function $f\in\cA_\cT$ and an arbitrary section $\vr\in
\gE_Z$, the relation $\wh f\vr\in\gE_Z$ holds.
\end{lemma}

Thus, we have a representation of the quantum algebra $\cA_\cT$ in
the space $\gE_Z$. Therefore, by quantization of a function
$f\in\cA_\cT$ is meant the restriction of the operator $\wh f$
(\ref{qm515}) to $\gE_Z$.

Let $g$ be an $A$-invariant Hermitian fibre metric in $C\to Z$ in
accordance with Proposition \ref{0216}. If $\gE_Z\neq 0$, the
Hermitian form
\mar{qm522}\beq
\lng s_1\ot\rho_1| s_2\ot\rho_2\rng= \op\int_Z
g(s_1,s_2)\rho_1\ol\rho_2 \label{qm522}
\eeq
brings $\gE_Z$ into a pre-Hilbert space. Its completion $\ol\gE_Z$
is called a  quantum Hilbert space,  and the operators $\wh f$
(\ref{qm515}) in this Hilbert space are Hermitian.

In particular, let us consider the standard geometric quantization
of a cotangent bundle \cite{eche98,book05,book10,sni}.

Let $M$ be an $m$-dimensional simply connected smooth manifold
coordinated by $(q^i)$. Its cotangent bundle $T^*M$ is simply
connected. It is provided with the canonical symplectic form
$\Om_T$ (\ref{m83}) written with respect to holonomic coordinates
$(q^i,p_i=\dot q_i)$ on $T^*M$. Let us consider the trivial
complex line bundle
\mar{qm501}\beq
C=T^*M\times\mathbb C\to T^*M. \label{qm501}
\eeq
The canonical symplectic form (\ref{m83}) on $T^*M$ is exact,
i.e., it has the same zero de Rham cohomology class as the first
Chern class of the trivial $U(1)$-bundle $C$ (\ref{qm501}).
Therefore, $C$ is a prequantization bundle in accordance with
Definition \ref{+308}.

Coordinated by $(q^i,p_i,c)$, this bundle is provided with the
admissible connection (\ref{qm508a}):
\mar{qm502}\beq
A=dp_j\ot\dr^j +dq^j\ot(\dr_j- ip_jc\dr_c) \label{qm502}
\eeq
such that the condition (\ref{qm503}) is satisfied. The
corresponding $A$-invariant fibre metric in $C$ is given by the
expression (\ref{qm513}). The covariant derivative of sections $s$
of the prequantization bundle $C$ (\ref{qm501}) relative to the
connection $A$ (\ref{qm502}) along the vector field $u=u^j\dr_j
+u_j\dr^j$ on $T^*M$ reads
\mar{qq12}\beq
\nabla_u s= u^j(\dr_j + i p_j)s +u_j\dr^j s. \label{qq12}
\eeq
Given a function $f\in C^\infty(T^*M)$ and its Hamiltonian vector
field
\be
\vt_f=\dr^i f\dr_i -\dr_i f\dr^i,
\ee
the covariant derivative (\ref{qq12}) along $\vt_f$ is
\be
\nabla_{\vt_f}s= \dr^i f(\dr_i +  ip_i)s - \dr_i f\dr^i s.
\ee
With the connection $A$ (\ref{qm502}), the prequantization
(\ref{qm506}) of elements $f$ of the Poisson algebra
$C^\infty(T^*M)$ takes the form
\mar{qm504}\beq
\wh f=-i\dr^jf(\dr_j+ ip_j) +i\dr_jf\dr^j -  f. \label{qm504}
\eeq

Let us note that, since the complex line bundle (\ref{qm501}) is
trivial, its sections are simply smooth complex functions on
$T^*M$. Then the prequantum operators (\ref{qm504}) can be written
in the form
\mar{xjj}\beq
\wh f=-i\bL_{\vt_f} +(\bL_\up f-f), \label{xjj}
\eeq
where $\up=p_j\dr^j$ is the Liouville vector field (\ref{z112'})
on $T^*M\to M$.

It is readily observed that the vertical tangent bundle $VT^*M$ of
the cotangent bundle $T^*M\to M$ provides a polarization of
$T^*M$. Certainly, it is not a unique polarization of $T^*M$. We
call $VT^*M$ the vertical polarization. The corresponding quantum
algebra $\cA_T\subset C^\infty(T^*M)$ consists of affine functions
of momenta
\mar{xci13'}\beq
f=a^i(q^j)p_i +b(q^j) \label{xci13'}
\eeq
on $T^*M$. Their Hamiltonian vector fields read
\mar{kk132}\beq
\vt_f=a^i\dr_i-(p_j\dr_i a^j +\dr_i b)\dr^i. \label{kk132}
\eeq
We call $\cA_T$ the quantum algebra of a cotangent bundle.

Since the Jacobain of holonomic coordinate transformations of the
cotangent bundle $T^*M$ equals 1, the geometric quantization of
$T^*M$ need no metaplectic correction. Consequently, the quantum
algebra $\cA_T$ of the affine functions (\ref{xci13'}) acts on the
subspace $\gE_{T^*M}\subset C^\infty(T^*M)$ of complex functions
of compact support on $T^*M$ which obey the condition
(\ref{kk120}):
\be
\nabla_\up s=\up_i\dr^i s=0, \qquad \cT_\Om\ni \up=\up_i\dr^i.
\ee
A glance at this equality shows that elements of $\gE_{T^*M}$ are
independent of momenta $p_i$, i.e., they are the pull-back of
complex functions on $M$ with respect to the fibration $T^*M\to
M$. These functions fail to be of compact support, unless $s=0$.
Consequently, the carrier space $\gE_{T^*M}$ of the quantum
algebra $\cA_T$ is reduced to zero. One can overcome this
difficulty as follows.

Given the canonical zero section $\wh 0(M)$ of the cotangent
bundle $T^*M \to M$, let
\mar{kk130}\beq
C_M=\wh 0(M)^*C \label{kk130}
\eeq
be the pull-back of the complex line bundle $C$ (\ref{qm501}) over
$M$. It is a trivial complex line bundle $C_M=M\times\mathbb C$
provided with the pull-back Hermitian fibre metric $g(c,c')=c\ol
c'$ and the pull-back (\ref{mos82}):
\be
A_M =\wh 0(M)^*A= dq^j\ot(\dr_j - ip_jc\dr_c)
\ee
of the connection $A$ (\ref{qm502}) on $C$. Sections of $C_M$ are
smooth complex functions on $M$. One can consider a representation
of the quantum algebra $\cA_T$ of the affine functions
(\ref{xci13'}) in the space of complex functions on $M$ by the
prequantum operators (\ref{qm504}):
\be
\wh f=-ia^j\dr_j - b.
\ee
However, this representation need a metaplectic correction.

Let us assume that $M$ is oriented and that its cohomology
$H^2(M;\mathbb Z_2)$ with coefficients in the constant sheaf
$\mathbb Z_2$ vanishes. Let $\cD_{1/2}[M]$ be the metalinear
complex line over $M$. Since the complex line bundle $C_M$
(\ref{kk130}) is trivial, the quantization bundle (\ref{0260}):
\mar{kk131}\beq
C_M\op\ot_M \cD_{1/2}[M]\to M \label{kk131}
\eeq
is isomorphic to $\cD_{1/2}[M]$.

Because the Hamiltonian vector fields (\ref{kk132}) of functions
$f$ (\ref{xci13'}) project onto vector fields $a^j\dr_j$ on $M$
and $\bL_\up f-f=-b$ in the formula (\ref{xjj}) is a function on
$M$, one can assign to each element $f$ of the quantum algebra
$\cA_T$ the following first order differential operator in the
space $\cD_{1/2}(M)$ of complex half-densities $\rho$ on $M$:
\mar{qq82}\beq
\wh f\rho=(-i\bL_{a^j\dr_j} -b)\rho= (-ia^j\dr_j -\frac{i}{2}\dr_j
a^j- b) \rho,\label{qq82}
\eeq
where $\bL_{a^j\dr_j}$ is the Lie derivative (\ref{0231}) of
half-densities. A glance at the expression (\ref{qq82}) shows that
it is the  Schr\"odinger representation of the quantum algebra
$A_T$ of the affine functions (\ref{xci13'}). We call $\wh f$
(\ref{qq82}) the  Schr\"odinger operators.

Let $\gE_M\subset\cD_{1/2}(M)$ be a space of complex
half-densities $\rho$ of compact support on $M$ and $\ol \gE_M$
the completion of $\gE_M$ with respect to the non-degenerate
Hermitian form
\mar{nn}\beq
\lng \rho|\rho'\rng=\op\int_Q \rho\ol\rho'. \label{nn}
\eeq
The (unbounded) Schr\"odinger operators (\ref{qq82}) in the domain
$\gE_M$ of the Hilbert space $\ol \gE_M$ are Hermitian.

\subsection{Leafwise geometric quantization}

Developed for symplectic manifolds \cite{eche98,sni}, the
geometric quantization technique has been generalized to Poisson
manifolds in terms of contravariant connections
\cite{vais91,vais}. Though there is one-to-one correspondence
between  the Poisson structures on a smooth manifold and its
symplectic foliations, geometric quantization of a Poisson
manifold need not imply quantization of its symplectic leaves
\cite{book10,vais97}.

$\bullet$ Firstly, contravariant connections fail to admit the
pull-back operation. Therefore, prequantization of a Poisson
manifold does not determine straightforwardly prequantization of
its symplectic leaves.

$\bullet$ Secondly, polarization of a Poisson manifold is defined
in terms of sheaves of functions, and it need not be associated to
any distribution. As a consequence, its pull-back onto a leaf is
not polarization of a symplectic manifold in general.

$\bullet$ Thirdly, a quantum algebra of a Poisson manifold
contains  the center of a Poisson algebra. However, there are
models  where quantization of this center has no physical meaning.
For instance, a center of the Poisson algebra of a mechanical
system with classical parameters consists of functions of these
parameters.

Geometric quantization of symplectic foliations disposes of these
problems. A quantum algebra $\cA_\cF$ of a symplectic foliation
$\cF$ also is a quantum algebra of the associated Poisson manifold
such that its restriction to each symplectic leaf $F$ is a quantum
algebra of $F$. Thus, geometric quantization of a symplectic
foliation provides leafwise quantization of a Poisson manifold
\cite{jmp02,book05,book10}.

Geometric quantization of a symplectic foliation is phrased in
terms of leafwise connections on a foliated manifold (see
Definition \ref{lmp100} below). Any leafwise connection on a
complex line bundle over a Poisson manifold is proved to come from
a connection on this bundle (Theorem \ref{lmp42}).  Using this
fact, one can state the equivalence of prequantization of a
Poisson manifold to prequantization of its symplectic foliation
\cite{book10}, which also yields prequantization of each
symplectic leaf (Proposition \ref{lmp70}). We show that
polarization of a symplectic foliation is associated to particular
polarization of a Poisson manifold (Proposition \ref{lmp72}), and
its restriction to any symplectic leaf is polarization of this
leaf (Proposition \ref{lmp75}). Therefore, a quantum algebra of a
symplectic foliation is both a quantum algebra of a Poisson
manifold and, restricted to each symplectic leaf, a quantum
algebra of this leaf.

We define metaplectic correction of a symplectic foliation so that
its quantum algebra is represented by Hermitian operators in the
pre-Hilbert module of leafwise half-forms, integrable over the
leaves of this foliation.

Let $(Z,\{,\})$ be a Poisson manifold and $(\cF,\Om_\cF)$ its
symplectic foliation such that $\{,\}=\{,\}_\cF$ (see
(\ref{spr902})). Let leaves of $\cF$ be simply connected.

Prequantization of a symplectic foliation $(\cF,\Om_\cF)$ provides
a representation
\mar{lqm514}\beq
f\to \wh f, \qquad [\wh f,\wh f']=-i\wh{\{f,f'\}}_\cF,
\label{lqm514}
\eeq
of the Poisson algebra $(C^\infty(Z),\{,\}_\cF)$ by first order
differential operators on sections $s$ of some complex line bundle
$C\to Z$, called the  prequantization bundle. These operators are
given by the Kostant -- Souriau  prequantization formula
\mar{lqq46}\beq
\wh f=-i\nabla_{\vt_f}^\cF s +\ve fs, \qquad \vt_f=\Om_\cF^\sh(\wt
df), \qquad \ve\neq 0, \label{lqq46}
\eeq
where $\nabla^\cF$ is an admissible leafwise connection on $C\to
Z$  such that its curvature form $\wt R$ (\ref{lmp08}) obeys the
admissible condition
\mar{lmp61}\beq
\wt R=i\ve\Om_\cF\ot u_C, \label{lmp61}
\eeq
where $u_C$ is the Liouville vector field (\ref{0215}) on $C$.

Using the above mentioned fact that any leafwise connection comes
from a connection, we show  that prequantization of a symplectic
foliation yields prequantization of its symplectic leaves.

\begin{remark} If $Z$ is a symplectic manifold whose symplectic
foliation is reduced to $Z$ itself, the formulas (\ref{lqq46}) --
(\ref{lmp61}), $\ve=-1$, of leafwise prequantization restart the
formulas (\ref{qm506}) and (\ref{qm503}) of geometric quantization
of a symplectic manifold $Z$.
\end{remark}

Let $S_\cF(Z)\subset C^\infty(Z)$ be a subring of functions
constant on leaves of a foliation $\cF$, and let $\cT_1(\cF)$ be
the real Lie algebra of global sections of the tangent bundle
$T\cF\to Z$ to $\cF$. It is the Lie $S_\cF(Z)$-algebra of
derivations of $C^\infty(Z)$, regarded as a $S_\cF(Z)$-ring.

\begin{definition} \label{lmp100} \mar{lmp100}
In the framework of the leafwise differential calculus $\gF^*(Z)$
(\ref{spr892}), a (linear)  leafwise connection on a complex line
bundle $C\to Z$  is defined as a connection $\nabla^\cF$ on the
$C^\infty(Z)$-module $C(Z)$ of global sections of this bundle,
where $C^\infty(Z)$ is regarded as an $S_\cF(Z)$-ring (see
Definition \ref{1016}). It associates to each element $\tau\in
\cT_1(\cF)$ an $S_\cF(Z)$-linear endomorphism $\nabla_\tau^\cF$ of
$C(Z)$ which obeys the Leibniz rule
\mar{lmp55}\beq
\nabla_\tau^\cF(fs)=(\tau\rfloor\wt df)s +f\nabla_\tau^\cF(s),
\qquad f\in C^\infty(Z), \qquad s\in C(Z). \label{lmp55}
\eeq
\end{definition}

A linear connection on $C\to Z$ can be equivalently defined as a
connection on the module $C(Z)$ which assigns to each vector field
$\tau\in \cT_1(Z)$ on $Z$ an $\mathbb R$-linear endomorphism of
$C(Z)$ obeying the Leibniz rule (\ref{lmp55}). Restricted to
$\cT_1(\cF)$, it obviously yields a leafwise connection. In order
to show that any leafwise connection is of this form, we appeal to
an alternative definition of a leafwise connection in terms of
leafwise forms.

The inverse images $\pi^{-1}(F)$ of leaves $F$ of the foliation
$\cF$ of $Z$ provide a (regular) foliation $C_\cF$ of the line
bundle $C$. Given the (holomorphic) tangent bundle $TC_\cF$ of
this foliation, we have the exact sequence of vector bundles
\mar{lmp18}\beq
0\to VC\ar_C TC_\cF\ar_C C\op\times_Z T\cF\to 0, \label{lmp18}
\eeq
where $VC$ is the (holomorphic) vertical tangent bundle of $C\to
Z$.

\begin{definition} \label{lmp101} \mar{lmp101}
A (linear) leafwise connection on the complex line bundle $C\to Z$
is a splitting of the exact sequence (\ref{lmp18}) which is linear
over $C$.
\end{definition}

One can choose an adapted coordinate atlas $\{(U_\xi;z^\la,
z^i)\}$ (\ref{spr850}) of a foliated manifold $(Z,\cF)$ such that
$U_\xi$ are trivialization domains of the complex line bundle
$C\to Z$. Let $(z^\la, z^i,c)$, $c\in\mathbb C$, be the
corresponding bundle coordinates on $C\to Z$. They also are
adapted coordinates on the foliated manifold $(C,C_\cF)$. With
respect to these coordinates, a (linear) leafwise connection is
represented by a $TC_\cF$-valued leafwise one-form
\mar{lmp21}\beq
A_\cF=\wt dz^i\ot(\dr_i +A_ic\dr_c), \label{lmp21}
\eeq
where $A_i$ are local complex functions on $C$.

The exact sequence (\ref{lmp18}) is obviously a subsequence of the
exact sequence
\be
0\to VC\ar_C TC\ar_C C\op\times_Z TZ\to 0,
\ee
where $TC$ is the holomorphic tangent bundle of $C$. Consequently,
any connection
\mar{lmp103}\beq
\cK=dz^\la\ot(\dr_\la + \cK_\la c\dr_c) + dz^i\ot(\dr_i
+\cK_ic\dr_c) \label{lmp103}
\eeq
on the complex line bundle $C\to Z$ yields a leafwise connection
\mar{lmp23}\beq
\cK_\cF=\wt dz^i\ot(\dr_i +\cK_ic\dr_c). \label{lmp23}
\eeq

\begin{theorem} \label{lmp42} \mar{lmp42}
Any leafwise connection on the complex line bundle $C\to Z$ comes
from a connection on it \cite{book10}.
\end{theorem}

In particular, it follows that Definitions \ref{lmp100} and
\ref{lmp101} of a leafwise connection are equivalent, namely,
\be
\nabla^\cF s=\wt ds- A_i s\wt dz^i, \qquad s\in C(Z).
\ee

The  curvature of a leafwise connection $\nabla^\cF$ is defined as
a $C^\infty(Z)$-linear endomorphism
\mar{lmp07}\beq
\wt R(\tau,\tau')=\nabla_{[\tau,\tau']}^\cF- [\nabla_\tau^\cF,
\nabla_{\tau'}^\cF]=\tau^i \tau'^j R_{ij}, \quad R_{ij}=\dr_i
A_j-\dr_j A_i, \label{lmp07}
\eeq
of $C(Z)$ for any vector fields $\tau,\tau'\in \cT_1(\cF)$. It is
represented by the vertical-valued leafwise two-form
\mar{lmp08}\beq
\wt R=\frac12 R_{ij}\wt dz^i\w \wt dz^j\ot u_C. \label{lmp08}\\
\eeq
If a leafwise connection $\nabla^\cF$ comes from a connection
$\nabla$, its  curvature leafwise form $\wt R$ (\ref{lmp08}) is an
image $\wt R=i^*_\cF R$ of the curvature form $R$ (\ref{161'}) of
the connection $\nabla$ with respect to the morphism $i^*_\cF$
(\ref{lmp04}).

Now let us turn to the admissible condition (\ref{lmp61}).

\begin{lemma} \label{lmp63} \mar{lmp63}
Let us assume that there exists a leafwise connection $\cK_\cF$ on
the complex line bundle $C\to Z$ which fulfils the admissible
condition (\ref{lmp61}). Then, for any Hermitian fibre metric $g$
in $C\to Z$, there exists a leafwise connection $A_\cF^g$ on $C\to
Z$ which:

(i) satisfies the admissible condition (\ref{lmp61}),

(ii) preserves $g$,

(iii) comes from a $U(1)$-principal connection on $C\to Z$.

\noindent This leafwise connection $A_\cF^g$ is called admissible.
\end{lemma}

\begin{proof}
Given a Hermitian fibre metric $g$ in $C\to Z$, let
$\Psi^g=\{(z^\la,z^i,c)\}$ an associated bundle atlas of $C$ with
$U(1)$-valued transition functions such that $g(c,c')=c\ol c'$
(Proposition \ref{0220}). Let the above mentioned leafwise
connection $\cK_\cF$ come from a linear connection $\cK$
(\ref{lmp103}) on $C\to Z$ written with respect to the atlas
$\Psi^g$. The connection $\cK$ is split into the sum $A^g + \g$
where
\mar{lmp62}\beq
A^g=dz^\la\ot(\dr_\la + {\rm Im}(\cK_\la)c\dr_c) + dz^i\ot(\dr_i
+{\rm Im}(\cK_i) c\dr_c) \label{lmp62}
\eeq
is a $U(1)$-principal connection, preserving the Hermitian fibre
metric $g$. The curvature forms $R$ of $\cK$ and $R^g$ of $A^g$
obey the relation $R^g={\rm Im}(R)$. The connection $A^g$
(\ref{lmp62}) defines the leafwise connection
\mar{lmp73}\beq
A_\cF^g=i_\cF^*A= \wt dz^i\ot(\dr_i + iA^g_i c\dr_c), \qquad
iA^g_i= {\rm Im}(\cK_i), \label{lmp73}
\eeq
preserving the Hermitian fibre metric $g$. Its curvature fulfils a
desired relation
\mar{lmp65}\beq
\wt R^g=i_\cF^*R^g={\rm Im}(i_\cF^*R)= i\ve\Om_\cF\ot u_C.
\label{lmp65}
\eeq
{}
\end{proof}

Since $A^g$ (\ref{lmp62}) is a $U(1)$-principal connection, its
curvature form $R^g$ is related to the first Chern form of
integral de Rham cohomology class by the formula (\ref{0223}). If
the admissible condition (\ref{lmp61}) holds, the relation
(\ref{lmp65}) shows that the leafwise cohomology class of the
leafwise form $-(2\pi)^{-1}\ve\Om_\cF$ is an image of an integral
de Rham cohomology class with respect to the cohomology morphism
$[i^*_\cF]$. Conversely, if a leafwise symplectic form $\Om_\cF$
on a foliated manifold $(Z,\cF)$ is of this type, there exist a
prequantization bundle $C\to Z$ and a $U(1)$-principal connection
$A$ on $C\to Z$ such that the leafwise connection $i^*_\cF A$
fulfils the relation (\ref{lmp61}). Thus, we have stated the
following.

\begin{proposition} \label{lmp66} \mar{lmp66}
A symplectic foliation $(\cF,\Om_\cF)$ of a manifold $Z$ admits
the prequantization (\ref{lqq46}) iff the leafwise cohomology
class of $-(2\pi)^{-1}\ve\Om_\cF$ is an image of an integral de
Rham cohomology class of $Z$.
\end{proposition}

Let $F$ be a leaf of a symplectic foliation $(\cF,\Om_\cF)$
provided with the symplectic form
\be
\Om_F=i^*_F\Om_\cF.
\ee
In accordance with Proposition \ref{lmp32} and the commutative
diagram
\be
\begin{array}{ccc}
  H^*(Z;\mathbb Z) &\ar & H^*_{\rm DR}(Z)\\
  \put(0,10){\vector(0,-1){20}} & & \put(0,10){\vector(0,-1){20}}\\
H^*(F;\mathbb Z) &\ar & H^*_{\rm DR}(F)
\end{array}
\ee
of groups of the de Rham cohomology $H^*_{\rm DR}(*)$ and the
cohomology $H^*(*;\mathbb Z)$ with coefficients in the constant
sheaf $\mathbb Z$, the symplectic form $-(2\pi)^{-1}\ve \Om_F$
belongs to an integral de Rham cohomology class if a leafwise
symplectic form $\Om_\cF$ fulfils the condition of Proposition
\ref{lmp66}. This states the following.

\begin{proposition} \label{lmp70} \mar{lmp70}
If a symplectic foliation admits prequantization, each its
symplectic leaf does prequantization too.
\end{proposition}

The corresponding  prequantization bundle for $F$ is the pull-back
complex line bundle $i^*_FC$, coordinated by $(z^i,c)$.
Furthermore, let $A_\cF^g$ (\ref{lmp73}) be a leafwise connection
on the prequantization bundle $C\to Z$ which obeys Lemma
\ref{lmp63}, i.e., comes from a $U(1)$-principal connection $A^g$
on $C\to Z$. Then the pull-back
\mar{lmp130}\beq
A_F=i^*_FA^g=dz^i\ot(\dr_i +ii^*_F(A^g_i)c\dr_c) \label{lmp130}
\eeq
of the connection $A^g$ onto $i^*_FC\to F$ satisfies the
admissible condition
\be
R_F=i^*_FR=i\ve \Om_F,
\ee
and preserves the pull-back Hermitian fibre metric $i^*_Fg$ in
$i^*_\cF C\to F$.

Let us define  polarization of a symplectic foliation
$(\cF,\Om_\cF)$ of a manifold $Z$ as a maximal (regular)
involutive distribution $\bT\subset T\cF$ on $Z$ such that
\mar{lmp71}\beq
\Om_\cF(u,v)=0, \qquad  u,v\in\bT_z, \qquad z\in Z. \label{lmp71}
\eeq
Given  the Lie algebra $\bT(Z)$ of $\bT$-subordinate vector fields
on $Z$, let $\cA_\cF\subset C^\infty(Z)$ be the complexified
subalgebra of functions $f$ whose leafwise Hamiltonian vector
fields $\vt_f$ (\ref{spr902}) fulfil the condition
\be
[\vt_f,\bT(Z)]\subset \bT(Z).
\ee
It is called the  quantum algebra of a symplectic foliation
$(\cF,\Om_\cF)$ with respect to the polarization $\bT$.  This
algebra obviously contains the center $S_\cF(Z)$ of the Poisson
algebra $(C^\infty(Z),\{,\}_\cF)$, and it is a Lie
$S_\cF(Z)$-algebra.

\begin{proposition} \label{lmp72} \mar{lmp72}
Every polarization $\bT$ of a symplectic foliation $(\cF,\Om_\cF)$
yields polarization of the associated Poisson manifold
$(Z,w_\Om)$.
\end{proposition}

\begin{proof}
Let us consider the presheaf of local smooth functions $f$ on $Z$
whose leafwise Hamiltonian vector fields $\vt_f$ (\ref{spr902})
are subordinate to $\bT$. The sheaf $\Phi$ of germs of these
functions is polarization of the Poisson manifold $(Z,w_\Om)$ (see
Remark \ref{0245} below). Equivalently, $\Phi$ is the sheaf of
germs of functions on $Z$ whose leafwise differentials are
subordinate to the codistribution $\Om_\cF^\fl\bT$.
\end{proof}

\begin{remark} \label{0245} \mar{0245}
Let us recall that  polarization of a Poisson manifold $(Z,\{,\})$
is defined as a sheaf $\bT^*$ of germs of complex functions on $Z$
whose stalks $\bT^*_z$, $z\in Z$, are Abelian algebras with
respect to the Poisson bracket $\{,\}$ \cite{vais97}. Let
$\bT^*(Z)$ be the structure algebra  of global sections of the
sheaf $\bT^*$; it also is called the Poisson polarization
\cite{vais91,vais}. A  quantum algebra $\cA$ associated to the
Poisson polarization $\bT^*$ is defined as a subalgebra of the
Poisson algebra $C^\infty(Z)$ which consists of functions $f$ such
that
\be
\{f,\bT^*(Z)\}\subset \bT^*(Z).
\ee
Polarization of a symplectic manifold yields its Poisson one.
\end{remark}

Let us note that the polarization $\Phi$ in the proof of
Proposition \ref{lmp72}) need not be maximal, unless $\bT$ is of
maximal dimension $\di\cF/2$. It belongs to the following
particular type of polarizations of a Poisson manifold. Since the
cochain morphism $i^*_\cF$ (\ref{lmp04}) is an epimorphism, the
leafwise differential calculus $\gF^*$ is universal, i.e., the
leafwise differentials $\wt df$ of functions $f\in C^\infty(Z)$ on
$Z$ make up a basis for the $C^\infty(Z)$-module $\gF^1(Z)$. Let
$\Phi(Z)$ denote the structure $\mathbb R$-module of global
sections of the sheaf $\Phi$. Then the leafwise differentials of
elements of $\Phi(Z)$ make up a basis for the $C^\infty(Z)$-module
of global sections of the codistribution $\Om_\cF^\fl\bT$.
Equivalently, the leafwise Hamiltonian vector fields of elements
of $\Phi(Z)$ constitute a basis for the $C^\infty(Z)$-module
$\bT(Z)$. Then one can easily show that polarization $\bT$ of a
symplectic foliation $(\cF,\Om_\cF)$ and the corresponding
polarization $\Phi$ of the Poisson manifold $(Z,w_\Om)$ in
Proposition \ref{lmp72} define the same quantum algebra $\cA_\cF$.

Let $(F,\Om_F)$ be a symplectic leaf of a symplectic foliation
$(\cF,\Om_\cF)$. Given polarization $\bT\to Z$ of $(\cF,\Om_\cF)$,
its restriction
\be
\bT_F=i^*_F\bT\subset i^*_FT\cF=TF
\ee
to $F$ is an involutive distribution on $F$. It obeys the
condition
\be
i^*_F\Om_\cF(u,v)=0, \qquad  u,v\in\bT_{Fz}, \qquad z\in F,
\ee
i.e., it is  polarization of the symplectic manifold $(F,\Om_F)$.
Thus, we have stated the following.

\begin{proposition} \label{lmp75} \mar{lmp75}
Polarization of a symplectic foliation defines polarization of
each symplectic leaf.
\end{proposition}

Clearly, the  quantum algebra $\cA_F$  of a symplectic leaf $F$
with respect to the polarization $\bT_F$ contains all elements
$i^*_Ff$ of the quantum algebra $\cA_\cF$ restricted to $F$.

Since $\cA_\cF$ is the quantum algebra both of a symplectic
foliation $(\cF,\Om_\cF)$ and the associated Poisson manifold
$(Z,w_\Om)$, let us follow the standard metaplectic correction
technique \cite{eche98,book10}.

Assuming that $Z$ is oriented and that $H^2(Z;\mathbb Z_2)=0$, let
us consider the metalinear complex line bundle $\cD_{1/2}[Z]\to Z$
characterized by an atlas
\be
\Psi_Z=\{(U;z^\la,z^i,r)\}
\ee
with the transition functions (\ref{0232}). Global sections $\rho$
of this bundle are half-densities on $Z$. Their Lie derivative
(\ref{0231}) along a vector field $u$ on $Z$ reads
\mar{lmp78}\beq
\bL_u\rho=u^\la \dr_\la\rho + u^i\dr_i\rho+\frac12(\dr_\la u^\la+
\dr_i u^i)\rho. \label{lmp78}
\eeq

Given an admissible connection $A^g_\cF$, the prequantization
formula (\ref{lqq46}) is extended to sections $\vr=s\ot\rho$ of
the fibre bundle
\mar{0240}\beq
C\op\ot_Z\cD_{1/2}[Z] \label{0240}
\eeq
as follows
\mar{lmp80}\ben
&& \wh f=-i[(\nabla_{\vt_f}^\cF +i\ve f)\ot\id
+\id\ot\bL_{\vt_f}]=
\label{lmp80}\\
&& \qquad -i[\nabla_{\vt_f}^\cF +i\ve f+\frac12\dr_i\vt_f^i],
\qquad f\in\cA_\cF. \nonumber
\een
This extension is the  metaplectic correction of leafwise
quantization. It is readily observed that the operators
(\ref{lmp80}) obey Dirac's condition (\ref{lqm514}). Let us denote
by $\gE_Z$ the complex space of sections $\vr$ of the fibre bundle
(\ref{0240}) of compact support such that
\be
(\nabla_\vt^\cF\ot\id +\id\ot\bL_\vt)\vr=(\nabla_\vt^\cF
+\frac12\dr_i\vt^i)\vr=0
\ee
for all $\bT$-subordinate leafwise Hamiltonian vector fields
$\vt$.

\begin{lemma} \label{+310'} \mar{+310'}
For any function $f\in\cA_\cT$ and an arbitrary section $\vr\in
\gE_Z$, the relation $\wh f\vr\in\gE_Z$ holds.
\end{lemma}

Thus, we have a representation of the quantum algebra $\cA_\cF$ in
the space $\gE_Z$. Therefore, by quantization of a function
$f\in\cA_\cF$ is meant the restriction of the operator $\wh f$
(\ref{lmp80}) to $\gE_Z$.

The space $\gE_Z$ is provided with the non-degenerate Hermitian
form
\mar{lmp77}\beq
\lng\rho|\rho'\rng=\op\int_Z \rho\rho', \label{lmp77}
\eeq
which brings $\gE_Z$ into a pre-Hilbert space. Its completion
carries a representation of the quantum algebra $\cA_\cF$ by
(unbounded) Hermitian operators.

However, it may happen that the above quantization has no physical
meaning because the Hermitian form (\ref{lmp77}) on the carrier
space $\gE_Z$ and, consequently, the mean values of operators
(\ref{lmp80}) are defined by integration over the whole manifold
$Z$. For instance, it implies integration over time and classical
parameters. Therefore, we suggest a different scheme of
quantization of symplectic foliations.

Let us consider the exterior bundle $\op\w^{2m}T\cF^*$,
$2m=\di\cF$. Its structure group $GL(2m,\mathbb R)$ is reducible
to the group $GL^+(2m,\mathbb R)$ since a symplectic foliation is
oriented. One can regard this fibre bundle as being associated to
a $GL(2m,\mathbb C)$-principal bundle $P\to Z$. As earlier, let us
assume that $H^2(Z;\mathbb Z_2)=0$. Then the principal bundle $P$
admits a two-fold covering principal bundle with the  structure
 metalinear group $ML(2m,\mathbb C)$ \cite{eche98}.  As a
consequence, there exists a complex line bundle $\cD_\cF\to Z$
characterized by an atlas
\be
\Psi_\cF=\{(U_\xi;z^\la,z^i,r)\}
\ee
with the transition functions $r'=J_\cF r$ such that
\mar{kk203}\beq
J_\cF\ol J_\cF=\det\left(\frac{\dr z'^i}{\dr z^j}\right).
\label{kk203}
\eeq
One can think of its sections as being complex  leafwise
half-densities on $Z$. The  metalinear bundle $\cD_{1/2}[\cF]\to
Z$ admits the canonical lift of any $\bT$-subordinate vector field
$u$ on $Z$. The corresponding Lie derivative of its sections reads
\mar{lmp82}\beq
\bL_u^\cF=u^i\dr_i+\frac12\dr_i u^i. \label{lmp82}
\eeq

We define the  quantization bundle as the tensor product
\mar{0241}\beq
Y_\cF=C\op\ot_Z\cD_{1/2}[\cF]\to Z. \label{0241}
\eeq
Its sections are $C$-valued  leafwise half-forms. Given an
admissible leafwise connection $A^g_\cF$ and the Lie derivative
$\bL^\cF_u$ (\ref{lmp82}), let us associate the first order
differential operator
\mar{lmp84}\ben
&& \wh f=-i[(\nabla_{\vt_f}^\cF +i\ve f)\ot\id
+\id\ot\bL_{\vt_f}^\cF]= \label{lmp84}\\
&& \qquad -i[\nabla_{\vt_f}^\cF +i\ve f+\frac12\dr_i\vt_f^i],
\qquad f\in\cA_\cF, \nonumber
\een
on sections $\vr_\cF$ of $Y_\cF$ to each element of the quantum
algebra $\cA_\cF$. A direct computation with respect to the local
Darboux coordinates on $Z$ proves the following.

\begin{lemma} \label{lmp120} \mar{lmp120}
The operators (\ref{lmp84}) obey Dirac's condition (\ref{lqm514}).
\end{lemma}

\begin{lemma} \label{lmp121} \mar{lmp121}
If a section $\vr_\cF$ fulfils the condition
\mar{lmp86}\beq
(\nabla_\vt^\cF\ot\id +\id\ot\bL_\vt^\cF)\vr_\cF=(\nabla_\vt^\cF
+\frac12\dr_i\vt^i)\vr_\cF=0  \label{lmp86}
\eeq
for all $\bT$-subordinate leafwise Hamiltonian vector field $\vt$,
then $\wh f\vr_\cF$ for any $f\in\cA_\cF$ possesses the same
property.
\end{lemma}

Let us restrict the representation of the quantum algebra
$\cA_\cF$ by the operators (\ref{lmp84}) to the subspace $\gE_\cF$
of sections $\vr_\cF$ of the quantization bundle (\ref{0241})
which obey the condition (\ref{lmp86}) and whose restriction to
any leaf of $\cF$ is of compact support. The last condition is
motivated by the following.

Since $i^*_FT\cF^*=T^*F$, the pull-back $i^*_F\cD_{1/2}[\cF]$ of
$\cD_{1/2}[\cF]$ onto a leaf $F$ is a metalinear bundle of
half-densities on $F$. By virtue of Propositions \ref{lmp70} and
\ref{lmp75}, the pull-back $i^*_FY_\cF$ of the quantization bundle
$Y_\cF\to Z$ onto $F$ is a quantization bundle for the symplectic
manifold $(F,i^*_F\Om_\cF)$. Given the pull-back connection $A_F$
(\ref{lmp130}) and the polarization $\bT_F=i^*_F\bT$, this
symplectic manifold is subject to the standard geometric
quantization by the first order differential operators
\mar{lmp92}\beq
\wh f=-i(i_F^*\nabla_{\vt_f}^\cF +i\ve f +\frac12\dr_i\vt_f^i),
\qquad f\in \cA_F,  \label{lmp92}
\eeq
on sections $\vr_F$ of $i^*_FY_\cF\to F$ of compact support which
obey the condition
\mar{lmp133}\beq
(i_F^*\nabla_\vt^\cF +\frac12\dr_i\vt^i)\vr_F=0 \label{lmp133}
\eeq
for all $\bT_F$-subordinate Hamiltonian vector fields $\vt$ on
$F$. These sections constitute a pre-Hilbert space $\gE_F$ with
respect to the Hermitian form
\be
\lng\rho_F|\rho'_F\rng=\op\int_F \vr_F \vr'_F.
\ee
The key point is the following.

\begin{proposition} \label{lmp132} \mar{lmp132}
We have $i^*_F\gE_\cF\subset \gE_F$, and the relation
\beq
i^*_F(\wh f\vr_\cF)=\wh{(i^*_Ff)}(i^*_F\vr_\cF) \label{lmp93}
\eeq
holds for all elements $f\in\cA_\cF$ and $\vr_\cF\in \gE_\cF$.
\end{proposition}

\begin{proof} One can use the fact that the expressions (\ref{lmp92})
and (\ref{lmp133}) have the same coordinate form as the
expressions (\ref{lmp84}) and (\ref{lmp86}) where $z^\la=$const.
\end{proof}

The relation (\ref{lmp93}) enables one to think of the operators
$\wh f$ (\ref{lmp84}) as being the leafwise quantization of the
$S_\cF(Z)$-algebra $\cA_\cF$ in the  pre-Hilbert $S_\cF(Z)$-module
$\gE_\cF$ of leafwise half-forms.

\subsection{Quantization of integrable systems in action-angle variables}

In accordance with Theorem \ref{nc0'}, any superintegrable
Hamiltonian system (\ref{nc1}) on a symplectic manifold $(Z,\Om)$
restricted to some open neighborhood $U_M$ (\ref{kk600}) of its
invariant submanifold $M$ is characterized by generalized
action-angle coordinates $(I_\la,p_A,q^A, y^\la)$, $\la=1,\ldots,
m$, $A=1,\ldots,n-m$. They are canonical for the symplectic form
$\Om$ (\ref{cmp6}) on $U_M$. Then one can treat the coordinates
$(I_\la,p_A)$ as $n$ independent functions in involution on a
symplectic annulus $(U_M,\Om)$ which constitute a completely
integrable system in accordance with Definition \ref{cmp21}.
Strictly speaking, its quantization fails to be a quantization of
the original superintegrable system (\ref{nc1}) because
$F_i(I_\la,q^A,p_A)$ are not linear functions and, consequently,
the algebra (\ref{nc1}) and the algebra
\mar{kk610}\beq
\{I_\la, p_A\}=\{I_\la, q^A\}=0, \qquad \{p_A,q^B\}=\dl^B_A
\label{kk610}
\eeq
are not isomorphic in general. However, one can obtain the
Hamilton operator $\wh\cH$ and the Casimir operators $\wh C_\la$
of an original superintegrable system and their spectra.

There are different approaches to quantization of completely
integrable and superintegrable systems
\cite{acang1,plet02,gos,gutz,myk2}. It should be emphasized that
action-angle coordinates need not be globally defined on a phase
space, but form an algebra of the Poisson canonical commutation
relations (\ref{kk610}) on an open neighborhood $U_M$ of an
invariant submanifold $M$. Therefore, quantization of an
integrable system with respect to the action-angle variables is a
quantization of the Poisson algebra $C^\infty(U_M)$ of real smooth
functions on $U_M$. Since there is no morphism $C^\infty(U_M)\to
C^\infty(Z)$, this quantization is not equivalent to quantization
of an original integrable system on $Z$ and, from on a physical
level, is interpreted as quantization around an invariant
submanifold $M$. A key point is that, since $U_M$ is not a
contractible manifold, the geometric quantization technique should
be called into play in order to quantize an integrable system
around its invariant submanifold. A peculiarity of the geometric
quantization procedure is that it remains equivalent under
symplectic isomorphisms, but essentially depends on the choice of
a polarization \cite{blat,raw}.

Geometric quantization of completely integrable systems has been
studied at first with respect to the polarization spanned by
Hamiltonian vector fields of integrals of motion \cite{myk2}. For
example, the well-known Simms quantization of a harmonic
oscillator is of this type \cite{eche98}. However, one meets a
problem that the associated quantum algebra contains affine
functions of angle coordinates on a torus which are ill defined.
As a consequence, elements of the carrier space of this
quantization fail to be smooth, but are tempered distributions. We
have developed a different variant of geometric quantization of
completely integrable systems \cite{acang1,plet02,book05}. Since a
Hamiltonian of a completely integrable system depends only on
action variables, it seems natural to provide the Schr\"odinger
representation of action variables by first order differential
operators on functions of angle coordinates. For this purpose, one
should choose the angle polarization of a symplectic manifold
spanned by almost-Hamiltonian vector fields of angle variables.

Given an open neighborhood $U_M$ (\ref{kk600}) in Theorem
\ref{nc0'}, let us consider its fibration
\mar{kk601}\ben
&& U_M=N_M\times \mathbb R^{m-r}\times T^r\to V\times \mathbb
R^{m-r}\times T^r =\cM, \label{kk601}\\
&& (I_\la,p_A,q^A, y^\la)\to (q^A, y^\la).
\een
Then one can think of a symplectic annulus $(U_M,\Om)$ as being an
open subbundle of the cotangent bundle $T^*\cM$ endowed with the
canonical symplectic form $\Om_T=\Om$ (\ref{cmp6}). This fact
enables us to provide quantization of any superintegrable system
on a neighborhood of its invariant submanifold as geometric
quantization of the cotangent bundle $T^*\cM$ over the toroidal
cylinder $\cM$ (\ref{kk601}) \cite{pl07}. Note that this
quantization however differs from that in Section 5.1 because
$\cM$ (\ref{kk601}) is not simply connected in general.

Let $(q^A,r^a,\al^i)$ be coordinates on the toroidal cylinder
$\cM$ (\ref{kk601}), where $(\al^1,\ldots, \al^r)$ are angle
coordinates on a torus $T^r$, and let $(p_A, I_a, I_i)$ be the
corresponding action coordinates (i.e., the holonomic fibre
coordinates on $T^*\cM$).  Since the symplectic form $\Om$
(\ref{cmp6}) is exact, the quantum bundle is defined as a trivial
complex line bundle $C$ over $T^*\cM$. Let its trivialization hold
fixed. Any other trivialization leads to an equivalent
quantization of $T^*\cM$. Given the associated fibre coordinate
$c\in\mathbb C$ on $C\to T^*\cM$, one can treat its sections as
smooth complex functions on $T^*\cM$.

The Kostant -- Souriau prequantization formula (\ref{qm506})
associates to every smooth real function $f$ on $T^*\cM$ the first
order differential operator
\be
\wh f= -i\vt_f\rfloor D^A -fc\dr_c
\ee
on sections of $C\to T^*\cM$, where $\vt_f$ is the Hamiltonian
vector field of $f$ and $D^A$ is the covariant differential
(\ref{0214}) with respect to an admissible $U(1)$-principal
connection $A$ on $C$. This connection preserves the Hermitian
fibre metric $g(c,c')=c\ol c'$ in $\cC$, and its curvature obeys
the prequantization condition (\ref{qm503}). Such a connection
reads
\mar{ci20}\beq
A=A_0 -ic(p_Adq^A + I_adr^a + I_i d\al^i)\ot\dr_c, \label{ci20}
\eeq
where $A_0$ is a flat $U(1)$-principal connection on $C\to
T^*\cM$.

The classes of gauge non-conjugate flat principal connections on
$C$ are indexed by the set $\mathbb R^r/\mathbb Z^r$ of
homomorphisms of the de Rham cohomology group
\be
H^1_{\rm DR}(T^*\cM)=H^1_{\rm DR}(\cM)=H^1_{\rm DR}(T^r)=\mathbb
R^r
\ee
of $T^*\cM$ to $U(1)$. We choose their representatives of the form
\be
&& A_0[(\la_i)]=dp_A\ot\dr^A + dI_a\ot\dr^a + dI_j\ot\dr^j+
dq^A\ot\dr_A +dr^a\ot\dr_a + \\
&&\qquad d\al^j\ot(\dr_j -i\la_j c\dr_c),
\qquad \la_i\in [0,1).
\ee
Accordingly, the relevant connection (\ref{ci20}) on $C$ reads
\mar{ci14}\ben
&& A[(\la_i)]= dp_A\ot\dr^A + dI_a\ot\dr^a + dI_j\ot\dr^j
+\label{ci14}\\
&& \qquad dq^A\ot(\dr_A - ip_Ac\dr_c) + dr^a\ot(\dr_a - iI_ac\dr_c)
+ \nonumber\\
&& \qquad d\al^j\ot(\dr_j - i(I_j +\la_j)c\dr_c). \nonumber
\een
For the sake of simplicity, we further assume that the numbers
$\la_i$ in the expression (\ref{ci14}) belong to $\mathbb R$, but
bear in mind that connections $A[(\la_i)]$ and $A[(\la'_i)]$ with
$\la_i-\la'_i\in\mathbb Z$ are gauge conjugate.

Let us choose the above mentioned  angle polarization coinciding
with the vertical polarization $VT^*\cM$.  Then the corresponding
quantum algebra $\cA$ of $T^*\cM$ consists of affine functions
\be
f=a^A(q^B,r^b,\al^j)p_A + a^b(q^B,r^a,\al^j)I_b +
a^i(q^B,r^a,\al^j)I_i + b(q^B,r^a,\al^j)
\ee
in action coordinates $(p_A, I_a, I_i)$. Given a connection
(\ref{ci14}), the corresponding Schr\"odinger operators
(\ref{qq82}) read
\mar{lmp135x}\ben
&& \wh f=\left(-ia^A\dr_A-\frac{i}{2}\dr_Aa^A\right)
+\left(-ia^b\dr_i-\frac{i}{2}\dr_ba^b\right) + \label{lmp135x}\\
&& \qquad \left(-ia^i\dr_i-\frac{i}{2}\dr_i a^i +a^i\la_i\right)- b.
\nonumber
\een
They are Hermitian operators in the pre-Hilbert space $\gE_\cM$ of
complex half-densities $\psi$ of compact support on $\cM$ endowed
with the Hermitian form
\be
\lng \psi|\psi'\rng=\op\int_{\cM} \psi \ol \psi' d^{n-m} q d^{m-r}
r d^r\al.
\ee
Note that, being a complex function on a toroidal cylinder
$\mathbb R^{m-r}\times T^r$, any half-density $\psi\in \gE_\cM$ is
expanded into the series
\mar{j25}\beq
\psi= \op\sum_{(n_\m)} \f(q^B, r^a)_{(n_j)}\exp[in_j\al^j], \qquad
(n_j)=(n_1,\ldots,n_r)\in\mathbb Z^r, \label{j25}
\eeq
where $\f(q^B, r^a)_{(n_\m)}$ are half-densities of compact
support on $\mathbb R^{n-r}$. In particular, the action operators
(\ref{lmp135x}) read
\mar{j11}\beq
\wh p_A=-i\dr_A, \qquad \wh I_a=-i\dr_a, \qquad \wh I_j=-i\dr_j
+\la_j. \label{j11}
\eeq
It should be emphasized that
\mar{j31}\beq
\wh a \wh p_A\neq \wh{ap_A}, \qquad \wh a \wh I_b\neq \wh{a I_b},
\qquad \wh a \wh I_j\neq \wh{aI_j}, \qquad a\in C^\infty(\cM).
\label{j31}
\eeq

The operators (\ref{lmp135x}) provide a desired quantization of a
superintegrable Hamiltonian system written with respect to the
action-angle coordinates. They satisfy Dirac's condition
(\ref{qm514}). However, both a Hamiltonian $\cH$ and original
integrals of motion $F_i$ do not belong to the quantum algebra
$\cA$, unless they are affine functions in the action coordinates
$(p_A, I_a, I_i)$. In some particular cases, integrals of motion
$F_i$ can be represented by differential operators, but this
representation fails to be unique because of inequalities
(\ref{j31}), and Dirac's condition need not be satisfied. At the
same time, both the Casimir functions $C_\la$ and a Hamiltonian
$\cH$ (Proposition \ref{zz1}) depend only on action variables
$I_a,I_i$. If they are polynomial in $I_a$, one can associate to
them the operators $\wh C_\la=C_\la(\wh I_a,\wh I_j)$,
$\wh\cH=\cH(\wh I_a,\wh I_j)$, acting in the space $\gE_\cM$ by
the law
\be
&& \wh \cH\psi= \op\sum_{(n_j)} \cH(\wh
I_a,n_j+\la_j)\f(q^A,r^a)_{(n_j)}
\exp[in_j\al^j], \\
&& \wh C_\la\psi= \op\sum_{(n_j)} C_\la(\wh
I_a,n_j+\la_j)\f(q^A,r^a)_{(n_j)} \exp[in_j\al^j].
\ee

\begin{example} Let us consider a superintegrable system with the Lie algebra
$\cG=so(3)$ of integrals of motion $\{F_1,F_2,F_3\}$ on a
four-dimensional symplectic manifold $(Z,\Om)$, namely,
\be
\{F_1,F_2\}=F_3, \qquad \{F_2,F_3\}=F_1, \qquad \{F_3,F_1\}=F_2.
\ee
Since it is compact, an invariant submanifold of a superintegrable
system in question is a circle $M=S^1$. We have a fibred manifold
$F:Z\to N$ (\ref{zz103}) onto an open subset $N\subset \cG^*$ of
the Lie coalgebra $\cG^*$. This fibred manifold is a fibre bundle
since its fibres are compact (Theorem \ref{11t2}). Its base $N$ is
endowed with the coordinates $(x_1,x_2,x_3)$ such that integrals
of motion $\{F_1,F_2,F_3\}$ on $Z$ read
\be
F_1=x_1, \qquad F_2=x_2, \qquad F_3=x_3.
\ee
The coinduced Poisson structure on $N$ is the Lie -- Poisson
structure (\ref{j51b}). The coadjoint action of $so(3)$ is given
by the expression (\ref{kk605}). An orbit of the coadjoint action
of dimension 2 is given by the equality (\ref{kk606}). Let $M$ be
an invariant submanifold such that the point $F(M)\in \cG^*$
belongs to the orbit (\ref{kk606}). Let us consider an open fibred
neighborhood $U_M=N_M\times S^1$ of $M$ which is a trivial bundle
over an open contractible neighborhood $N_M$ of $F(M)$ endowed
with the coordinates $(I,x_1,\g)$ defined by the equalities
(\ref{j52x}). Here, $I$ is the Casimir function (\ref{j52}) on
$\cG^*$. These coordinates are the Darboux coordinates of the Lie
-- Poisson structure (\ref{j53}) on $N_M$. Let $\vt_I$ be the
Hamiltonian vector field of the Casimir function $I$ (\ref{j52}).
Its flows are invariant submanifolds. Let $\al$ be a parameter
(\ref{zz121}) along the flows of this vector field. Then $U_M$ is
provided with the action-angle coordinates $(I,x_1,\g,\al)$ such
that the Poisson bivector on $U_M$ takes the form (\ref{j54}). The
action-angle variables $\{I, H_1=x_1, \g\}$ constitute a
superintegrable system
\mar{j55}\beq
\{I, F_1\}=0, \qquad \{I, \g\}=0, \qquad \{F_1,\g\}=1, \label{j55}
\eeq
on $U_M$. It is related to the original one by the transformations
\be
&& I=-\frac12(F_1^2 + F_2^2 + F_3^2)^{1/2}, \\
&& F_2=\left(-\frac1{2I}-F_1^2\right)^{1/2}\sin\g, \qquad
F_3=\left(-\frac1{2I}-H_1^2\right)^{1/2}\cos\g.
\ee
Its Hamiltonian is expressed only in the action variable $I$. Let
us quantize the superintegrable system (\ref{j55}). We obtain the
algebra of operators
\be
\wh f= a\left(-i\frac{\dr}{\dr \al} -\la\right) -ib\frac{\dr}{\dr
\g} -\frac{i}2\left(\frac{\dr a}{\dr \al} + \frac{\dr b}{\dr
\g}\right) - c,
\ee
where $a$, $b$, $c$ are smooth functions of angle coordinates
$(\g,\al)$ on the cylinder $\mathbb R\times S^1$. In particular,
the action operators read
\be
\wh I= -i\frac{\dr}{\dr \al} -\la, \qquad \wh F_1=-i\frac{\dr}{\dr
\g}.
\ee
These operators act in the space of smooth complex functions
\be
\psi(\g,\al)= \op\sum_k \f(\g)_k\exp[ik\al]
\ee
on $T^2$. A Hamiltonian $\cH(I)$ of a classical superintegrable
system also can be represented by the operator
\be
\wh \cH(I)\psi=\op\sum_k \cH(I-\la)\f(\g)_k\exp[ik\al]
\ee
on this space.
\end{example}

\section{Mechanics with time-dependent parameters}

At present, quantum systems with classical parameters attract
special attention in connection with holonomic quantum
computation.

This Section addresses mechanical systems with time-dependent
parameters. These parameters can be seen as sections of some
smooth fibre bundle $\Si\to\mathbb R$ called the  parameter
bundle. Then a configuration space of a mechanical system with
time-dependent parameters is  a composite fibre bundle
\mar{z500}\beq
Q\ar^{\pi_{Q\Si}}\Si\ar\mathbb R \label{z500}
\eeq
\cite{book05,book98,sard00}. Indeed, given a section $\vs(t)$ of a
parameters bundle $\Si\to\mathbb R$, the pull-back bundle
\mar{kk156}\beq
Q_\vs=\vs^*Q\to\mathbb R \label{kk156}
\eeq
is a subbundle $i_\vs: Q_\vs\to Q$ of a fibre bundle $Q\to\mathbb
R$ which is a configuration space of a mechanical system with a
fixed  parameter function $\vs(t)$.

In order to obtain the Lagrange and Hamilton equations, we treat
parameters on the same level as dynamic variables. The
corresponding total velocity and phase spaces are the first order
jet manifold $J^1Q$ and the vertical cotangent bundle $V^*Q$ of
the configuration bundle $Q\to\mathbb R$, respectively.

Section 6.2 addresses quantization of mechanical systems with
time-dependent parameters. Since parameters remain classical, a
phase space, that we quantize, is the vertical cotangent bundle
$V_\Si^*Q$ of a fibre bundle $Q\to\Si$. We apply to
$V_\Si^*Q\to\Si$ the technique of leafwise geometric quantization
(Section 5.2).

Berry's phase factor is a phenomenon peculiar to quantum systems
depending on classical time-dependent parameters
\cite{anan,bohm03,kir,montg}. It is described by driving a carrier
Hilbert space of a Hamilton operator over a parameter manifold.
Berry's phase factor depending only on the geometry of a path in a
parameter manifold is called geometric (Section 6.3). It is
characterized by a holonomy operator. A problem lies in separation
of a geometric phase factor from the total evolution operator
without using an adiabatic assumption.

In Section 6.4, we address the Berry phase phenomena in completely
integrable systems. The reason is that, being constant under an
internal dynamic evolution, action variables of a completely
integrable system are driven only by a perturbation holonomy
operator without any adiabatic approximation
\cite{jmp04,book05,book10}.

\subsection{Lagrangian and Hamiltonian mechanics with parameters}

Let the composite bundle (\ref{z500}), treated as a configuration
space of a mechanical system with parameters, be equipped with
bundle coordinates $(t,\si^m,q^i)$ where $(t,\si^m)$ are
coordinates on a fibre bundle $\Si\to\mathbb R$.

\begin{remark} Though $Q\to\mathbb R$ is a trivial
bundle, a fibre bundle $Q\to\Si$ need not be trivial.
\end{remark}

For a time, it is convenient to regard parameters as dynamic
variables. Then a total velocity space of a mechanical system with
parameters is the first order jet manifold $J^1Q$ of the fibre
bundle $Q\to\mathbb R$. It is equipped with the adapted
coordinates $(t,\si^m,q^i,\si^m_t,q^i_t)$.

Let a fibre bundle $Q\to\Si$ be provided with a connection
\mar{z501}\beq
A_\Si=dt\ot(\dr_t+ A^i_t\dr_i) + d\si^m\ot(\dr_m+A^i_m\dr_i).
\label{z501}
\eeq
Then the corresponding vertical covariant differential
(\ref{7.10}):
\mar{z502}\beq
\wt D: J^1Q\to V_\Si Q, \qquad \wt D= (q^i_t- A^i_t
-A^i_m\si^m_t)\dr_i, \label{z502}
  \eeq
is defined on a configuration bundle $Q\to\mathbb R$.

Given a section $\vs$ of a parameter bundle $\Si\to \mathbb R$,
the restriction of $\wt D$ to $J^1i_\vs(J^1Q_\vs)\subset J^1Q$ is
the familiar covariant differential on a fibre bundle $Q_\vs$
(\ref{kk156}) corresponding to the pull-back (\ref{mos83}):
\mar{z503}\beq
A_\vs=\dr_t+[(A^i_m\circ \vs)\dr_t \vs^m +(A\circ \vs)^i_t]\dr_i,
\label{z503}
\eeq
of the connection $A_\Si$ (\ref{z501}) onto $Q_\vs\to\mathbb R$.
Therefore, one can use the vertical covariant differential $\wt D$
(\ref{z502}) in order to construct a Lagrangian for a mechanical
system with parameters on the configuration space $Q$
(\ref{z500}).

We suppose that such a Lagrangian $L$ depends on derivatives of
parameters $\si^m_t$ only via the vertical covariant differential
$\wt D$ (\ref{z502}), i.e.,
\mar{z520}\beq
L=\cL(t,\si^m,q^i,\wt D^i=q^i_t- A^i_t -A^i_m\si^m_t)dt.
\label{z520}
\eeq
Obviously, this Lagrangian is non-regular because of the
Lagrangian constraint
\be
\dr_m^t\cL+A^i_m\dr_i^t\cL=0.
\ee
As a consequence, the corresponding Lagrange equation
\mar{z511,12}\ben
&& (\dr_i-d_t\dr_i^t)\cL=0, \label{z511}\\
&& (\dr_m-d_t\dr_m^t)\cL=0 \label{z512}
\een
is overdefined, and it admits a solution only if a rather
particular relation
\be
(\dr_m +A^i_m\dr_i)\cL +\dr_i^t\cL d_tA^i_m=0
\ee
is satisfied.

However, if a parameter function $\vs$ holds fixed, the equation
(\ref{z512}) is replaced with the condition
\mar{kk400}\beq
\si^m=\vs^m(t), \label{kk400}
\eeq
and the Lagrange equation (\ref{z511}) only should be considered
One can think of this equation under the condition (\ref{kk400})
as being the Lagrange equation for the Lagrangian
\mar{kk158}\beq
L_\vs=J^1\vs^*L =\cL(t,\vs^m,q^i,\wt D^i=q^i_t- A^i_t
-A^i_m\dr_t\vs^m)dt\label{kk158}
\eeq
on a velocity space $J^1Q_\vs$.

A total phase space of a mechanical system with time-dependent
parameters on the composite bundle (\ref{z500}) is the vertical
cotangent bundle $V^*Q$ of $Q\to\mathbb R$. It is coordinated by
$(t,\si^m,q^i,p_m,p_i)$.

Let us consider Hamiltonian forms on a phase space $V^*Q$ which
are associated with the Lagrangian $L$ (\ref{z520}). The
Lagrangian constraint space $N_L\subset V^*Q$ defined by this
Lagrangian is given by the equalities
\mar{kk160}\beq
p_i=\dr_i^t\cL, \qquad p_m+A^i_m p_i=0, \label{kk160}
\eeq
where $A_\Si$ is the connection (\ref{z501}) on a fibre bundle
$Q\to\Si$.

Let
\mar{kk157}\beq
\G=\dr_t +\G^m (t,\si^r)\dr_m \label{kk157}
\eeq
be some connection on a parameter bundle $\Si\to\mathbb R$, and
let
\mar{kk180}\beq
\g=\dr_t +\G^m\dr_m +( A^i_t + A^i_m\G^m)\dr_i \label{kk180}
\eeq
be the composite connection (\ref{b1.114}) on a fibre bundle $Q\to
\mathbb R$ which is defined by the connection $A_\Si$ (\ref{z501})
on $Q\to\Si$ and the connection $\G$ (\ref{kk157}) on
$\Si\to\mathbb R$. Then a desired $L$-associated Hamiltonian form
reads
\mar{7.12}\ben
&& H=(p_md\si^m+p_idq^i)- \label{7.12}\\
&& \qquad [p_m\G^m +p_i (A^i_t+ A^i_m\G^m)
+\cE_\g(t,\si^m,q^i,p_i)]dt,\nonumber
\een
where a Hamiltonian function $\cE_\g$ satisfies the relations
\mar{kk162,3}\ben
&&\dr_i^t\cL(t,\si^m,q^i,\wt
D^i=\dr^i\cE_\g(t,\si^m,q^i,\dr_i^t\cL)=\dr_i^t\cL,\label{kk162}\\
&& p_i\dr^i\cE_\g-\cE_\g=\cL(t,\si^m,q^i,\wt
D^i=\dr^i\cE_\g). \label{kk163}
\een
A key point is that the Hamiltonian form (\ref{7.12}) is affine in
momenta $p_m$ and that the relations (\ref{kk162}) --
(\ref{kk163}) are independent of the connection $\G$
(\ref{kk157}).

The Hamilton equation (\ref{z20a}) -- (\ref{z20b}) for the
Hamiltonian form $H$ (\ref{7.12}) reads
\mar{7.13}\ben
&&q^i_t= A^i_t+A^i_m\G^m +\dr^i \cE_\g, \label{7.13a} \\
&&p_{ti}=-p_j(\dr_i A^j_t +\dr_iA^j_m\G^m)-\dr_i \cE_\g,
\label{7.13b} \\
&&\si^m_t=\G^m, \label{7.13c}\\
&&p_{tm}= -p_i(\dr_m A^i_t+\G^n\dr_mA^i_n) -\dr_m \cE_\g, \label{7.13d}
\een
whereas the Lagrangian constraint (\ref{kk160}) takes the form
\mar{z523,18}\ben
&& p_i=\dr_i^t\cL(t,q^i,\si^m,\dr^i\cE_\g(t,\si^m,q^i,p_i)), \label{z523}\\
&& p_m+A^i_mp_i=0. \label{z518}
\een
If a parameter function $\vs(t)$ holds fixed, we ignore the
equation (\ref{7.13d}) and  treat the rest ones as follows.

Given $\vs(t)$, the equations (\ref{kk400}) and (\ref{z518})
define a subbundle
\mar{kk402}\beq
P_\vs\to Q_\vs\to \mathbb R \label{kk402}
\eeq
over $\mathbb R$ of  a total phase space $V^*Q\to\mathbb R$. With
the connection (\ref{z501}), we have the splitting (\ref{46b})  of
$V^*Q$ which reads
\be
&& V^*Q=A_\Si(V^*_\Si Q)\op\oplus_Q(Q\op\times_QV^*\Si),\\
&& p_i\ol dq^i +p_m\ol d\si^m= p_i(\ol dq^i-A^i_m\ol
d\si^m) +(p_m +A^i_m p_i)\ol d\si^m,
\ee
where $V^*_\Si Q$ is the vertical cotangent bundle of $Q\to\Si$.
Then $V^*Q\to Q$ can be provided with the bundle coordinates
\be
\ol p_i=p_i, \qquad \ol p_m = p_m +A^i_mp_i
\ee
compatible with this splitting. Relative to these coordinates, the
equation (\ref{z518}) takes the form $\ol p_m=0$. It follows that
the subbundle
\mar{kk410}\beq
i_P:P_\vs=i_\vs^*(A_\Si(V^*_\Si Q))\to V^*Q, \label{kk410}
\eeq
coordinated by $(t,q^i,p_i)$, is isomorphic to the vertical
cotangent bundle
\be
V^*Q_\vs=i_\vs^*V_\Si^*Q
\ee
of the configuration space $Q_\vs\to\mathbb R$ (\ref{kk156}) of a
mechanical system with a parameter function $\vs(t)$.
Consequently, the fibre bundle $P_\vs$ (\ref{kk402}) is a phase
space of this system.

Given a parameter function $\vs$, there exists a connection $\G$
on a parameter bundle $\Si\to\mathbb R$ such that $\vs(t)$ is its
integral section, i.e., the equation (\ref{7.13c}) takes the form
\mar{kk404}\beq
\dr_t\vs^m(t)=\G^m(t,\vs(t)). \label{kk404}
\eeq
Then a system of equations (\ref{7.13a}), (\ref{7.13b}) and
(\ref{z523}) under the conditions (\ref{kk400}) and (\ref{kk404})
describes a mechanical system with a given parameter function
$\vs(t)$ on a phase space $P_\vs$. Moreover, this system is the
Hamilton equation for the pull-back Hamiltonian form
\mar{kk193}\beq
H_\vs=i^*_PH=p_idq^i-[p_i(A^i_t +A^i_m\dr_t \vs^m) +\vs^*\cE_\g]dt
\label{kk193}
\eeq
on $P_\vs$ where
\be
A^i_t +A^i_m\dr_t \vs^m=(i_\vs^*\g)^i_t
\ee
is the pull-back connection (\ref{mos83}) on $Q_\vs\to\mathbb R$.

It is readily observed that the Hamiltonian form $H_\vs$
(\ref{kk193}) is associated with the Lagrangian $L_\vs$
(\ref{kk158}) on $J^1Q_\vs$, and the equations (\ref{7.13a}),
(\ref{7.13b}) and (\ref{z523}) are corresponded to the Lagrange
equation (\ref{z511}).

\subsection{Quantum mechanics with classical parameters}

This Section is devoted to quantization of mechanical systems with
time-dependent parameters on the composite bundle $Q$
(\ref{z500}). Since parameters remain classical, a phase space
that we quantize is the vertical cotangent bundle $V_\Si^*Q$ of a
fibre bundle $Q\to\Si$. This phase space is equipped with
holonomic coordinates $(t,\si^m,q^i,p_i)$. It is provided with the
following canonical Poisson structure. Let $T^*Q$ be the cotangent
bundle of $Q$ equipped with the holonomic coordinates
$(t,\si^m,q^i,p_0,p_m,p_i)$. It is endowed with the canonical
Poisson structure $\{,\}_T$ (\ref{m116}). There is the canonical
fibration
\mar{lmp44}\beq
\zeta_\Si: T^*Q\ar^\zeta V^*Q\ar V_\Si^*Q \label{lmp44}
\eeq
(see the exact sequence (\ref{63b})). Then the Poisson bracket
$\{,\}_\Si$ on the space $C^\infty(V_\Si^*Q)$ of smooth real
functions on $V_\Si^*Q$ is defined by the relation
\mar{m72x',2x}\ben
&& \zeta_\Si^*\{f,f'\}_\Si=\{\zeta_\Si^*f,\zeta_\Si^*f'\}_T, \qquad  \label{m72x'} \\
&& \{f,f'\}_\Si = \dr^kf\dr_kf'-\dr_kf\dr^kf', \qquad
f,f'\in C^\infty(V_\Si^*Q). \label{m72x}
\een
The corresponding characteristic symplectic foliation $\cF$
coincides with the fibration $V_\Si^*Q\to\Si$. Therefore, we can
apply to a phase space $V_\Si^*Q\to\Si$ the technique of leafwise
geometric quantization \cite{jmp02,book10}.

Let us assume that a manifold $Q$ is oriented, that fibres of
$V_\Si^*Q\to\Si$ are simply connected, and that
\be
H^2(Q;\mathbb Z_2)= H^2(V_\Si^*Q;\mathbb Z_2)=0.
\ee
Being the characteristic symplectic foliation of the Poisson
structure (\ref{m72x}), the fibration $V_\Si^*Q\to \Si$ is endowed
with the symplectic leafwise form (\ref{spr903}):
\be
\Om_\cF=\wt dp_i\w\wt dq^i.
\ee
Since this form is $\wt d$-exact, its leafwise de Rham cohomology
class equals zero and, consequently, it is the image of the zero
de Rham cohomology class. Then, in accordance with Proposition
\ref{lmp66}, the symplectic foliation $(V_\Si^*Q\to\Si,\Om_\cF)$
admits prequantization.

Since the leafwise form $\Om_\cF$ is $\wt d$-exact, the
prequantization bundle $C\to V_\Si^*Q$ is trivial. Let its
trivialization
\mar{kk200}\beq
C=V_\Si^*Q\times \mathbb C \label{kk200}
\eeq
hold fixed, and let $(t,\si^m,q^k,p_k,c)$ be the corresponding
bundle coordinates. Then $C\to V_\Si^*Q$ admits a leafwise
connection
\be
A_\cF=\wt dp_k\ot\dr^k + \wt dq^k\ot(\dr_k-ip_kc\dr_c).
\ee
This connection preserves the Hermitian fibre metric $g$
(\ref{gm513}) in $C$, and its curvature fulfils the
prequantization condition (\ref{lmp61}):
\be
\wt R=-i\Om_\cF\ot u_C.
\ee
The corresponding prequantization operators (\ref{lqq46}) read
\be
&& \wh f=-i\vt_f +(p_k\dr^kf-f), \qquad  f\in
C^\infty(V_\Si^*Q),\\
&& \vt_f=\dr^kf\dr_k-\dr_kf\dr^k.
\ee

Let us choose the canonical vertical polarization of the
symplectic foliation $(V_\Si^*Q\to\Si,\Om_\cF)$ which is the
vertical tangent bundle $\bT=VV_\Si^*Q$ of a fibre bundle
\be
\pi_{VQ}:V_\Si^*Q\to Q.
\ee
It is readily observed that the corresponding quantum algebra
$\cA_\cF$ consists of functions
\mar{kk191}\beq
f=a^i(t,\si^m,q^k)p_i + b(t,\si^m,q^k) \label{kk191}
\eeq
on $V_\Si^*Q$ which are affine in momenta $p_k$.

Following the quantization procedure in Section 5.2, one should
consider the quantization bundle (\ref{0241}) which is isomorphic
to the prequantization bundle $C$ (\ref{kk200}) because the
metalinear bundle $\cD_{1/2}[\cF$] of complex  fibrewise
half-densities on $V_\Si^*Q\to \Si$ is trivial owing to the
identity transition functions $J_\cF=1$ (\ref{kk203}). Then we
define the representation (\ref{lmp84}) of the quantum algebra
$\cA_\cF$ of functions $f$ (\ref{kk191}) in the space $\gE_\cF$ of
sections $\rho$ of the prequantization bundle $C\to V_\Si^*Q$
which obey the condition (\ref{lmp86}) and whose restriction to
each fibre of $V_\Si^*Q\to \Si$ is of compact support. Since the
trivialization (\ref{kk200}) of $C$ holds fixed, its sections are
complex functions on $V_\Si^*Q$, and the above mentioned condition
(\ref{lmp86}) reads
\be
\dr_kf\dr^k\rho=0, \qquad  f\in C^\infty(Q),
\ee
i.e., elements of $\gE_\cF$ are constant on fibres of $V_\Si^*Q\to
Q$. Consequently, $\gE_\cF$ reduces to zero $\rho=0$.

Therefore, we modify the leafwise quantization procedure as
follows. Given a fibration
\be
\pi_{Q\Si}:Q\to\Si,
\ee
let us consider the corresponding metalinear bundle
$\cD_{1/2}[\pi_{Q\Si}]\to Q$ of leafwise half-densities  on
$Q\to\Si$ and the tensor product
\be
Y_Q=C_Q\ot\cD_{1/2}[\pi_{Q\Si}]=\cD_{1/2}[\pi_{Q\Si}]\to Q,
\ee
where $C_Q=\mathbb C\times Q$ is the trivial complex line bundle
over $Q$. It is readily observed that the Hamiltonian vector
fields
\be
\vt_f=a^k\dr_k-(p_j\dr_ka^j +\dr_kb)\dr^k
\ee
of elements $f\in\cA_\cF$ (\ref{kk191}) are projectable onto $Q$.
Then one can associate to each element $f$ of the quantum algebra
$\cA_\cF$ the first order differential operator
\mar{lmp135}\ben
&&\wh f=(-i\nabla_{\pi_{VQ}(\vt_f)}
+f)\ot\id+\id\ot\bL_{\pi_{VQ}(\vt_f)}= \label{lmp135}\\
&& \qquad -ia^k\dr_k-\frac{i}{2}\dr_ka^k-b \nonumber
\een
in the space $\gE_Q$ of sections of the fibre bundle $Y_Q\to Q$
whose restriction to each fibre of $Q\to\Si$ is of compact
support. Since the pull-back of $\cD_{1/2}[\pi_{Q\Si}]$ onto each
fibre $Q_\si$ of $Q\to\Si$ is the metalinear bundle of
half-densities on $Q_\si$, the restrictions $\rho_\si$ of elements
of $\rho\in \gE_Q$ to $Q_\si$ constitute a pre-Hilbert space with
respect to the non-degenerate Hermitian form
\be
\lng \rho_\si|\rho'_\si\rng_\si=\op\int_{Q_\si}
\rho_\si\ol{\rho'_\si}.
\ee
Then the Schr\"odinger operators (\ref{lmp135}) are Hermitian
operators in the pre-Hilbert $C^\infty(\Si)$-module $\gE_Q$, and
provide the desired geometric quantization of the symplectic
foliation $(V_\Si^*Q\to\Si,\Om_\cF)$.

In order to quantize the evolution equation of a mechanical system
on a phase space $V_\Si^*Q$, one should bear in mind that this
equation is not reduced to the Poisson bracket $\{,\}_\Si$ on
$V_\Si^*Q$, but is expressed in the Poisson bracket $\{,\}_T$ on
the cotangent bundle $T^*Q$ \cite{jmp02,book10}. Therefore, let us
start with the classical evolution equation.

Given the Hamiltonian form $H$ (\ref{7.12}) on a total phase space
$V^*Q$, let $(T^*Q,\cH^*)$ be an equivalent homogeneous
Hamiltonian system with the homogeneous Hamiltonian $\cH^*$
(\ref{mm16}):
\mar{kk190}\beq
\cH^*=p_0 + [p_m\G^m+ p_i (A^i_t+ A^i_m\G^m)
+\cE_\g(t,\si^m,q^i,p_i)].\label{kk190}
\eeq
Let us consider the homogeneous evolution equation (\ref{ws525})
where $F$ are functions on a phase space $V_\Si^*Q$. It reads
\mar{pr35}\ben
&& \{\cH^*,\zeta_\Si^*F\}_T=0, \qquad F\in C^\infty(V_\Si^*Q),
\label{pr35} \\
&& \dr_tF +\G^m\dr_m F + (A^i_t+A^i_m\G^m +\dr^i \cE_\g)\dr_iF
-\nonumber\\
&& \qquad [p_j(\dr_i A^j_t +\dr_iA^j_m\G^m)+\dr_i \cE_\g]\dr^iF=0.
\nonumber
\een
It is readily observed that a function $F\in C^\infty(V_\Si^*Q)$
obeys the equality (\ref{pr35}) iff it is constant on solutions of
the Hamilton equation (\ref{7.13a}) -- (\ref{7.13c}). Therefore,
one can think of the relation (\ref{pr35}) as being a classical
evolution equation on $C^\infty(V_\Si^*Q)$.

In order to quantize the evolution equation (\ref{pr35}), one
should quantize a symplectic manifold $(T^*Q,\{,\}_T)$ so that its
quantum algebra $\cA_T$ contains the pull-back
$\zeta_\Si^*\cA_\cF$ of the quantum algebra $\cA_\cF$ of the
functions (\ref{kk191}). For this purpose, we choose the vertical
polarization $VT^*Q$ on the cotangent bundle $T^*Q$. The
corresponding quantum algebra $\cA_T$ consists of functions on
$T^*Q$ which are affine in momenta $(p_0,p_m,p_i)$ (see Section
5.2). Clearly, $\zeta_\Si^*\cA_\cF$ is a subalgebra of the quantum
algebra $\cA_T$ of $T^*Q$.

Let us restrict our consideration to the subalgebra $\cA'_T\subset
\cA_T$ of functions
\be
f=a(t,\si^r)p_0 +a^m(t,\si^r)p_m + a^i(t,\si^m,q^j)p_i +
b(t,\si^m,q^j),
\ee
where $a$ and $a^\la$ are the pull-back onto $T^*Q$ of functions
on a parameter space $\Si$. Of course, $\zeta_\Si^*\cA_\cF\subset
\cA'_T$. Moreover, $\cA'_T$ admits a representation by the
Hermitian operators
\mar{kk235}\beq
\wh f= -i(a\dr_t + a^m\dr_m +a^i\dr_i)-\frac{i}{2}\dr_ka^k-b
\label{kk235}
\eeq
in the carrier space $\gE_Q$ of the representation (\ref{lmp135})
of $\cA_\cF$. Then, if $\cH^*\in \cA'_T$, the evolution equation
(\ref{pr35}) is quantized as the Heisenberg equation
\mar{lmp50}\beq
i[\wh \cH^*,\wh f]=0, \qquad f\in \cA_\cF. \label{lmp50}
\eeq

A problem is that the function $\cH^*$ (\ref{kk190}) fails to
belong to the algebra $\cA'_T$, unless the Hamiltonian function
$\cE_\g$ (\ref{7.12}) is affine in momenta $p_i$. Let us assume
that $\cE_\g$ is polynomial in momenta. This is the case of almost
all physically relevant models.

\begin{lemma} \label{lmp70x} \mar{lmp70x}
Any smooth function $f$ on $V_\Si^*Q$ which is a polynomial of
momenta $p_k$ is decomposed in a finite sum of products of
elements of the algebra $\cA_\cF$.
\end{lemma}

By virtue of Lemma \ref{lmp70x}, one can associate to a polynomial
Hamiltonian function $\cE_\g$ an element of the enveloping algebra
$\ol\cA_\cF$ of the Lie algebra $\cA_\cF$ (though it by no means
is unique). Accordingly, the homogeneous Hamiltonian $\cH^*$
(\ref{kk190}) is represented by an element of the enveloping
algebra $\ol\cA'_T$ of the Lie algebra $\cA'_T$. Then the
Schr\"odinger representation (\ref{lmp135}) and (\ref{kk235}) of
the Lie algebras $\cA_\cF$ and $\cA'_T$ is naturally extended to
their enveloping algebras $\ol\cA_\cF$ and $\ol\cA'_T$ that
provides quantization
\mar{kk240}\beq
\wh\cH^*= -i[\dr_t +\G^m\dr_m + (A^k_t+  A^k_m\G^m)\dr_k]-
\frac{i}2\dr_k(A^k_t+ A^k_m\G^m) +\wh\cE_\g \label{kk240}
\eeq
of the homogeneous Hamiltonian $\cH^*$ (\ref{kk190}).

It is readily observed that the operator $i\wh\cH^*$ (\ref{kk240})
obeys the Leibniz rule
\mar{kk230'}\beq
i\wh\cH^*(r\rho)=\dr_tr\rho +r(i\wh\cH^*\rho), \qquad r\in
C^\infty(\mathbb R), \qquad \rho\in \gE_Q. \label{kk230'}
\eeq
Therefore, it is a connection on pre-Hilbert $C^\infty(\mathbb
R)$-module $\gE_Q$. The corresponding Schr\"odinger equation reads
\be
i\wh\cH^*\rho=0, \qquad \rho\in \gE_Q.
\ee
Given a trivialization
\mar{kk431}\beq
Q=\mathbb R\times M, \label{kk431}
\eeq
there is the corresponding global decomposition
\be
\wh\cH^*=-i\dr_t +\wh\cH,
\ee
where $\wh\cH$ plays a role of the Hamilton operator. Then we can
introduce the evolution operator $U$ which obeys the equation
\be
\dr_t U(t)=-i\wh\cH^*\circ U(t), \qquad U(0)=\bb.
\ee
It can be written as the formal time-ordered exponent
\be
U=T\exp\left[-i\op\int^t_0\wh\cH dt'\right].
\ee

Given the quantum operator $\wh\cH^*$ (\ref{kk240}), the bracket
\mar{qq120x}\beq
\nabla\wh f= i[\wh \cH^*,\wh f] \label{qq120x}
\eeq
defines a derivation of the quantum algebra $\ol\cA_\cF$. Since
$\wh p_0=-i\dr_t$, the derivation (\ref{qq120x}) obeys the Leibniz
rule
\be
\nabla (r\wh f)=\dr_t r\wh f + r\nabla \wh f, \qquad r\in
C^\infty(\mathbb R).
\ee
Therefore, it is a connection on the $C^\infty(\mathbb R)$-algebra
$\ol\cA_\cF$, which enables one to treat quantum evolution of
$\ol\cA_\cF$ as a parallel displacement along time. In particular,
$\wh f$ is parallel with respect to the connection (\ref{qq120x})
if it obeys the Heisenberg equation (\ref{lmp50}).

Now let us consider a mechanical system depending on a given
parameter function $\vs:\mathbb R\to\Si$. Its configuration space
is the pull-back bundle $Q_\vs$ (\ref{kk156}). The corresponding
phase space is the fibre bundle $P_\vs$ (\ref{kk410}). The
pull-back $H_\vs$ of the Hamiltonian form $H$ (\ref{7.12}) onto
$P_\vs$ takes the form (\ref{kk193}).

The homogeneous phase space of a mechanical system with a
parameter function $\vs$ is the pull-back
\mar{kk411}\beq
\ol P_\vs=i_P^*T^*Q \label{kk411}
\eeq
onto $P_\vs$ of the fibre bundle $T^*Q\to V^*Q$ (\ref{z11'}). The
homogeneous phase space $\ol P_\vs$ (\ref{kk411}) is coordinated
by $(t,q^i,p_0,p_i)$, and it isomorphic to the cotangent bundle
$T^*Q_\vs$. The associated homogeneous Hamiltonian on $\ol P_\vs$
reads
\mar{kk246}\beq
\cH^*_\vs=p_0+ [p_i(A^i_t +A^i_m\dr_t \vs^m) +\vs^*\cE_\g].
\eeq
It characterizes the dynamics of a mechanical system with a given
parameter function $\vs$.

In order to quantize this system, let us consider the pull-back
bundle
\be
\cD_{1/2}[Q_\vs]=i_\vs^*\cD_{1/2}[\pi_{Q\Si}]
\ee
over $Q_\vs$ and its pull-back sections $\rho_\vs=i_\vs^*\rho$,
$\rho\in \gE_Q$. It is easily justified that these are fibrewise
half-densities on a fibre bundle $Q_\vs\to\mathbb R$ whose
restrictions to each fibre $i_t:Q_t\to Q_\vs$ are of compact
support. These sections constitute a pre-Hilbert $C^\infty(\mathbb
R)$-module $\gE_\vs$ with respect to the Hermitian forms
\be
\lng i^*_t\rho_\vs|i^*_t\rho'_\vs\rng_t=\op\int_{Q_t}
i^*_t\rho_\vs\ol{i^*_t\rho'_\vs}.
\ee
Then the pull-back operators
\be
&& (\vs^*\wh f)\rho_\vs=(\wh f\rho)_\vs, \\
&& \vs^*\wh f= -ia^k(t,\vs^m(t),q^j)\dr_k -\frac{i}{2}\dr_ka^k(t,\vs^m(t),q^j)
- b(t,\vs^m(t),q^j),
\ee
in $\gE_\vs$ provide the representation of the pull-back functions
\be
i_\vs^*f=a^k(t,\vs^m(t),q^j)p_k + b(t,\vs^m(t),q^j), \qquad f\in
\cA_\cF,
\ee
on $V^*Q_\vs$. Accordingly, the quantum operator
\mar{kk241}\beq
\wh\cH_\vs^*=-i\dr_t -i(A^i_t +A^i_m\dr_t \vs^m)\dr_i
-\frac{i}2\dr_i(A^i_t +A^i_m\dr_t \vs^m) -\wh{\vs^*\cE_\g}
\label{kk241}
\eeq
coincides with the pull-back operator $\vs^*\wh \cH^*$, and it
yields the Heisenberg equation
\be
i[\wh\cH^*_\vs,\vs^*\wh f]=0
\ee
of a quantum system with a parameter function $\vs$.

The operator $\wh\cH_\vs^*$ (\ref{kk241}) acting in the
pre-Hilbert $C^\infty(\mathbb R)$-module $\gE_\vs$ obeys the
Leibniz rule
\mar{kk230x}\beq
i\wh\cH^*_\vs(r\rho_\vs)=\dr_tr\rho_\vs +r(i\wh\cH^*_\vs\rho_\vs),
\qquad r\in C^\infty(\mathbb R), \qquad \rho_\vs\in \gE_Q,
\label{kk230x}
\eeq
and, therefore, it is a connection on $\gE_\vs$. The corresponding
Schr\"odinger equation reads
\mar{ggg}\ben
&& i\wh\cH^*_\vs\rho_\vs=0, \qquad
\rho_\vs\in \gE_\vs, \label{ggg} \\
&&\left[\dr_t +(A^i_t +A^i_m\dr_t \vs^m)\dr_i + \frac12\dr_i(A^i_t
+A^i_m\dr_t \vs^m) -i\wh{\vs^*\cE_\g}\right]\rho_\vs=0.\nonumber
\een
With the trivialization (\ref{kk431}) of $Q$, we have a
trivialization of $Q_\vs\to\mathbb R$ and the corresponding global
decomposition
\be
\wh\cH^*_\vs=-i\dr_t +\wh\cH_\vs,
\ee
where
\mar{kk430}\beq
\wh\cH_\vs= -i(A^i_t +A^i_m\dr_t \vs^m)\dr_i -
\frac{i}2\dr_i(A^i_t +A^i_m\dr_t \vs^m) +\wh{\vs^*\cE_\g}
\label{kk430}
\eeq
is a Hamilton operator. Then we can introduce an evolution
operator $U_\vs$ which obeys the equation
\be
\dr_t U_\vs(t)=-i\wh\cH^*_\vs\circ U_\vs(t), \qquad U_\vs(0)=\bb.
\ee
It can be written as the formal time-ordered exponent
\mar{lmp60x}\beq
U_\vs(t)=T\exp\left[-i\op\int^t_0 \wh\cH_\vs dt'\right].
\label{lmp60x}
\eeq

\subsection{Berry geometric factor}

As was mentioned above, the Berry phase factor is a standard
attribute of quantum mechanical systems with time-dependent
classical parameters \cite{bohm03,book00}. The quantum Berry phase
factor is described by driving a carrier Hilbert space of a
Hamilton operator over cycles in a parameter manifold. The Berry
geometric factor depends only on the geometry of a path in a
parameter manifold and, therefore, provides a possibility to
perform quantum gate operations in an intrinsically fault-tolerant
way. A problem lies in separation of the Berry geometric factor
from the total evolution operator without using an adiabatic
assumption. Firstly, holonomy quantum computation implies exact
cyclic evolution, but exact adiabatic cyclic evolution almost
never exists. Secondly, an adiabatic condition requires that the
evolution time must be long enough.

In a general setting, let us consider a linear (not necessarily
finite-dimensional) dynamical system
\be
\dr_t\psi=\wh S \psi
\ee
whose linear (time-dependent) dynamic operator $\wh S$ falls into
the sum
\mar{r0}\beq
\wh S=\wh S_0 + \Delta=\wh S_0 +\dr_t\vs^m \Delta_m, \label{r0}
\eeq
where $\vs(t)$ is a parameter function given by a section of some
smooth fibre bundle $\Si\to\mathbb R$ coordinated by $(t,\si^m)$.
Let assume the following:

(i) the operators $\wh S_0(t)$ and $\Delta(t')$ commute for all
instants $t$ and $t'$,

(ii) the operator $\Delta$ depends on time only through a
parameter function $\vs(t)$.

\noindent Then the corresponding evolution operator $U(t)$ can be
represented by the product of time-ordered exponentials
\mar{rr}\beq
U(t)=U_0(t)\circ U_1(t) = T\exp\left[\op\int^t_0 \Delta
dt'\right]\circ T\exp\left[\op\int^t_0 \wh S_0dt'\right],
\label{rr}
\eeq
where the first one is brought into the ordered exponential
\mar{r4}\ben
&& U_1(t) = T\exp\left[\op\int^t_0 \Delta_m(\vs(t'))\dr_t\vs^m(t')
dt'\right] = \label{r4}\\
&& \qquad T\exp\left[\op\int_{\vs[0,t]} \Delta_m(\si)d\si^m\right]
\nonumber
\een
along the curve $\vs[0,t]$ in a parameter bundle $\Si$. It is the
 Berry geometric factor depending only on a trajectory of a
parameter function $\vs$. Therefore, one can think of this factor
as being a displacement operator along a curve $\vs[0,t]\subset
\Si$. Accordingly,
\mar{kk436}\beq
\Delta=\Delta_m\dr_t\vs^m \label{kk436}
\eeq
is called the  holonomy operator.

However, a problem is that the above mentioned commutativity
condition (i) is rather restrictive.

Turn now to the quantum Hamiltonian system with classical
parameters in Section 6.2. The Hamilton operator $\wh\cH_\vs$
(\ref{kk430}) in the evolution operator $U$ (\ref{lmp60x}) takes
the form (\ref{r0}):
\mar{lmp72x}\beq
\wh\cH_\vs= -i\left[A^k_m\dr_k + \frac12\dr_kA^k_m\right]\dr_t
\vs^m +\wh \cH'(\vs). \label{lmp72x}
\eeq
Its second term $\wh\cH'$ can be regarded as a dynamic Hamilton
operator of a quantum system, while the first one is responsible
for the Berry geometric factor as follows.

Bearing in mind possible applications to holonomic quantum
computations, let us simplify the quantum system in question. The
above mentioned trivialization (\ref{kk431}) of $Q$ implies a
trivialization of a parameter bundle $\Si=\mathbb R\times W$ such
that a fibration $Q\to\Si$ reads
\be
\mathbb R\times M\ar^{\id\times \pi_M}\mathbb R\times W,
\ee
where $\pi_M: M\to W$ is a fibre bundle. Let us suppose that
components $A^k_m$ of the connection $A_\Si$ (\ref{z501}) are
independent of time. Then one can regard the second term in this
connection as a connection on a fibre bundle $M\to W$. It also
follows that the first term in the Hamilton operator
(\ref{lmp72x}) depends on time only through parameter functions
$\vs^m(t)$. Furthermore, let the two terms in the  Hamilton
operator (\ref{lmp72x}) mutually commute on $[0,t]$.  Then the
evolution operator $U$ (\ref{lmp60x}) takes the form
\mar{kk432}\ben
&& U= T\exp\left[ -\op\int_{\vs([0,t])}\left(A^k_m\dr_k +
\frac12\dr_kA^k_m\right) d\si^m \right] \circ \label{kk432}\\
&& \qquad T\exp\left[ -i\op\int_0^t\wh\cH' dt'\right]. \nonumber
\een
One can think of its first factor as being the parallel
displacement operator along the curve $\vs([0,t])\subset W$ with
respect to the connection
\mar{kk433}\beq
\nabla_m\rho=\left(\dr_m + A^k_m\dr_k +
\frac12\dr_kA^k_m\right)\rho, \qquad \rho\in \gE_Q, \label{kk433}
\eeq
called the  Berry connection on a $C^\infty(W)$-module $\gE_Q$. A
peculiarity of this factor in comparison with the second one  lies
in the fact that integration over time through a parameter
function $\vs(t)$ depends only on a trajectory of this function in
a parameter space, but not on parametrization of this trajectory
by time. Therefore, the first term of the evolution operator $U$
(\ref{kk432}) is the Berry geometric factor. The corresponding
holonomy operator (\ref{kk436}) reads
\be
\Delta= \left(A^k_m\dr_k + \frac12\dr_kA^k_m\right)\dr_t\vs^m.
\ee

\subsection{Non-adiabatic holonomy operator}

We address the Berry phase phenomena in a completely integrable
system  of $m$ degrees of freedom around its invariant torus
$T^m$. The reason is that, being constant under an internal
evolution, its action variables are driven only by a perturbation
holonomy operator $\Delta$. We construct such an operator for an
arbitrary connection on a fibre bundle
\mar{r25}\beq
W\times T^m\to W, \label{r25}
\eeq
without any adiabatic approximation \cite{jmp04,book05,book10}. In
order that a holonomy operator and a dynamic Hamiltonian mutually
commute, we first define a holonomy operator with respect to
initial data action-angle coordinates and, afterwards, return to
the original ones. A key point is that both classical evolution of
action variables and mean values of quantum action operators
relative to original action-angle coordinates are determined by
the dynamics of initial data action and angle variables.

A generic phase space of a Hamiltonian system with time-dependent
parameters is a composite fibre bundle
\be
P\to\Si\to \mathbb R,
\ee
where $\Pi\to\Si$ is a symplectic bundle (i.e., a symplectic
foliation whose leaves are fibres of $\Pi\to\Si$), and
\be
\Si=\mathbb R\times W\to\mathbb R
\ee
is a parameter bundle whose sections are parameter functions. In
the case of a completely integrable system with time-dependent
parameters, we have the product
\be
P= \Si\times U= \Si\times (V\times T^m)\to \Si\to\mathbb R,
\ee
equipped with the coordinates $(t,\si^\al,I_k,\vf^k)$. Let us
suppose for a time that parameters also are dynamic variables. The
total phase space of such a system is the product
\be
\Pi=V^*\Si\times U
\ee
coordinated by $(t,\si^\al,p_\al=\dot\si_\al,I_k,\vf^k)$. Its
dynamics is characterized by the Hamiltonian form (\ref{7.12}):
\mar{pr30}\ben
&& H_\Si= p_\al d\si^\al+ I_k d\vf^k -
\cH_\Si(t,\si^\bt,p_\bt, I_j,\vf^j)dt,
\nonumber\\
&& \cH_\Si=p_\al\G^\al +I_k(\La^k_t +\La^k_\al\G^\al_t)
+\wt\cH, \label{pr30}
\een
where $\wt\cH$ is a function, $\dr_t +\G^\al\dr_\al$ is the
connection (\ref{kk157}) on the parameter bundle $\Si\to\mathbb
R$, and
\mar{r30}\beq
\La=dt\ot(\dr_t +\La^k_t\dr_k) + d\si^\al\ot(\dr_\al
+\La^k_\al\dr_k) \label{r30}
\eeq
is the connection (\ref{z501}) on the fibre bundle
\be
\Si\times T^m\to\Si.
\ee
Bearing in mind that $\si^\al$ are parameters, one should choose
the Hamiltonian $\cH_\Si$ (\ref{pr30}) to be affine in their
momenta $p_\al$. Then a Hamiltonian system with a fixed parameter
function $\si^\al=\vs^\al(t)$ is described by the pull-back
Hamiltonian form (\ref{kk193}):
\mar{r12x}\ben
&& H_\vs=I_k d\vf^k - \{I_k[\La^k_t(t,\vf^j)+
\label{r12x} \\
&& \qquad \La^k_\al (t,\vs^\bt,\vf^j)\dr_t\vs^\al]
+\wt\cH(t,\vs^\bt,I_j,\vf^j)\}dt \nonumber
\een
on a Poisson manifold
\mar{z46}\beq
\mathbb R\times U=\mathbb R\times (V\times T^m). \label{z46}
\eeq

Let $\wt\cH=\cH(I_i)$ be a Hamiltonian of an original autonomous
completely integrable system on the toroidal domain $U$
(\ref{z46}) equipped with the action-angle coordinates
$(I_k,\vf^k)$. We introduce a desired holonomy operator by the
appropriate choice of the connection $\La$ (\ref{r30}).

For this purpose, let us choose the initial data action-angle
coordinates $(I_k,\f^k)$ by the converse to the canonical
transformation (\ref{ww26}):
\mar{ww26'}\beq
\vf^k=\f^k- t\dr^k\cH. \label{ww26'}
\eeq

With respect to these coordinates, the Hamiltonian of an original
completely integrable system vanishes and the Hamiltonian form
(\ref{r12x}) reads
\mar{r13}\beq
H_\vs=I_k d\f^k - I_k[\La^k_t(t,\f^j)+ \La^k_\al
(t,\vs^\bt,\f^j)\dr_t\vs^\al]dt. \label{r13}
\eeq
Let us put $\La^k_t=0$ by the choice of a reference frame
associated to the initial data coordinates $\f^k$, and let us
assume that coefficients $\La^k_\al$ are independent of time,
i.e., the part
\mar{r31}\beq
\La_W=d\si^\al\ot(\dr_\al +\La^k_\al\dr_k) \label{r31}
\eeq
of the connection $\La$ (\ref{r30}) is a connection on the fibre
bundle (\ref{r25}). Then the Hamiltonian form (\ref{r13}) reads
\mar{r14}\beq
H_\vs=I_k d\f^k - I_k \La^k_\al (\vs^\bt,\f^j)\dr_t\vs^\al dt.
\label{r14}
\eeq
Its Hamilton vector field (\ref{z3}) is
\mar{r21}\beq
\g_H=\dr_t + \La^i_\al \dr_t\vs^\al\dr_i -I_k\dr_i\La^k_\al
\dr_t\vs^\al\dr^i, \label{r21}
\eeq
and it leads to the Hamilton equation
\mar{zz10x,1}\ben
&& d_t\f^i=\La^i_\al(\vs(t),\f^l) \dr_t\vs^\al, \label{zz10x}\\
&& d_t I_i=-I_k\dr_i\La^k_\al(\vs(t),\f^l) \dr_t\vs^\al. \label{zz11}
\een
Let us consider the lift
\mar{r22}\beq
V^*\La_W=d\si^\al\ot(\dr_\al +\La^i_\al\dr_i
-I_k\dr_i\La^k_\al\dr^i) \label{r22}
\eeq
of the connection $\La_W$ (\ref{r31}) onto the fibre bundle
\be
W\times (V\times T^m)\to W,
\ee
seen as a subbundle of the vertical cotangent bundle
\be
V^*(W\times T^m)=W\times T^*T^m
\ee
of the fibre bundle (\ref{r25}). It follows that any solution
$I_i(t)$, $\f^i(t)$ of the Hamilton equation (\ref{zz10x}) --
(\ref{zz11}) (i.e., an integral curve of the Hamilton vector field
(\ref{r21})) is a horizontal lift  of the curve $\vs(t)\subset W$
with respect to the connection $V^*\La_W$ (\ref{r22}), i.e.,
\be
I_i(t)=I_i(\vs(t)), \qquad \f^i(t)=\f^i(\vs(t)).
\ee
Thus, the right-hand side of the Hamilton equation (\ref{zz10x})
-- (\ref{zz11}) is the holonomy operator
\mar{r33}\beq
\Delta= (\La^i_\al \dr_t\vs^\al, -I_k\dr_i\La^k_\al \dr_t\vs^\al).
\label{r33}
\eeq
It is not a linear operator, but the substitution of a solution
$\f(\vs(t))$ of the equation (\ref{zz10x}) into the Hamilton
equation (\ref{zz11}) results in a linear holonomy operator on the
action variables $I_i$.

Let us show that the holonomy operator (\ref{r33}) is well
defined. Since any vector field $\vt$ on $\mathbb R\times T^m$
such that $\vt\rfloor dt=1$ is complete, the Hamilton equation
(\ref{zz10x}) has solutions for any parameter function $\vs(t)$.
It follows that any connection $\La_W$ (\ref{r31}) on the fibre
bundle (\ref{r25}) is an Ehresmann connection, and so is its lift
(\ref{r22}). Because $V^*\La_W$ (\ref{r22}) is an Ehresmann
connection, any curve $\vs([0,1])\subset W$ can play a role of the
parameter function in the holonomy operator $\Delta$ (\ref{r33}).

Now, let us return to the original action-angle coordinates
$(I_k,\vf^k)$ by means of the canonical transformation
(\ref{ww26'}). The perturbed Hamiltonian reads
\be
\cH'= I_k\La^k_\al(\vs(t), \vf^i-t\dr^i\cH(I_j))\dr_t\vs^\al(t) +
\cH(I_j),
\ee
while the Hamilton equation (\ref{zz10x}) -- (\ref{zz11}) takes
the form
\be
&& \dr_t\vf^i=\dr^i\cH(I_j)+ \La^i_\al(\vs(t), \vf^l-t\dr^l\cH(I_j))
\dr_t\vs^\al(t) \\
&& \qquad - tI_k\dr^i\dr^s\cH(I_j)\dr_s\La^k_\al(\vs(t),
\vf^l-t\dr^l\cH(I_j))\dr_t\vs^\al(t),\\
&& \dr_t I_i=-I_k\dr_i\La^k_\al(\vs(t), \vf^l-t\dr^l\cH(I_j)) \dr_t\vs^\al(t).
\ee
Its solution is
\be
I_i(\vs(t)), \qquad \vf^i(t)=\f^i(\vs(t)) +t\dr^i\cH(I_j(\vs(t))),
\ee
 where
$I_i(\vs(t))$, $\f^i(\vs(t))$ is a solution of the Hamilton
equation (\ref{zz10x}) -- (\ref{zz11}). We observe that the action
variables $I_k$ are driven only by the holonomy operator, while
the angle variables $\vf^i$ have a non-geometric summand.

Let us emphasize that, in the construction of the holonomy
operator (\ref{r33}), we do not impose any restriction on the
connection $\La_W$ (\ref{r31}). Therefore, any connection on the
fibre bundle (\ref{r25}) yields a holonomy operator of a
completely integrable system. However, a glance at the expression
(\ref{r33}) shows that this operator becomes zero on action
variables if all coefficients $\La^k_\la$ of the connection
$\La_W$ (\ref{r31}) are constant, i.e., $\La_W$ is a principal
connection on the fibre bundle (\ref{r25}) seen as a principal
bundle with the structure group $T^m$.

In order to quantize a non-autonomous completely integrable system
on the Poisson toroidal domain $(U,\{,\}_V)$ (\ref{z46}) equipped
with action-angle coordinates $(I_i,\vf^i)$, one may follow the
instantwise geometric quantization of non-autonomous mechanics. As
a result, we can simply replace functions on $T^m$ with those on
$\mathbb R\times T^m$ \cite{acang1}. Namely, the corresponding
quantum algebra $\cA\subset C^\infty(U)$ consists of affine
functions
\mar{ww7}\beq
f=a^k(t,\vf^j)I_k + b(t,\vf^j)  \label{ww7}
\eeq
of action coordinates $I_k$ represented by the operators
(\ref{lmp135}) in the space
\mar{r41x}\beq
\gE=\mathbb C^\infty(\mathbb R\times T^m) \label{r41x}
\eeq
of smooth complex functions $\psi(t,\vf)$ on $\mathbb R\times
T^m$. This space is provided with the structure of the pre-Hilbert
$\mathbb C^\infty(\mathbb R)$-module endowed with the
non-degenerate $\mathbb C^\infty(\mathbb R)$-bilinear form
\be
\lng \psi|\psi'\rng=\left(\frac1{2\pi}\right)^m\op\int_{T^m} \psi
\ol \psi'd^m\vf, \qquad \psi,\psi'\in \gE.
\ee
Its basis consists of the pull-back onto $\mathbb R\times T^m$ of
the functions
\mar{r41}\beq
\psi_{(n_r)}=\exp[i(n_r\f^r)], \qquad
(n_r)=(n_1,\ldots,n_m)\in\mathbb Z^m. \label{r41}
\eeq

Furthermore, this quantization of a non-autonomous completely
integrable system on the Poisson manifold $(U,\{,\}_V)$ is
extended to the associated homogeneous completely integrable
system on the symplectic annulus (\ref{0910}):
\be
U'=\zeta^{-1}(U)=N'\times T^m\to N'
\ee
by means of the operator $\wh I_0=-i\dr_t$ in the pre-Hilbert
module $\gE$ (\ref{r41x}). Accordingly, the homogeneous
Hamiltonian $\wh\cH^*$ is quantized as
\be
\wh\cH^*=-i\dr_t +\wh\cH.
\ee
It is a Hamiltonian of a quantum non-autonomous completely
integrable system. The corresponding Schr\"odinger equation is
\mar{r42}\beq
\wh\cH^*\psi=-i\dr_t\psi +\wh\cH\psi=0, \qquad \psi\in  \gE.
\label{r42}
\eeq

For instance, a quantum Hamiltonian of an original autonomous
completely integrable system seen as the non-autonomous one is
\be
\wh\cH^*=-i\dr_t +\cH(\wh I_j).
\ee
Its spectrum
\be
\wh\cH^*\psi_{(n_r)}=E_{(n_r)}\psi_{(n_r)}
\ee
on the basis $\{\psi_{(n_r)}\}$ (\ref{r41}) for $\gE$ (\ref{r41})
coincides with that of the autonomous Hamiltonian
$\wh\cH(I_k)=\cH(\wh I_k)$. The Schr\"odinger equation (\ref{r42})
reads
\be
\wh\cH^*\psi=-i\dr_t\psi +\cH(-i\dr_k +\la_k)\psi=0, \qquad
\psi\in E.
\ee
Its solutions are the Fourier series
\be
\psi=\op\sum_{(n_r)} B_{(n_r)} \exp[-itE_{(n_r)}]\psi_{(n_r)},
\qquad B_{(n_r)}\in\mathbb C.
\ee

Now, let us quantize this completely integrable system with
respect to the initial data action-angle coordinates $(I_i,\f^i)$.
As was mentioned above, it is given on a toroidal domain $U$
(\ref{z46}) provided with another fibration over $\mathbb R$. Its
quantum algebra $\cA_0\subset C^\infty(U)$ consists of affine
functions
\mar{r45}\beq
f=a^k(t,\f^j)I_k + b(t,\f^j). \label{r45}
\eeq
The canonical transformation (\ref{ww26}) ensures an isomorphism
of Poisson algebras $\cA$ and $\cA_0$. Functions $f$ (\ref{r45})
are represented by the operators $\wh f$ (\ref{lmp135}) in the
pre-Hilbert module $\gE_0$ of smooth complex functions
$\Psi(t,\f)$ on $\mathbb R\times T^m$. Given its basis
\be
\Psi_{(n_r)}(\f)=[in_r\f^r],
\ee
the operators $\wh I_k$ and $\wh\psi_{(n_r)}$ take the form
\mar{ci11}\ben
&& \wh I_k\psi_{(n_r)}=(n_k+\la_k)\psi_{(n_r)}, \nonumber\\
&& \psi_{(n_r)} \psi_{(n'_r)}=\psi_{(n_r)}
\psi_{(n'_r)}=\psi_{(n_r+n'_r)}. \label{ci11}
\een
The Hamiltonian of a quantum completely integrable system with
respect to the initial data variables is $\wh\cH^*_0=-i\dr_t$.
Then one easily obtains the isometric isomorphism
\mar{r46}\beq
R(\psi_{(n_r)})=\exp[itE_{(n_r)}]\Psi_{(n_r)}, \qquad \lng
R(\psi)| R(\psi')\rng= \lng \psi | \psi'\rng, \label{r46}
\eeq
of the pre-Hilbert modules $\gE$ and $\gE_0$ which provides the
equivalence
\mar{r47}\beq
\wh I_i=R^{-1}\wh I_iR, \qquad
\wh\psi_{(n_r)}=R^{-1}\wh\Psi_{(n_r)} R, \qquad \wh\cH^*=
R^{-1}\wh\cH^*_0 R \label{r47}
\eeq
of the quantizations of a completely integrable system with
respect to the original and initial data action-angle variables.

In view of the isomorphism (\ref{r47}), let us first construct a
holonomy operator of a quantum completely integrable system
$(\cA_0,\wh\cH^*_0)$ with respect to the initial data action-angle
coordinates. Let us consider the perturbed homogeneous Hamiltonian
\be
\bH_\vs=\cH^*_0+\bH_1=I_0+\dr_t\vs^\al(t)\La_\al^k(\vs(t),\f^j)I_k
\ee
of the classical perturbed completely integrable system
(\ref{r14}). Its perturbation term $\bH_1$ is of the form
(\ref{ww7}) and, therefore, is quantized by the operator
\be
\wh\bH_1=-i\dr_t\vs^\al\wh\Delta_\al=-i\dr_t\vs^\al\left[\La^k_\al\dr_k
+\frac12\dr_k(\La^k_\al) +i\la_k \La^k_\al\right].
\ee

The quantum Hamiltonian $\wh\bH_\vs=\wh\cH^*_0+\wh\bH_1$ defines
the Schr\"odinger equation
\mar{r52}\beq
\dr_t\Psi+ \dr_t\vs^\al\left[\La^k_\al\dr_k
+\frac12\dr_k(\La^k_\al) +i\la_k \La^k_\al\right]\Psi=0.
\label{r52}
\eeq
If its solution exists, it can be written by means of the
evolution operator $U(t)$ which is reduced to the geometric factor
\be
 U_1(t)= T\exp\left[i\op\int_0^t \dr_{t'}\vs^\al(t')
\wh\Delta_\al(t')dt' \right].
\ee
The latter can be viewed as a displacement operator along the
curve $\vs[0,1]\subset W$ with respect to the connection
\mar{r53}\beq
\wh\La_W=d\si^\al(\dr_\al +\wh\Delta_\al) \label{r53}
\eeq
on the $\mathbb C^\infty(W)$-module $\mathbb C^\infty(W\times
T^m)$ of smooth complex functions on $W\times T^m$. Let us study
weather this displacement operator exists.

Given a connection $\La_W$ (\ref{r31}), let $\Phi^i(t,\f)$ denote
the flow of the complete vector field
\be
\dr_t+\La_\al(\vs,\f)\dr_t\vs^\al
\ee
on $\mathbb R\times T^m$. It is a solution of the Hamilton
equation (\ref{zz10x}) with the initial data $\f$. We need the
inverse flow $(\Phi^{-1})^i(t,\f)$ which obeys the equation
\be
&&
\dr_t(\Phi^{-1})^i(t,\f)=-\dr_t\vs^\al\La^i_\al(\vs,(\Phi^{-1})^i(t,\f))=\\
&& \qquad -\dr_t\vs^\al\La^k_\al(\vs,\f)\dr_k(\Phi^{-1})^i(t,\f).
\ee
Let $\Psi_0$ be an arbitrary complex half-form $\Psi_0$ on $T^m$
possessing identical transition functions, and let the same symbol
stand for its pull-back onto $\mathbb R\times T^m$. Given its
pull-back
\mar{r56}\beq
(\Phi^{-1})^*\Psi_0= \det\left(\frac{\dr (\Phi^{-1})^i}{\dr
\f^k}\right)^{1/2}\Psi_0(\Phi^{-1}(t,\f)), \label{r56}
\eeq
it is readily observed that
\mar{r55}\beq
\Psi=(\Phi^{-1})^*\Psi_0\exp[i\la_k\f^k] \label{r55}
\eeq
obeys the Schr\"odinger equation (\ref{r52}) with the initial data
$\Psi_0$. Because of the multiplier $\exp[i\la_k\f^k]$, the
function $\Psi$ (\ref{r55}) however is ill defined, unless all
numbers $\la_k$ equal $0$ or $\pm 1/2$. Let us note that, if some
numbers $\la_k$ are equal to $\pm 1/2$, then
$\Psi_0\exp[i\la_k\f^k]$ is a half-density on $T^m$ whose
transition functions equal $\pm 1$, i.e., it is a section of a
non-trivial metalinear bundle over $T^m$.

Thus, we observe that, if $\la_k$ equal $0$ or $\pm 1/2$, then the
displacement operator always exists and $\Delta=i\bH_1$ is a
holonomy operator. Because of the action law (\ref{ci11}), it is
essentially infinite-dimensional.

For instance, let $\La_W$ (\ref{r31}) be the above mentioned
principal connection, i.e., $\La^k_\al=$const. Then the
Schr\"odinger equation (\ref{r52}) where $\la_k=0$ takes the form
\be
\dr_t\Psi(t,\f^j)+ \dr_t\vs^\al(t)\La^k_\al\dr_k\Psi(t,\f^j)=0,
\ee
and its solution (\ref{r56}) is
\be
\Psi(t,\f^j)=\Psi_0(\f^j-(\vs^\al(t)-\vs^\al(0))\La^j_\al).
\ee
The corresponding evolution operator $U(t)$ reduces to Berry's
phase multiplier
\be
U_1\Psi_{(n_r)}=
\exp[-in_j(\vs^\al(t)-\vs^\al(0))\La^j_\al]\Psi_{(n_r)}, \qquad
n_j\in (n_r).
\ee
It keeps the eigenvectors of the action operators $\wh I_i$.

In order to return to the original action-angle variables, one can
employ the morphism $R$ (\ref{r46}). The corresponding Hamiltonian
reads
\be
\bH=R^{-1}\bH_\vs R.
\ee
The key point is that, due to the relation (\ref{r47}), the action
operators  $\wh I_i$ have the same mean values
\be
\lng I_k\psi| \psi\rng= \lng I_k\Psi | \Psi\rng, \qquad
\Psi=R(\psi),
\ee
with respect both to the original and the initial data
action-angle variables. Therefore, these mean values are defined
only by the holonomy operator.

In conclusion, let us note that, since action variables are driven
only by a holonomy operator, one can use this operator in order to
perform  a dynamic transition between classical solutions or
quantum states of an unperturbed completely integrable system by
an appropriate choice of a parameter function $\vs$. A key point
is that this transition can take an arbitrary short time because
we are entirely free with time parametrization of $\vs$ and can
choose it quickly changing, in contrast with slowly varying
parameter functions in adiabatic models. This fact makes
non-adiabatic holonomy operators in completely integrable systems
promising for several applications, e.g., quantum control and
quantum computation.

\section{Appendix}

For the sake of convenience of the reader, this Section summarizes
the relevant material on differential geometry of fibre bundles
\cite{gre,book00,sard08}.

\subsection{Geometry of fibre bundles}

Throughout this work, all morphisms are smooth, and manifolds are
smooth real and finite-dimensional. A smooth manifold is
customarily assumed to be Hausdorff, second-countable and,
consequently, paracompact. Being paracompact, a smooth manifold
admits a partition of unity by smooth real functions. Unless
otherwise stated, manifolds are assumed to be connected. The
symbol $C^\infty(Z)$ stands for a ring of smooth real functions on
a manifold $Z$.

Given a smooth manifold $Z$, by $\pi_Z:TZ\to Z$ is denoted its
tangent bundle. Given manifold coordinates $(z^\al)$ on $Z$, the
tangent bundle $TZ$ is equipped with the  holonomic coordinates
\be
(z^\la,\dot z^\la), \qquad \dot z'^\la= \frac{\dr z'^\la}{\dr
z^\mu}\dot z^\m,
\ee
with respect to the holonomic frames $\{\dr_\la\}$ in the tangent
spaces to $Z$. Any manifold morphism $f:Z\to Z'$ yields the
tangent morphism
\be
Tf:TZ\to TZ', \qquad \dot z'^\la\circ Tf = \frac{\dr f^\la}{\dr
z^\m}\dot z^\m,
\ee
of their tangent bundles. A morphism $f$ is said to be an
immersion if $T_zf$, $z\in Z$, is injective, and a submersion if
$T_zf$, $z\in Z$, is surjective. Note that a submersion is an open
map (i.e., an image of any open set is open).

If $f:Z\to Z'$ is an injective immersion, its range is called a
submanifold of $Z'$. A submanifold is said to be imbedded if it
also is a topological subspace. In this case, $f$ is called an
imbedding. If $Z\subset Z'$, its natural injection is denoted by
$i_Z:Z\to Z'$.

If a manifold morphism
\mar{11f1}\beq
\p :Y\to X, \qquad \di X=n>0, \label{11f1}
\eeq
is a surjective submersion, one says that: (i) its domain $Y$ is a
 fibred manifold, (ii) $X$ is its  base, (iii) $\p$ is a
fibration, and (iv) $Y_x=\p^{-1}(x)$ is a  fibre over $x\in X$. A
fibred manifold admits an atlas of fibred coordinate charts $(U_Y;
x^\la, y^i)$ such that $(x^\la)$ are coordinates on
$\p(U_Y)\subset X$ and coordinate transition functions read
\be
x'^\la =f^\la(x^\m), \qquad y'^i=f^i(x^\m,y^j).
\ee

For each point $y\in Y$ of a fibred manifold, there exists a local
section $s$ of $Y\to X$ passing through $y$. By a local section of
the fibration (\ref{11f1}) is meant an injection $s:U\to Y$ of an
open subset $U\subset X$ such that $\p\circ s=\id U$, i.e., a
section sends any point $x\in X$ into the fibre $Y_x$ over this
point. A local section also is defined over any subset $N\in X$ as
the restriction to $N$ of a local section over an open set
containing $N$. If $U=X$, one calls $s$ the global section. A
range $s(U)$ of a local section $s:U\to Y$ of a fibred manifold
$Y\to X$ is an imbedded submanifold of $Y$. A local section is a
closed map, sending closed subsets of $U$ onto closed subsets of
$Y$. If  $s$ is a global section, then $s(X)$ is a closed imbedded
submanifold of $Y$. Global sections of a fibred manifold need not
exist.

\begin{theorem} \label{mos9} \mar{mos9}
Let $Y\to X$ be a fibred manifold whose fibres are diffeomorphic
to $\mathbb R^m$.  Any its section over a closed imbedded
submanifold (e.g., a point) of $X$ is extended to a global section
\cite{ste}. In particular, such a fibred manifold always has a
global section.
\end{theorem}

Given fibred coordinates $(U_Y;x^\la,y^i)$, a section $s$ of a
fibred manifold $Y\to X$ is represented by collections of local
functions $\{s^i=y^i\ \circ s\}$ on $\p(U_Y)$.

Morphisms of fibred manifolds, by definition, are  fibrewise
morphisms, sending a fibre to a fibre. Namely, a  fibred morphism
of a fibred manifold $\pi:Y\to X$ to a fibred manifold $\pi':
Y'\to X'$ is defined as a pair $(\Phi,f)$ of manifold morphisms
which form a commutative diagram
\be
\begin{array}{rcccl}
& Y &\ar^\Phi & Y'&\\
_\pi& \put(0,10){\vector(0,-1){20}} & & \put(0,10){\vector(0,-1){20}}&_{\pi'}\\
& X &\ar^f & X'&
\end{array}, \qquad \pi'\circ\Phi=f\circ\pi.
\ee
Fibred injections and surjections are called monomorphisms and
epimorphisms, respectively. A fibred diffeomorphism is called an
isomorphism or an automorphism if it is an isomorphism to itself.
For the sake of brevity, a fibred morphism over $f=\id X$ usually
is said to be a fibred morphism over $X$, and is denoted by
$Y\ar_XY'$. In particular, a fibred automorphism over $X$ is
called a  vertical automorphism.

A fibred manifold $Y\to X$ is said to be  trivial if $Y$ is
isomorphic to the product $X\times V$. Different trivializations
of $Y\to X$ differ from each other in surjections $Y\to V$.

A fibred manifold $Y\to X$ is called a  fibre bundle if it is
locally trivial, i.e., if it admits a fibred coordinate atlas
$\{(\pi^{-1}(U_\xi); x^\la, y^i)\}$ over a cover
$\{\pi^{-1}(U_\xi)\}$ of $Y$ which is the inverse image of a cover
$\gU=\{U_\xi\}$ of $X$. In this case, there exists a manifold $V$,
called a  typical fibre, such that $Y$ is locally diffeomorphic to
the  splittings
\mar{mos02}\beq
\psi_\xi:\pi^{-1}(U_\xi) \to U_\xi\times V, \label{mos02}
\eeq
glued together by means of  transition functions
\mar{mos271}\beq
\vr_{\xi\zeta}=\psi_\xi\circ\psi_\zeta^{-1}: U_\xi\cap
U_\zeta\times V \to  U_\xi\cap U_\zeta\times V \label{mos271}
\eeq
on overlaps $U_\xi\cap U_\zeta$. Restricted to a point $x\in X$,
 trivialization morphisms $\psi_\xi$ (\ref{mos02}) and
transition functions $\vr_{\xi\zeta}$ (\ref{mos271}) define
diffeomorphisms of fibres
\mar{sp21,2}\ben
&&\psi_\xi(x): Y_x\to V, \qquad x\in U_\xi,\label{sp21}\\
&& \vr_{\xi\zeta}(x):V\to V, \qquad x\in U_\xi\cap U_\zeta. \label{sp22}
\een
 Trivialization charts $(U_\xi, \psi_\xi)$ together with
transition functions $\vr_{\xi\zeta}$ (\ref{mos271}) constitute a
 bundle atlas
\mar{sp5}\beq
\Psi = \{(U_\xi, \psi_\xi), \vr_{\xi\zeta}\} \label{sp5}
\eeq
of a fibre bundle $Y\to X$. Two bundle atlases are said to be
equivalent if their union also is a bundle atlas, i.e., there
exist transition functions between trivialization charts of
different atlases. All atlases of a fibre bundle are equivalent,
and a fibre bundle $Y\to X$ is uniquely defined by a bundle atlas.

Given a bundle atlas $\Psi$ (\ref{sp5}), a fibre bundle $Y$ is
provided with the fibred coordinates
\be
x^\la(y)=(x^\la\circ \pi)(y), \qquad y^i(y)=(y^i\circ\psi_\xi)(y),
\qquad y\in \pi^{-1}(U_\xi),
\ee
called the  bundle coordinates, where $y^i$ are coordinates on a
typical fibre $V$.

There is the following useful criterion for a fibred manifold to
be a fibre bundle.

\begin{theorem} \label{11t2} \mar{11t2} A fibred manifold whose
fibres are diffeomorphic either to a compact manifold or $\mathbb
R^r$ is a fibre bundle \cite{meig}.
\end{theorem}

In particular, a compact fibred manifold is a fibre bundle.

\begin{theorem} \label{11t3} \mar{11t3} Any fibre
bundle over a contractible base is trivial \cite{gre}.
\end{theorem}

Note that a fibred manifold over a contractible base need not be
trivial. It follows from Theorem \ref{11t3} that any cover of a
base $X$ by  domains (i.e., contractible open subsets) is a bundle
cover.

A fibred morphism of fibre bundles is called a  bundle morphism. A
bundle monomorphism $\Phi:Y\to Y'$ over $X$ onto a submanifold
$\Phi(Y)$ of $Y'$ is called a subbundle of a fibre bundle $Y'\to
X$.

The following are the standard constructions of new fibre bundles
from old ones.

$\bullet$ Given a fibre bundle $\pi:Y\to X$ and a manifold
morphism $f: X'\to X$, the  pull-back of $Y$ by $f$ is called the
manifold
\be
f^*Y =\{(x',y)\in X'\times Y \,: \,\, \pi(y) =f(x')\}
\ee
together with the natural projection $(x',y)\to x'$. It is a fibre
bundle over $X'$ such that the fibre of $f^*Y$ over a point $x'\in
X'$ is that of $Y$ over the point $f(x')\in X$. Any section $s$ of
a fibre bundle $Y\to X$ yields the pull-back section
$f^*s(x')=(x',s(f(x'))$ of $f^*Y\to X'$.

$\bullet$ If $X'\subset X$ is a submanifold of $X$ and $i_{X'}$ is
the corresponding natural injection, then the pull-back bundle
\be
i_{X'}^*Y=Y|_{X'}
\ee
is called the restriction of a fibre bundle $Y$ to the submanifold
$X'\subset X$. If $X'$ is an imbedded submanifold, any section of
the pull-back bundle $Y|_{X'}\to X'$ is the restriction to $X'$ of
some section of $Y\to X$.

$\bullet$ Let $\pi:Y\to X$ and $\pi':Y'\to X$ be fibre bundles
over the same base $X$. Their  bundle product $Y\times_X Y'$ over
$X$ is defined as the pull-back
\be
Y\op\times_X Y'=\pi^*Y'\quad {\rm or} \quad Y\op\times_X
Y'={\pi'}^*Y
\ee
together with its natural surjection onto $X$.  Fibres of the
bundle product $Y\times Y'$ are the Cartesian products $Y_x\times
Y'_x$ of fibres of fibre bundles $Y$ and $Y'$.

$\bullet$ Let us consider the  composite fibre bundle
\mar{1.34}\beq
Y\to \Si\to X. \label{1.34}
\eeq
It is provided with bundle coordinates $(x^\la,\si^m,y^i)$, where
$(x^\la,\si^m)$ are bundle coordinates on a fibre bundle $\Si\to
X$, i.e., transition functions of coordinates $\si^m$ are
independent of coordinates $y^i$. Let $h$ be a global section of a
fibre bundle $\Si\to X$. Then the restriction $Y_h=h^*Y$ of a
fibre bundle $Y\to\Si$ to $h(X)\subset \Si$ is a subbundle of a
fibre bundle $Y\to X$.

A fibre bundle $\pi:Y\to X$ is called a vector bundle if both its
typical fibre and fibres are finite-dimensional real vector
spaces, and if it admits a bundle atlas whose trivialization
morphisms and transition functions are linear isomorphisms. Then
the corresponding bundle coordinates on $Y$ are linear bundle
coordinates $(y^i)$ possessing linear transition functions
$y'^i=A^i_j(x)y^j$. We have
\be
y=y^ie_i(\pi(y))=y^i \psi_\xi(\pi(y))^{-1}(e_i), \qquad \pi(y)\in
U_\xi,
\ee
where $\{e_i\}$ is a fixed basis for a typical fibre $V$ of $Y$
and $\{e_i(x)\}$ are the fibre bases (or the  frames) for the
fibres $Y_x$ of $Y$ associated to a bundle atlas $\Psi$.

By virtue of Theorem \ref{mos9}, any vector bundle has a global
section, e.g., the canonical global  zero-valued section $\wh
0(x)=0$.

Global sections of a vector bundle $Y\to X$ constitute a
projective $C^\infty(X)$-module $Y(X)$ of finite rank. It is
called the structure module of a vector bundle.

There are the following particular constructions of new vector
bundles from the old ones.

$\bullet$ Let $Y\to X$ be a vector bundle with a typical fibre
$V$. By $Y^*\to X$ is denoted the  dual vector bundle with the
typical fibre $V^*$, dual of $V$. The  interior product of $Y$ and
$Y^*$ is defined as a fibred morphism
\be
\rfloor: Y\otimes Y^*\ar_X X\times \mathbb R.
\ee

$\bullet$ Let $Y\to X$ and $Y'\to X$ be vector bundles with
typical fibres $V$ and $V'$, respectively. Their  Whitney sum
$Y\oplus_X Y'$ is a vector bundle over $X$ with the typical fibre
$V\oplus V'$.

$\bullet$ Let $Y\to X$ and $Y'\to X$ be vector bundles with
typical fibres $V$ and $V'$, respectively. Their  tensor product
$Y\ot_X Y'$ is a vector bundle over $X$ with the typical fibre
$V\ot V'$. Similarly, the  exterior product of vector bundles
$Y\w_X Y'$ is defined. The exterior product
\mar{ss12f11}\beq
\w Y=X\times \mathbb R \op\oplus_X Y \op\oplus_X \op\w^2
Y\op\oplus_X\cdots\oplus \op\w^k Y, \qquad k=\di Y-\di X,
\label{ss12f11}
\eeq
is  called the  exterior bundle.

$\bullet$ If $Y'$ is a subbundle of a vector bundle $Y\to X$, the
 factor bundle $Y/Y'$ over $X$ is defined as a vector bundle
whose fibres are the quotients $Y_x/Y'_x$, $x\in X$.

By a morphism of vector bundles is meant a  linear bundle
morphism, which is a linear fibrewise map whose restriction to
each fibre is a linear map.

Given a linear bundle morphism $\Phi: Y'\to Y$ of vector bundles
over $X$, its  kernel Ker$\,\Phi$ is defined as the inverse image
$\Phi^{-1}(\wh 0(X))$ of the canonical zero-valued section $\wh
0(X)$ of $Y$. If $\Phi$ is of constant rank, its kernel and its
range are vector subbundles of the vector bundles $Y'$ and $Y$,
respectively. For instance, monomorphisms and epimorphisms of
vector bundles fulfil this condition.

\begin{remark}\label{mos30} \mar{mos30}
Given vector bundles $Y$ and $Y'$ over the same base $X$, every
linear bundle morphism
\be
\Phi: Y_x\ni \{e_i(x)\}\to \{\Phi^k_i(x)e'_k(x)\}\in Y'_x
\ee
over $X$ defines a global section
\be
\Phi: x\to \Phi^k_i(x)e^i(x)\ot e'_k(x)
\ee
of the tensor product $Y\ot Y'^*$, and {\it vice versa}.
\end{remark}

A sequence of vector bundles
\mar{sp10}\beq
0\to Y'\ar^i Y\ar^j Y'' \to 0 \label{sp10}
\eeq
over $X$ is said to be a  short exact sequence if it is exact at
all terms $Y'$, $Y$, and $Y''$. This means that $i$ is a bundle
monomorphism, $j$ is a bundle epimorphism, and Ker$\,j=\im i$.
Then $Y''$ is isomorphic to a factor bundle $Y/Y'$.

One says that the exact sequence (\ref{sp10}) is  split if there
exists a bundle monomorphism $s:Y''\to Y$ such that $j\circ s=\id
Y''$ or, equivalently,
\be
Y=i(Y')\oplus s(Y'')= Y'\oplus Y''.
\ee

\begin{theorem} \label{sp11} \mar{sp11}
Every exact sequence of vector bundles (\ref{sp10}) is split.
\end{theorem}

The tangent bundle $TZ$ and the cotangent bundle $T^*Z$ of a
manifold $Z$ exemplify vector bundles. The  cotangent bundle of a
manifold $Z$ is the dual $T^*Z\to Z$ of the tangent bundle $TZ\to
Z$. It is equipped with the holonomic coordinates
\be
(z^\la,\dot z_\la). \qquad \dot z'_\la=\frac{\dr z^\m}{\dr
z'^\la}\dot z_\m,
\ee
with respect to the  coframes  $\{dz^\la\}$ for $T^*Z$ which are
the duals of $\{\dr_\la\}$.

The tensor product of tangent and cotangent bundles
\mar{sp20}\beq
T=(\op\ot^mTZ)\ot(\op\ot^kT^*Z), \qquad m,k\in \mathbb N,
\label{sp20}
\eeq
is called a  tensor bundle,  provided with holonomic bundle
coordinates $\dot z^{\al_1\cdots\al_m}_{\bt_1\cdots\bt_k}$
possessing transition functions
\be
\dot z'^{\al_1\cdots\al_m}_{\bt_1\cdots\bt_k}=\frac{\dr
z'^{\al_1}}{\dr z^{\m_1}}\cdots\frac{\dr z'^{\al_m}}{\dr
z^{\m_m}}\frac{\dr z^{\nu_1}}{\dr z'^{\bt_1}}\cdots\frac{\dr
z^{\nu_k}}{\dr z'^{\bt_k}} \dot
z^{\m_1\cdots\m_m}_{\nu_1\cdots\nu_k}.
\ee

Let $\pi_Y:TY\to Y$ be the tangent bundle of a fibred manifold
$\pi: Y\to X$. Given fibred coordinates $(x^\la,y^i)$ on $Y$, it
is equipped with the holonomic coordinates $(x^\la,y^i,\dot x^\la,
\dot y^i)$. The tangent bundle $TY\to Y$ has the subbundle $VY =
\Ker (T\pi)$,  which consists of the vectors tangent to fibres of
$Y$. It is called the  vertical tangent bundle of $Y$, and it is
provided with the holonomic coordinates $(x^\la,y^i,\dot y^i)$
with respect to the  vertical frames $\{\dr_i\}$. Every fibred
morphism $\Phi: Y\to Y'$ yields the linear bundle morphism over
$\Phi$ of the vertical tangent bundles
\mar{ws538}\beq
V\Phi: VY\to VY', \qquad \dot y'^i\circ V\Phi=\frac{\dr
\Phi^i}{\dr y^j}\dot y^j. \label{ws538}
\eeq
It is called the  vertical tangent morphism.

In many important cases, the vertical tangent bundle $VY\to Y$ of
a fibre bundle $Y\to X$ is trivial, and it is isomorphic to the
bundle product
\mar{48'}\beq
VY= Y\op\times_X\ol Y, \label{48'}
\eeq
where $\ol Y\to X$ is some vector bundle. One calls (\ref{48'})
the  vertical splitting. For instance, every vector bundle $Y\to
X$ admits the  canonical vertical splitting
\mar{12f10}\beq
VY= Y\op\oplus_X Y. \label{12f10}
\eeq

The  vertical cotangent bundle $V^*Y\to Y$ of a fibred manifold
$Y\to X$ is defined as the dual of the vertical tangent bundle
$VY\to Y$. It is not a subbundle of the cotangent bundle $T^*Y$,
but there is the canonical surjection
\mar{z11z}\beq
\zeta: T^*Y\ni \dot x_\la dx^\la +\dot y_i dy^i \to \dot y_i \ol
dy^i\in V^*Y, \label{z11z}
\eeq
where the bases $\{\ol dy^i\}$, possessing transition functions
\be
\ol dy'^i=\frac{\dr y'^i}{\dr y^j}\ol dy^j,
\ee
are the duals of the vertical frames $\{\dr_i\}$ of the vertical
tangent bundle $VY$.

For any fibred manifold $Y$, there exist the exact sequences of
vector bundles
\mar{1.8a,b}\ben
&& 0\to VY\ar TY\ar^{\pi_T} Y\op\times_X TX\to 0,
\label{1.8a} \\
&& 0\to Y\op\times_X T^*X\to T^*Y\to V^*Y\to 0.
\label{1.8b}
\een
Their splitting, by definition, is a connection on $Y\to X$.

Let $\ol\pi:\ol Y\to X$ be a vector bundle with a typical fibre
$\ol V$. An  affine bundle modelled over the vector bundle $\ol
Y\to X$ is a fibre bundle $\pi:Y\to X$ whose typical fibre $V$ is
an affine space modelled over $\ol V$, all the fibres $Y_x$ of $Y$
are affine spaces modelled over the corresponding fibres $\ol Y_x$
of the vector bundle $\ol Y$, and there is an affine bundle atlas
\be
\Psi=\{(U_\al,\psi_\chi),\vr_{\chi\zeta}\}
\ee
of $Y\to X$ whose local trivializations morphisms $\psi_\chi$
(\ref{sp21}) and transition functions $\vr_{\chi\zeta}$
(\ref{sp22}) are affine isomorphisms. Dealing with affine bundles,
we use only  affine bundle coordinates $(y^i)$ associated to an
affine bundle atlas $\Psi$.

By virtue of Theorem \ref{mos9}, affine bundles  have  global
sections, but in contrast with vector bundles, there is no
canonical global section of an affine bundle.

By a morphism of affine bundles is meant a bundle morphism
$\Phi:Y\to Y'$ whose restriction to each fibre of $Y$ is an affine
map. It is called an  affine bundle morphism. Every affine bundle
morphism $\Phi:Y\to Y'$ of an affine bundle $Y$ modelled over a
vector bundle $\ol Y$ to an affine bundle $Y'$ modelled over a
vector bundle $\ol Y'$ yields an unique linear bundle morphism
\be
\ol \Phi: \ol Y\to \ol Y', \qquad \ol y'^i\circ \ol\Phi=
\frac{\dr\Phi^i}{\dr y^j}\ol y^j,
\ee
called the  linear derivative of $\Phi$.

Every affine bundle $Y\to X$ modelled over a vector bundle $\ol
Y\to X$ admits the  canonical vertical splitting
\mar{48}\beq
VY= Y\op\times_X\ol Y. \label{48}
\eeq

\subsection{Vector and multivector fields}

 Vector fields  on a manifold $Z$ are global sections of the
tangent bundle $TZ\to Z$.

The set $\cT_1(Z)$  of vector fields on $Z$ is both a
$C^\infty(Z)$-module and a real Lie algebra with respect to the
 Lie bracket
\be
&& u=u^\la\dr_\la, \qquad v=v^\la\dr_\la,\\
&& [v,u] = (v^\la\dr_\la u^\m - u^\la\dr_\la v^\m)\dr_\m.
\ee

Given a vector field $u$ on $X$, a  curve $c:\mathbb R\supset
(,)\to Z$ in $Z$ is said to be an  integral curve of $u$ if
$Tc=u(c)$. Every vector field $u$ on a manifold $Z$ can be seen as
an  infinitesimal generator of a local one-parameter group of
local diffeomorphisms (a  flow), and {\it vice versa} \cite{kob}.
One-dimensional orbits of this group are integral curves of $u$.

\begin{remark} \label{10b2} \mar{10b2}
Let  $U\subset Z$ be an open subset and $\e>0$. Recall that by a
local one-parameter group of local diffeomorphisms of $Z$ defined
on $(-\e,\e)\times U$ is meant a map
\be
G: (-\e,\e)\times U \ni (t,z)\to G_t(z)\in Z
\ee
which possesses the following properties:

$\bullet$ for each $t\in (-\e,\e)$, the mapping $G_t$ is a
diffeomorphism of $U$ onto the open subset $G_t(U)\subset Z$;

$\bullet$ $G_{t+t'}(z) = (G_t\circ G_{t'})(z)$ if $t+t'\in
(-\e,\e)$.

\noindent If such a map $G$ is defined on $\mathbb R\times Z$, it
is called the one-parameter group of diffeomorphisms of $Z$. If a
local one-parameter group of local diffeomorphisms of $Z$ is
defined on $(-\e,\e)\times Z$, it is uniquely prolonged onto
$\mathbb R\times Z$ to a one-parameter group of diffeomorphisms of
$Z$ \cite{kob}. As was mentioned above, a local one-parameter
group of local diffeomorphisms $G$ on $U\subset Z$ defines a local
vector field $u$ on $U$ by setting $u(z)$ to be the tangent vector
to the curve $s(t)=G_t(z)$ at $t=0$. Conversely, let $u$ be a
vector field on a manifold $Z$. For each $z\in Z$, there exist a
number $\e> 0$, a neighborhood $U$ of $z$ and a unique local
one-parameter group of local diffeomorphisms on $(-\e,\e)\times
U$, which determines $u$.
\end{remark}

A vector field is called complete if its flow is a one-parameter
group of diffeomorphisms of $Z$.

\begin{theorem} \label{10b3} \mar{10b3}
Any vector field on a compact manifold is complete.
\end{theorem}

A vector field $u$ on a fibred manifold $Y\to X$ is called
projectable if it is projected onto a vector field on $X$, i.e.,
there exists a vector field $\tau$ on $X$ such that
\be
\tau\circ\pi= T\pi\circ u.
\ee
A projectable vector field takes the coordinate form
\mar{11f30}\beq
u=u^\la(x^\m) \dr_\la + u^i(x^\m,y^j) \dr_i, \qquad
\tau=u^\la\dr_\la. \label{11f30}
\eeq
A projectable vector field is called  vertical if its projection
onto $X$ vanishes, i.e., if it lives in the vertical tangent
bundle $VY$.

A vector field $\tau=\tau^\la\dr_\la$ on a base $X$ of a fibred
manifold $Y\to X$ gives rise to a vector field on $Y$ by means of
a connection on this fibre bundle (see the formula (\ref{b1.85})
below). Nevertheless, every tensor bundle (\ref{sp20}) admits the
 functorial lift of vector fields
\mar{l28}\beq
\wt\tau = \tau^\m\dr_\m + [\dr_\nu\tau^{\al_1}\dot
x^{\nu\al_2\cdots\al_m}_{\bt_1\cdots\bt_k} + \ldots
-\dr_{\bt_1}\tau^\nu \dot
x^{\al_1\cdots\al_m}_{\nu\bt_2\cdots\bt_k} -\ldots]\dot \dr
_{\al_1\cdots\al_m}^{\bt_1\cdots\bt_k}, \label{l28}
\eeq
where we employ the compact notation
\mar{vvv}\beq
\dot\dr_\la = \frac{\dr}{\dr\dot x^\la}. \label{vvv}
\eeq
This lift is an $\mathbb R$-linear monomorphism of the Lie algebra
$\cT_1(X)$ of vector fields on $X$ to the Lie algebra $\cT_1(Y)$
of vector fields on $Y$. In particular, we have the functorial
lift
\mar{l27}\beq
\wt\tau = \tau^\m\dr_\m +\dr_\nu\tau^\al\dot
x^\nu\frac{\dr}{\dr\dot x^\al} \label{l27}
\eeq
of vector fields on $X$ onto the tangent bundle $TX$ and their
functorial lift
\mar{l27'}\beq
\wt\tau = \tau^\m\dr_\m -\dr_\bt\tau^\nu\dot
x_\nu\frac{\dr}{\dr\dot x_\bt} \label{l27'}
\eeq
onto the cotangent bundle $T^*X$.

Let $Y\to X$ be a vector bundle. Using the canonical vertical
splitting (\ref{12f10}), we obtain the canonical vertical vector
field
\mar{z112'}\beq
u_Y=y^i\dr_i \label{z112'}
\eeq
on $Y$, called the  Liouville vector field.

A  multivector field $\vt$ of degree $\nm\vt=r$ (or, simply, an
$r$-vector field) on a manifold $Z$ is a section
\mar{cc6}\beq
\vt =\frac{1}{r!}\vt^{\la_1\dots\la_r} \dr_{\la_1}\w\cdots\w
\dr_{\la_r} \label{cc6}
\eeq
of the exterior product $\op\w^r TZ\to Z$.  Let $\cT_r(Z)$ denote
the $C^\infty(Z)$-module space of $r$-vector fields on $Z$. All
multivector fields on a manifold $Z$ make up the graded
commutative algebra $\cT_*(Z)$ of global sections of the exterior
bundle $\w TZ$ (\ref{ss12f11}) with respect to the exterior
product $\w$.

The graded commutative algebra $\cT_*(Z)$ is endowed with the
Schouten -- Nijenhuis  bracket
\mar{z30}\ben
&& [.,.]_{\rm SN}: \cT_r(Z)\xx\cT_s(Z) \to \cT_{r+s-1}(Z), \label{z30} \\
&& [\vt,\up]_{\rm SN}=\vt\bullet\up
+(-1)^{rs}\up\bullet\vt, \nonumber \\
&& \vt\bullet\up = \frac{r}{r!s!}(\vt^{\mu\la_2\ldots\la_r}
\dr_\m\up^{\al_1\ldots\al_s}\dr_{\la_2}\w
\cdots\w\dr_{\la_r}\w\dr_{\al_1}\w\cdots\w\dr_{\al_s}). \nonumber
\een
This generalizes the Lie bracket of vector fields. It obeys the
relations
\be
&&[\vt,\up]_{\rm SN}=(-1)^{\nm\vt\nm\up}[\up,\vt]_{\rm SN}, \\\
&&[\nu,\vt\w\up]_{\rm SN} =[\nu,\vt]_{\rm SN}\w\up +(-1)^{(\nm\nu-1)\nm\vt}
\vt\w[\nu,\up]_{\rm SN}.
\ee

The  Lie derivative of a multivector field $\vt$ along a vector
field $u$ is defined as
\be
\bL_u\up=[u,\vt]_{\rm SN},\ qquad \bL_u(\vt\w\up)= \bL_u\vt\w\up
+\vt\w\bL_u\up.
\ee

\subsection{Differential forms}

An exterior $r$-form on a manifold $Z$ is a section
\be
\f =\frac{1}{r!}\f_{\la_1\dots\la_r} dz^{\la_1}\w\cdots\w
dz^{\la_r}
\ee
of the exterior product $\op\w^r T^*Z\to Z$, where
\be
&& dz^{\la_1}\w\cdots\w dz^{\la_r}=
\frac{1}{r!}\e^{\la_1\ldots\la_r}{}_{\m_1\ldots\m_r}dz^{\m_1}\ot\cdots\ot
dz^{\m_r},\\
&& \e^{\ldots \la_i\ldots\la_j\ldots}{}_{\ldots
\m_p\ldots\m_k\ldots}= -\e^{\ldots
\la_j\ldots\la_i\ldots}{}_{\ldots \m_p\ldots\m_k\ldots} = -
\e^{\ldots \la_i\ldots\la_j\ldots}{}_{\ldots
\m_k\ldots\m_p\ldots}, \\
&& \e^{\la_1\ldots\la_r}{}_{\la_1\ldots\la_r}=1.
\ee
Sometimes, it is convenient to write
\be
\f =\f'_{\la_1\dots\la_r} dz^{\la_1}\w\cdots\w dz^{\la_r}
\ee
without the coefficient $1/r!$.

Let $\cO^r(Z)$ denote the $C^\infty(Z)$-module of exterior
$r$-forms on a manifold $Z$. By definition, $\cO^0(Z)=C^\infty(Z)$
is the ring of smooth real functions on $Z$. All exterior forms on
$Z$ constitute the graded algebra $\cO^*(Z)$ of global sections of
the exterior bundle $\w T^*Z$ (\ref{ss12f11}) endowed with the
 exterior product
\be
&&\f=\frac{1}{r!}\f_{\la_1\dots\la_r} dz^{\la_1}\w\cdots\w
dz^{\la_r}, \qquad \si= \frac{1}{s!}\si_{\m_1\dots\m_s}
dz^{\m_1}\w\cdots\w dz^{\m_s},\\
&& \f\w\si=\frac{1}{r!s!}\f_{\nu_1\ldots\nu_r}\si_{\nu_{r+1}\ldots\nu_{r+s}}
dz^{\nu_1}\w\cdots\w
dz^{\nu_{r+s}}=\\
&& \qquad
\frac{1}{r!s!(r+s)!}\e^{\nu_1\ldots\nu_{r+s}}{}_{\al_1\ldots\al_{r+s}}
\f_{\nu_1\ldots\nu_r}\si_{\nu_{r+1}\ldots\nu_{r+s}}dz^{\al_1}\w\cdots\w
dz^{\al_{r+s}},
\ee
such  that
\be
\f\w\si=(-1)^{|\f||\si|}\si\w\f,
\ee
where the symbol $|\f|$ stands for the form degree. The algebra
$\cO^*(Z)$  also is provided with the exterior differential
\be
d\f= dz^\m\w \dr_\m\f=\frac{1}{r!} \dr_\m\f_{\la_1\ldots\la_r}
dz^\m\w dz^{\la_1}\w\cdots \w dz^{\la_r}
\ee
which obeys the relations
\be
d\circ d=0, \qquad d(\f\w\si)= d(\f)\w \si +(-1)^{|\f|}\f\w
d(\si).
\ee
The exterior differential $d$ makes $\cO^*(Z)$ into a differential
graded algebra, called the exterior algebra.

Given a manifold morphism $f:Z\to Z'$, any exterior $k$-form $\f$
on $Z'$ yields the  pull-back exterior form $f^*\f$ on $Z$ given
by the condition
\be
f^*\f(v^1,\ldots,v^k)(z) = \f(Tf(v^1),\ldots,Tf(v^k))(f(z))
\ee
for an arbitrary collection of tangent vectors $v^1,\cdots, v^k\in
T_zZ$. We have the relations
\be
f^*(\f\w\si) =f^*\f\w f^*\si, \qquad df^*\f =f^*(d\f).
\ee

In particular, given a fibred manifold $\pi:Y\to X$, the pull-back
onto $Y$ of exterior forms on $X$ by $\pi$ provides the
monomorphism of graded commutative algebras $\cO^*(X)\to
\cO^*(Y)$. Elements of its range $\pi^*\cO^*(X)$ are called basic
forms. Exterior forms
\be
\phi : Y\to\op\w^r T^*X, \qquad \phi
=\frac{1}{r!}\phi_{\la_1\ldots\la_r}dx^{\la_1}\w\cdots\w
dx^{\la_r},
\ee
on $Y$ such that $u\rfloor\f=0$ for an arbitrary vertical vector
field $u$ on $Y$ are said to be  horizontal forms. Horizontal
forms of degree $n=\dim X$ are called  densities.

In the case of the tangent bundle $TX\to X$, there is a different
way to lift exterior forms on $X$ onto $TX$ \cite{grab}. Let $f$
be a function on $X$. Its  tangent lift onto $TX$ is defined as
the function
\mar{gm62}\beq
\wt f=\dot x^\la\dr_\la f. \label{gm62}
\eeq
Let $\si$ be an $r$-form on $X$. Its  tangent lift onto $TX$ is
said to be the $r$-form $\wt \si$ given by the relation
\mar{z115}\beq
\wt \si(\wt\tau_1,\ldots,\wt\tau_r)=
\wt{\si(\tau_1,\ldots,\tau_r)}, \label{z115}
\eeq
where $\tau_i$ are arbitrary vector fields on $X$ and $\wt\tau_i$
are their functorial lifts (\ref{l27}) onto $TX$. We have the
coordinate expression
\mar{gm61}\ben
&& \si=\frac{1}{r!}\si_{\la_1\cdots\la_r}dx^{\la_1}\w\cdots\w dx^{\la_r},
\nonumber\\
&& \wt\si =\frac1{r!}[\dot x^\m\dr_\m
\si_{\la_1\cdots\la_r}dx^{\la_1}\w\cdots\w dx^{\la_r}+ \label{gm61}\\
&&\qquad \op\sum_{i=1}^r
\si_{\la_1\cdots\la_r}dx^{\la_1}\w\cdots\w d\dot
x^{\la_i}\w\cdots\w dx^{\la_r}]. \nonumber
\een
The following equality holds:
\mar{u10}\beq
d\wt\si=\wt{d\si}. \label{u10}
\eeq

The  interior product (or  contraction) of a vector field $u$ and
an exterior $r$-form $\f$ on a manifold $Z$ is given by the
coordinate expression
\be
&& u\rfloor\f = \op\sum_{k=1}^r \frac{(-1)^{k-1}}{r!} u^{\la_k}
\f_{\la_1\ldots\la_k\ldots\la_r} dz^{\la_1}\w\cdots\w\wh
{dz}^{\la_k}\w\cdots \w dz^{\la_r}= \\
&& \qquad \frac{1}{(r-1)!}u^\m\f_{\m\al_2\ldots\al_r} dz^{\al_2}\w\cdots\w
dz^{\al_r},
\ee
where the caret $\,\wh{}\,$ denotes omission. It obeys the
relations
\mar{031}\ben
&& \f(u_1,\ldots,u_r)=u_r\rfloor\cdots u_1\rfloor\f,\nonumber\\
&& u\rfloor(\f\w\si)= u\rfloor\f\w\si +(-1)^{|\f|}\f\w
u\rfloor\si. \label{031}
\een

The  Lie derivative of an exterior form $\f$ along a vector field
$u$  is
\mar{034,5}\ben
&& \bL_u\f = u\rfloor d\f +d(u\rfloor\f), \label{034}\\
&& \bL_u(\f\w\si)= \bL_u\f\w\si +\f\w\bL_u\si. \label{035}
\een
In particular, if $f$ is a function, then
\be
\bL_u f =u(f)=u\rfloor d f.
\ee
An exterior form $\f$ is invariant under a local one-parameter
group of diffeomorphisms $G_t$ of $Z$ (i.e., $G_t^*\f=\f$) iff its
Lie derivative along the infinitesimal generator $u$ of this group
vanishes, i.e.,
\be
\bL_u\f=0.
\ee
Following physical terminology (Definition \ref{084}), we say that
a vector field $u$ is a  symmetry of an exterior form $\f$.

A  tangent-valued $r$-form on a manifold $Z$ is a section
\be
\phi = \frac{1}{r!}\phi_{\la_1\ldots\la_r}^\m dz^{\la_1}\w\cdots\w
dz^{\la_r}\ot\dr_\m
\ee
of the tensor bundle
\be
\op\w^r T^*Z\ot TZ\to Z.
\ee

\begin{remark}
There is one-to-one correspondence between the tangent-valued
one-forms $\f$ on a manifold $Z$ and the linear bundle
endomorphisms
\mar{29b,b'}\ben
&& \wh\f:TZ\to TZ,\quad
\wh\f: T_zZ\ni v\to v\rfloor\f(z)\in T_zZ, \label{29b} \\
&&\wh\f^*:T^*Z\to T^*Z,\quad \wh\f^*: T_z^*Z\ni v^*\to
\f(z)\rfloor v^*\in T_z^*Z, \label{29b'}
\een
over $Z$ (Remark \ref{mos30}). For instance, the  canonical
tangent-valued one-form
\mar{b1.51}\beq
\thh_Z= dz^\la\ot \dr_\la \label{b1.51}
\eeq
on $Z$ corresponds to the identity morphisms (\ref{29b}) and
(\ref{29b'}).
\end{remark}

The space $\cO^*(Z)\ot \cT_1(Z)$ of tangent-valued forms is
provided with the  Fr\"olicher -- Nijenhuis bracket
\be
&& [,]_{\rm FN}:\cO^r(Z)\ot \cT_1(Z)\times \cO^s(Z)\ot \cT_1(Z)
\to\cO^{r+s}(Z)\ot \cT_1(Z), \\
&& [\al\ot u,\, \bt\ot v]_{\rm FN} = (\al\w\bt)\ot [u, v] +
(\al\w \bL_u \bt)\ot v -  \\
&& \qquad(\bL_v \al\w\bt)\ot u +  (-1)^r (d\al\w u\rfloor\bt)\ot v
+ (-1)^r(v\rfloor\al\w d\bt)\ot u, \\
&&  \al\in\cO^r(Z), \qquad \bt\in \cO^s(Z), \qquad u,v\in\cT_1(Z).
\ee
Its coordinate expression is
\be
&& [\phi,\si]_{\rm FN} = \frac{1}{r!s!}(\phi_{\la_1
\dots\la_r}^\nu\dr_\n\si_{\la_{r+1}\dots\la_{r+s}}^\m -
\si_{\la_{r+1}
\dots\la_{r+s}}^\nu\dr_\nu\phi_{\la_1\dots\la_r}^\m -\\
&& \qquad r\phi_{\la_1\ldots\la_{r-1}\nu}^\m\dr_{\la_r}\si_{\la_{r+1}
\dots\la_{r+s}}^\nu + s \si_{\nu\la_{r+2}\ldots\la_{r+s}}^\m
\dr_{\la_{r+1}}\phi_{\la_1\ldots\la_r}^\nu)\\
&& \qquad dz^{\la_1}\wedge\cdots \wedge dz^{\la_{r+s}}\otimes\dr_\m,\\
&&
\f\in \cO^r(Z)\ot \cT_1(Z), \qquad \si\in \cO^s(Z)\ot \cT_1(Z).
\ee
There are the relations
\be
&& [\f,\si]_{\rm FN}=(-1)^{|\f||\s|+1}[\si,\f]_{\rm FN},  \\
&& [\f, [\si, \thh]_{\rm FN}]_{\rm FN} = [[\f, \si]_{\rm FN}, \thh]_{\rm
FN} + (-1)^{|\f||\si|}  [\si, [\f,\thh]_{\rm FN}]_{\rm FN}, \\
&& \f,\si,\thh\in \cO^*(Z)\ot \cT_1(Z).
\ee

Given a tangent-valued form  $\thh$, the  Nijenhuis differential
on $\cO^*(Z)\ot\cT_1(Z)$ is defined as the morphism
\be
d_\thh : \psi\to d_\thh\psi = [\thh,\psi]_{\rm FN}, \qquad
\psi\in\cO^*(Z)\ot\cT_1(Z).
\ee

In particular, if $\f$ is a tangent-valued one-form, the Nijenhuis
differential
\mar{spr613}\ben
&& d_\f\f=[\f,\f]_{\rm FN}=   \label{spr613}\\
&& \qquad (\f^\m_\nu\dr_\m \f^\al_\bt -\f^\m_\bt\dr_\m \f^\al_\nu -\f^\al_\m\dr_\nu \f^\m_\bt +
\f^\al_\m\dr_\bt \f^\m_\nu)dz^\nu\w dz^\bt\ot\dr_\al \nonumber
\een
is called the  Nijenhuis torsion.

Let $\Y$ be a fibred manifold. We consider the following subspaces
of the space $\cO^*(Y)\ot \cT_1(Y)$ of tangent-valued forms on
$Y$:

$\bullet$ horizontal tangent-valued forms
\be
&& \phi : Y\to\op\w^r T^*X\op\otimes_Y TY,\\
&& \phi =dx^{\la_1}\wedge\cdots\wedge dx^{\la_r}\otimes
\frac{1}{r!}[\phi_{\la_1\ldots\la_r}^\m(y) \dr_\m
+\phi_{\la_1\ldots\la_r}^i(y) \dr_i],
\ee

$\bullet$  vertical-valued form
\be
\phi : Y\to\op\w^r T^*X\op\otimes_Y VY,\quad \phi
=\frac{1}{r!}\phi_{\la_1\ldots\la_r}^i(y)dx^{\la_1}\wedge\cdots
\wedge dx^{\la_r}\otimes\dr_i,
\ee

$\bullet$ vertical-valued one-forms, called  soldering forms,
\mar{1290}\beq
\si = \si_\la^i(y) dx^\la\otimes\dr_i. \label{1290}
\eeq

\subsection{Distributions and foliations}

A subbundle $\bT$ of the tangent bundle $TZ$ of a manifold $Z$ is
called a regular  distribution (or, simply, a distribution). A
vector field $u$ on $Z$ is said to be  subordinate to a
distribution $\bT$ if it lives in $\bT$. A distribution $\bT$ is
called  involutive if the Lie bracket of $\bT$-subordinate vector
fields also is subordinate to $\bT$.

A subbundle of the cotangent bundle $T^*Z$ of $Z$ is called a
codistribution $\bT^*$ on a manifold $Z$. For instance, the
annihilator  $\rA\bT$ of a distribution $\bT$ is a codistribution
whose fibre over $z\in Z$ consists of covectors $w\in T^*_z$ such
that $v\rfloor w=0$ for all $v\in \bT_z$.

The following local coordinates can be associated to an involutive
distribution  \cite{war}.

\begin{theorem}\label{c11.0} \mar{c11.0} Let $\bT$ be an involutive
$r$-dimensional distribution on a manifold $Z$, $\di Z=k$. Every
point $z\in Z$ has an open neighborhood $U$ which is a domain of
an adapted coordinate chart $(z^1,\dots,z^k)$ such that,
restricted to $U$, the distribution $\bT$ and its annihilator
$\rA\bT$ are spanned by the local vector fields $\dr/\dr z^1,
\cdots,\dr/\dr z^r$ and the local one-forms $dz^{r+1},\dots,
dz^k$, respectively.
\end{theorem}

A connected submanifold $N$ of a manifold $Z$ is called an
integral manifold of a distribution $\bT$ on $Z$ if $TN\subset
\bT$. Unless otherwise stated, by an integral manifold is meant an
integral manifold of dimension of $\bT$. An integral manifold is
called  maximal if no other integral manifold contains it. The
following is the classical theorem of Frobenius \cite{kob,war}.

\begin{theorem}\label{to.1}  \mar{to.1} Let $\bT$ be an
involutive distribution on a manifold $Z$. For any $z\in Z$, there
exists a unique maximal integral manifold of $\bT$ through $z$,
and any integral manifold through $z$ is its open subset.
\end{theorem}

Maximal integral manifolds of an involutive distribution on a
manifold $Z$ are assembled into a regular foliation $\cF$ of $Z$.
A regular $r$-dimensional  foliation (or, simply, a foliation)
$\cF$ of a $k$-dimensional manifold $Z$ is defined as a partition
of $Z$ into connected $r$-dimensional submanifolds (the leaves of
a foliation) $F_\iota$, $\iota\in I$, which possesses the
following properties \cite{rei,tam}.

A manifold $Z$ admits an adapted coordinate atlas
\mar{spr850}\beq
\{(U_\xi;z^\la, z^i)\},\quad \la=1,\ldots,k-r, \qquad
i=1,\ldots,r, \label{spr850}
\eeq
such that transition functions of coordinates $z^\la$ are
independent of the remaining coordinates $z^i$. For each leaf $F$
of a foliation $\cF$, the connected components of $F\cap U_\xi$
are given by the equations $z^\la=$const. These connected
components and coordinates $(z^i)$ on them make up a coordinate
atlas of a leaf $F$. It follows that tangent spaces to leaves of a
foliation $\cF$ constitute an involutive distribution $T\cF$ on
$Z$, called the  tangent bundle to the foliation $\cF$. The factor
bundle $V\cF=TZ/T\cF$, called the  normal bundle to $\cF$, has
transition functions independent of coordinates $z^i$. Let
$T\cF^*\to Z$ denote the dual of $T\cF\to Z$. There are the exact
sequences
\mar{pp2,3}\ben
&& 0\to T\cF \ar^{i_\cF} TX \ar V\cF\to 0, \label{pp2} \\
&& 0\to {\rm Ann}\,T\cF\ar T^*X\ar^{i^*_\cF} T\cF^* \to 0
\label{pp3}
\een
of vector bundles over $Z$.

A pair $(Z,\cF)$, where $\cF$ is a foliation of $Z$, is called a
 foliated manifold. It should be emphasized that leaves of a
foliation need not be closed or imbedded submanifolds. Every leaf
has an open  saturated neighborhood $U$, i.e., if $z\in U$, then a
leaf through $z$ also belongs to $U$.

Any submersion $\zeta:Z\to M$ yields a foliation
\be
\cF=\{F_p=\zeta^{-1}(p)\}_{p\in \zeta(Z)}
\ee
of $Z$ indexed by elements of $\zeta(Z)$, which is an open
submanifold of $M$, i.e., $Z\to \zeta(Z)$ is a fibred manifold.
Leaves of this foliation are closed imbedded submanifolds. Such a
foliation is called  simple. Any (regular) foliation is locally
simple.

\begin{example} \label{10b6} \mar{10b6}
Every smooth real function $f$ on a manifold $Z$ with nowhere
vanishing differential $df$ is a submersion $Z\to \mathbb R$. It
defines a one-codimensional foliation whose leaves are given by
the equations
\be
f(z)=c, \qquad c\in f(Z)\subset\mathbb R.
\ee
This is the  foliation of level surfaces of the function $f$,
called a  generating function. Every one-codimensional foliation
is locally a foliation of level surfaces of some function on $Z$.
The level surfaces of an arbitrary smooth real function $f$ on a
manifold $Z$ define a  singular foliation $\cF$ on $Z$
\cite{kamb}. Its leaves are not submanifolds in general.
Nevertheless if $df(z)\neq 0$, the restriction of $\cF$ to some
open neighborhood $U$ of $z$ is a foliation with the generating
function $f|_U$.
\end{example}

Let $\cF$ be a (regular) foliation of a $k$-dimensional manifold
$Z$ provided with the adapted coordinate atlas (\ref{spr850}). The
real Lie algebra $\cT_1(\cF)$ of  global sections of the tangent
bundle $T\cF\to Z$ to $\cF$ is a $C^\infty(Z)$-submodule of the
derivation module of the $\mathbb R$-ring $C^\infty(Z)$ of smooth
real functions on $Z$. Its kernel $S_\cF(Z)\subset C^\infty(Z)$
consists of functions constant on leaves of $\cF$. Therefore,
$\cT_1(\cF)$ is the Lie $S_\cF(Z)$-algebra of derivations of
$C^\infty(Z)$, regarded as a $S_\cF(Z)$-ring. Then one can
introduce the leafwise differential calculus \cite{jmp02,book05}
as the Chevalley -- Eilenberg differential calculus over the
$S_\cF(Z)$-ring $C^\infty(Z)$. It is defined as a subcomplex
\mar{spr892}\beq
0\to S_\cF(Z)\ar C^\infty(Z)\ar^{\wt d} \gF^1(Z) \cdots \ar^{\wt
d} \gF^{\di\cF}(Z) \to 0 \label{spr892}
\eeq
of the Chevalley -- Eilenberg complex of the Lie
$S_\cF(Z)$-algebra $\cT_1(\cF)$ with coefficients in $C^\infty(Z)$
which consists of $C^\infty(Z)$-multilinear skew-symmetric maps
\be
\op\times^r \cT_1(\cF) \to C^\infty(Z), \qquad r=1,\ldots,\di\cF.
\ee
These maps are global sections of exterior products $\op\w^r
T\cF^*$ of the dual $T\cF^*\to Z$ of $T\cF\to Z$. They are called
the leafwise forms on a foliated manifold $(Z,\cF)$, and are given
by the coordinate expression
\be
\f=\frac1{r!}\f_{i_1\ldots i_r}\wt dz^{i_1}\w\cdots\w \wt
dz^{i_r},
\ee
where $\{\wt dz^i\}$ are the duals of the holonomic fibre bases
$\{\dr_i\}$ for $T\cF$. Then one can think of the Chevalley --
Eilenberg coboundary operator
\be
\wt d\f= \wt dz^k\w \dr_k\f=\frac{1}{r!} \dr_k\f_{i_1\ldots
i_r}\wt dz^k\w\wt dz^{i_1}\w\cdots\w\wt dz^{i_r}
\ee
as being the  leafwise exterior differential. Accordingly, the
complex (\ref{spr892}) is called the  leafwise de Rham complex (or
the tangential de Rham complex).

Let us consider the exact sequence (\ref{pp3}) of vector bundles
over $Z$. Since it admits a splitting, the epimorphism $i^*_\cF$
yields that of the algebra $\cO^*(Z)$ of exterior forms on $Z$ to
the algebra $\gF^*(Z)$ of leafwise forms. It obeys the condition
$i^*_\cF\circ d=\wt d\circ i^*_\cF$, and provides the cochain
morphism
\mar{lmp04}\ben
&& i^*_\cF: (\mathbb R,\cO^*(Z),d)\to (S_\cF(Z),\cF^*(Z),\wt d),
\label{lmp04}\\
&& dz^\la\to 0, \quad dz^i\to\wt dz^i, \nonumber
\een
of the de Rham complex of $Z$ to the leafwise de Rham complex
(\ref{spr892}).

Given a leaf $i_F:F\to Z$ of $\cF$, we have the pull-back
homomorphism
\mar{lmp30}\beq
(\mathbb R,\cO^*(Z),d) \to (\mathbb R,\cO^*(F),d) \label{lmp30}
\eeq
of the de Rham complex of $Z$ to that of $F$.

\begin{proposition} \label{lmp32} \mar{lmp32} The homomorphism (\ref{lmp30})
factorize through the homomorphism \cite{book05}.
\end{proposition}

\subsection{Differential geometry of Lie groups}

Let $G$ be a real Lie group of $\di G >0$, and let $L_g:G\to gG$
and $R_g:G\to Gg$ denote the action of $G$ on itself by left and
right multiplications, respectively. Clearly, $L_g$ and $R_{g'}$
for all $g,g'\in G$ mutually commute, and so do the tangent maps
$TL_g$ and $TR_{g'}$.

A vector field $\xi_l$ (resp. $\xi_r$) on a group $G$ is said to
be left-invariant (resp.  right-invariant) if $\xi_l\circ L_g
=TL_g\circ \xi_l$ (resp. $\xi_r\circ R_g=TR_g\circ \xi_r$).
Left-invariant (resp. right-invariant) vector fields make up the
 left Lie algebra $\cG_l$ (resp. the  right Lie algebra
$\cG_r$) of  $G$.

There is one-to-one correspondence between the left-invariant
vector field $\xi_l$ (resp. right-invariant vector fields $\xi_r$)
on $G$ and the vectors $\xi_l(e)=TL_{g^{-1}}\xi_l(g)$ (resp.
$\xi_r(e)=TR_{g^{-1}}\xi_l(g)$) of the tangent space $T_eG$ to $G$
at the unit element $e$ of $G$. This correspondence provides
$T_eG$ with the left and the right Lie algebra structures.
Accordingly, the left action $L_g$ of a Lie group $G$ on itself
defines its  adjoint representation
\mar{1210}\beq
\xi_r \to {\rm Ad}\,g(\xi_r)=TL_g\circ\xi_r\circ L_{g^{-1}}
\label{1210}
\eeq
in the right Lie algebra $\cG_r$.

Let $\{\e_m\}$ (resp. $\{\ve_m\}$) denote the basis for the left
(resp. right) Lie algebra, and let $c^k_{mn}$ be the  right
structure constants
\be
[\ve_m,\ve_n]=c^k_{mn}\ve_k.
\ee
There is the morphism
\be
\rho: \cG_l\ni \e_m \to\ve_m= -\e_m\in \cG_r
\ee
between left and right Lie algebras such that
\be
[\e_m,\e_n]=-c^k_{mn}\e_k.
\ee

The tangent bundle $\pi_G:TG\to G$ of a Lie group $G$ is trivial.
There are the following two canonical isomorphisms
\be
&& \vr_l: TG\ni q\to (g=\pi_G(q), TL^{-1}_g(q))\in G\times \cG_l,\\
&& \vr_r: TG\ni q\to (g=\pi_G(q), TR^{-1}_g(q))\in G\times \cG_r.
\ee
Therefore, any action
\be
G\times Z\ni (g,z)\to gz\in Z
\ee
of a Lie group $G$ on a manifold $Z$ on the left yields the
homomorphism
\mar{spr941}\beq
\cG_r\ni\ve\to \xi_\ve\in \cT_1(Z) \label{spr941}
\eeq
of the right Lie algebra $\cG_r$ of $G$ into the Lie algebra of
vector fields on $Z$ such that
\mar{z212}\beq
\xi_{{\rm Ad}\,g(\ve)}=Tg\circ \xi_\ve\circ g^{-1}. \label{z212}
\eeq
Vector fields $\xi_\ve$ are said to be the  infinitesimal
generators of a representation  of the Lie group $G$ in $Z$.

In particular, the adjoint representation (\ref{1210}) of a Lie
group in its right Lie algebra yields the {\it adjoint
representation}
\be
\ve': \ve\to {\rm ad}\,\ve' (\ve)=[\ve',\ve], \qquad {\rm
ad}\,\ve_m(\ve_n)=c^k_{mn}\ve_k,
\ee
of the right Lie algebra $\cG_r$ in itself.

The dual $\cG^*=T^*_eG$ of the tangent space $T_eG$ is called the
 Lie coalgebra). It is provided with the basis $\{\ve^m\}$
which is the dual of the basis $\{\ve_m\}$ for $T_eG$. The group
$G$ and the right Lie algebra $\cG_r$ act on $\cG^*$ by the
coadjoint representation
\mar{z211}\ben
&&\lng{\rm Ad}^*g(\ve^*),\ve\rng =\lng\ve^*,{\rm
Ad}\,g^{-1}(\ve)\rng, \quad
\ve^*\in \cG^*, \quad \ve\in \cG_r, \label{z211}\\
&&\lng{\rm ad}^*\ve'(\ve^*),\ve\rng=-\lng\ve^*,[\ve',\ve]\rng, \qquad
\ve'\in \cG_r,\nonumber\\
&&  {\rm ad}^*\ve_m(\ve^n)=-c^n_{mk}\ve^k.\nonumber
\een

The Lie coalgebra $\cG^*$ of a Lie group $G$ is provided with the
canonical Poisson structure, called the  Lie -- Poisson structure
\cite{abr,libe}. It is given by the bracket
\mar{gm510}\beq
\{f,g\}_{\rm LP}=\langle\ve^*,[df(\ve^*),dg(\ve^*)]\rangle, \qquad
f,g\in C^\infty(\cG^*), \label{gm510}
\eeq
where $df(\ve^*),dg(\ve^*)\in \cG_r$ are seen as linear mappings
from $T_{\ve^*}\cG^*=\cG^*$ to $\mathbb R$. Given coordinates
$z_k$ on $\cG^*$ with respect to the basis $\{\ve^k\}$, the Lie --
Poisson bracket (\ref{gm510}) and the corresponding Poisson
bivector field $w$ read
\be
\{f,g\}_{\rm LP}=c^k_{mn}z_k\dr^mf\dr^ng, \qquad
w_{mn}=c^k_{mn}z_k.
\ee
One can show that symplectic leaves of the Lie -- Poisson
structure on the coalgebra $\cG^*$ of a connected Lie group $G$
are orbits of the coadjoint representation (\ref{z211}) of $G$ on
$\cG^*$ \cite{wein}.

\subsection{Jet manifolds}

This Section collects the relevant material on jet manifolds of
sections of fibre bundles \cite{book09,kol,book00,sard08,sau}.

Given a fibre bundle $Y\to X$ with bundle coordinates
$(x^\la,y^i)$, let us consider the equivalence classes $j^1_xs$ of
its sections $s$, which are identified by their values $s^i(x)$
and the values of their partial derivatives $\dr_\mu s^i(x)$ at a
point $x\in X$. They are called the  first order jets of sections
at $x$. One can justify that the definition of jets is
coordinate-independent. A key point is that the set $J^1Y$ of
first order jets $j^1_xs$, $x\in X$, is a smooth manifold with
respect to the adapted coordinates $(x^\la,y^i,y_\la^i)$ such that
\mar{50}\beq
y_\la^i(j^1_xs)=\dr_\la s^i(x),\qquad {y'}^i_\la = \frac{\dr
x^\m}{\dr{x'}^\la}(\dr_\m +y^j_\m\dr_j)y'^i.\label{50}
\eeq
It is called the first order jet manifold of a fibre bundle $Y\to
X$.

A jet manifold $J^1Y$ admits the natural fibrations
\mar{1.14,5}\ben
&&\pi^1:J^1Y\ni j^1_xs\to x\in X, \label{1.14}\\
&&\pi^1_0:J^1Y\ni j^1_xs\to s(x)\in Y. \label{1.15}
\een
A glance at the transformation law (\ref{50}) shows that $\pi^1_0$
is an affine bundle modelled over the vector bundle
\mar{cc9}\beq
T^*X \op\otimes_Y VY\to Y.\label{cc9}
\eeq
It is convenient to call $\pi^1$ (\ref{1.14}) the  jet bundle,
while $\pi^1_0$ (\ref{1.15}) is said to be the  affine jet bundle.

Let us note that, if $Y\to X$ is a vector or an affine bundle, the
jet bundle $\pi_1$ (\ref{1.14}) is so.

Jets can be expressed in terms of familiar tangent-valued forms as
follows. There are the canonical imbeddings
\mar{18,24}\ben
&&\la_{(1)}:J^1Y\op\to_Y
T^*X \op\otimes_Y TY,\nonumber\\
&& \la_{(1)}=dx^\la
\otimes (\dr_\la + y^i_\la \dr_i)=dx^\la\otimes d_\la, \label{18}\\
&&\thh_{(1)}:J^1Y \op\to_Y T^*Y\op\otimes_Y VY,\nonumber\\
&&\thh_{(1)}=(dy^i- y^i_\la dx^\la)\otimes \dr_i=\thh^i \otimes
\dr_i,\label{24}
\een
where $d_\la$ are said to be  total derivatives, and $\thh^i$ are
called  contact forms.

We further identify the jet manifold $J^1Y$ with its images under
the canonical morphisms (\ref{18}) and (\ref{24}), and represent
the jets $j^1_xs=(x^\la,y^i,y^i_\m)$ by the tangent-valued forms
$\la_{(1)}$ (\ref{18}) and $\thh_{(1)}$ (\ref{24}).

Sections and morphisms of fibre bundles admit prolongations to jet
manifolds as follows.

Any section $s$ of a fibre bundle $Y\to X$ has the  jet
prolongation to the section
\be
(J^1s)(x)= j_x^1s, \qquad y_\la^i\circ J^1s= \dr_\la s^i(x),
\ee
of the jet bundle $J^1Y\to X$. A  section of the jet bundle
$J^1Y\to X$ is called  integrable if it is the jet prolongation of
some section of a fibre bundle $Y\to X$.

Any bundle morphism $\Phi:Y\to Y'$ over a diffeomorphism $f$
admits a   jet prolongation to a bundle morphism of affine jet
bundles
\be
J^1\Phi : J^1Y \ar_\Phi J^1Y', \qquad {y'}^i_\la\circ
J^1\Phi=\frac{\dr(f^{-1})^\m}{\dr x'^\la}d_\m\Phi^i.
\ee

Any projectable vector field $u$ (\ref{11f30}) on a fibre bundle
$Y\to X$ has a  jet prolongation to the projectable vector field
\be
J^1u =u^\la\dr_\la + u^i\dr_i + (d_\la u^i - y_\m^i\dr_\la
u^\m)\dr_i^\la,
\ee
on the jet manifold $J^1Y$.

\subsection{Connections on fibre bundles}

There are different equivalent definitions of a connection on a
fibre bundle $Y\to X$. We define it both as a splitting of the
exact sequence (\ref{1.8a}) and a global section of the affine jet
bundle $J^1Y\to $Y \cite{book09,book00,sard08}.

A  connection on a fibred manifold $Y\to X$ is defined as a
splitting (called the horizontal splitting)
\mar{150}\ben
&& \G: Y\op\times_X TX\op\to_Y TY, \qquad
   \G: \dot x^\la\dr_\la \to \dot x^\la(\dr_\la+\G^i_\la(y)\dr_i),
\label{150}\\
&& \dot x^\la \dr_\la + \dot y^i \dr_i = \dot
x^\la(\dr_\la + \G_\la^i \dr_i) + (\dot y^i - \dot
x^\la\G_\la^i)\dr_i,  \nonumber
\een
of the exact sequence (\ref{1.8a}). Its range is a subbundle of
$TY\to Y$ called the  horizontal distribution. By virtue of
Theorem \ref{sp11}, a connection on a fibred manifold always
exists. A connection $\G$ (\ref{150}) is represented by the
horizontal tangent-valued one-form
\mar{154}\beq
\G = dx^\la\otimes (\dr_\la + \G_\la^i\dr_i) \label{154}
\eeq
on $Y$ which is projected onto the canonical tangent-valued form
$\thh_X$ (\ref{b1.51}) on $X$.

Given a connection $\G$ on a fibred manifold $Y\to X$, any vector
field $\tau$ on a base $X$ gives rise to the projectable vector
field
\mar{b1.85}\beq
\G \tau=\tau\rfloor\G=\tau^\la(\dr_\la +\G^i_\la\dr_i)
\label{b1.85}
\eeq
on $Y$ which lives in the horizontal distribution determined by
$\G$. It is called the  horizontal lift of
 $\tau$ by means of a
connection $\G$.

The splitting (\ref{150}) also is given by the vertical-valued
form
\mar{b1.223}\beq
\G= (dy^i -\G^i_\la dx^\la)\ot\dr_i, \label{b1.223}
\eeq
which yields an epimorphism $TY\to VY$.

In an equivalent way, connections on a fibred manifold $Y\to X$
are introduced as global sections of the affine jet bundle
$J^1Y\to Y$. Indeed, any global section $\G$ of $J^1Y\to Y$
defines the tangent-valued form $\la_1\circ \G$ (\ref{154}). It
follows from this definition that connections on a fibred manifold
$Y\to X$ constitute an affine space modelled over the vector space
of soldering forms $\si$ (\ref{1290}). One also deduces from
(\ref{50}) the coordinate transformation law of connections
\be
\G'^i_\la = \frac{\dr x^\m}{\dr{x'}^\la}(\dr_\m
+\G^j_\m\dr_j)y'^i.
\ee

\begin{remark} \label{Ehresmann}  Any connection $\G$ on a fibred manifold $Y\to
X$ yields a horizontal lift of a vector field on $X$ onto $Y$, but
need not defines the similar lift of a  path  in $X$ into $Y$. Let
\be
\mathbb R\supset[,]\ni t\to x(t)\in X, \qquad \mathbb R\ni t\to
y(t)\in Y,
\ee
be smooth paths in $X$ and $Y$, respectively. Then $t\to y(t)$ is
called a  horizontal lift of $x(t)$ if
\be
\pi(y(t))= x(t), \qquad \dot y(t)\in H_{y(t)}Y, \qquad t\in\mathbb
R,
\ee
where $HY\subset TY$ is the horizontal subbundle associated to the
connection $\G$. If, for each path $x(t)$ $(t_0\leq t\leq t_1)$
and for any $y_0\in\pi^{-1}(x(t_0))$, there exists a horizontal
lift $y(t)$ $(t_0\leq t\leq t_1)$ such that $y(t_0)=y_0$, then
$\G$ is called the  Ehresmann connection. A fibred manifold is a
fibre bundle iff it admits an Ehresmann connection \cite{gre}.
\end{remark}

Hereafter, we restrict our consideration to connections on fibre
bundles. The following are two standard constructions of new
connections from old ones.

$\bullet$ Let $Y$ and $Y'$ be fibre bundles over the same base
$X$. Given connections $\G$ on $Y$ and $\G'$ on $Y'$, the bundle
product $Y\op\times_X Y'$  is provided with the  product
connection
\mar{b1.96}\beq
\G\times\G' = dx^\la\ot\left(\dr_\la +\G^i_\la\frac{\dr}{\dr y^i}
+ \G'^j_\la\frac{\dr}{\dr y'^j}\right). \label{b1.96}
\eeq

$\bullet$ Given a fibre bundle $Y\to X$, let $f:X'\to X$ be a
manifold morphism and $f^*Y$ the pull-back of $Y$ over $X'$. Any
connection $\G$ (\ref{b1.223}) on $Y\to X$ yields the pull-back
connection
\mar{mos82}\beq
f^*\G=\left(dy^i-\G^i_\la(f^\m(x'^\nu),y^j)\frac{\dr f^\la}{\dr
x'^\m}dx'^\m\right)\ot\dr_i \label{mos82}
\eeq
on the pull-back bundle $f^*Y\to X'$.

Every connection $\G$ on a fibre bundle $Y\to X$ defines the first
order differential operator
\mar{2116}\ben
&& D^\G:J^1Y\op\to_Y T^*X\op\otimes_Y VY, \label{2116}\\
&& D^\G=\la_1- \G\circ \pi^1_0 =(y^i_\la -\G^i_\la)dx^\la\otimes\dr_i,
\nonumber
\een
on $Y$ called the  covariant differential. If $s:X\to Y$ is a
section, its  covariant differential
\be
\nabla^\G s= D^\G\circ J^1s = (\dr_\la s^i - \G_\la^i\circ s)
dx^\la\ot \dr_i
\ee
and its  covariant derivative $\nabla_\tau^\G s
=\tau\rfloor\nabla^\G s$ along a vector field $\tau$ on $X$ are
introduced. In particular, a (local) section $s$ of $Y\to X$ is
called an  integral section for a connection $\G$ (or parallel
with respect to $\G$) if $s$ obeys the equivalent conditions
\be
\nabla^\G s=0 \quad {\rm or} \quad J^1s=\G\circ s.
\ee

Let $\G$ be a connection on a fibre bundle $Y\to X$. Given vector
fields $\tau$, $\tau'$ on $X$ and their horizontal lifts $\G\tau$
and $\G\tau'$ (\ref{b1.85}) on $Y$, let us consider the vertical
vector field
\be
&& R(\tau,\tau')=\G [\tau,\tau'] - [\G \tau, \G \tau']= \tau^\la
\tau'^\m R_{\la\m}^i\dr_i, \\
&& R_{\la\m}^i = \dr_\la\G_\m^i - \dr_\m\G_\la^i +
\G_\la^j\dr_j \G_\m^i - \G_\m^j\dr_j \G_\la^i.
\ee
It can be seen as the contraction of vector fields $\tau$ and
$\tau'$ with the vertical-valued horizontal two-form
\mar{161'}\beq
R =\frac{1}{2} [\G,\G]_{\rm FN} = \frac12 R_{\la\m}^i dx^\la\wedge
dx^\m\otimes\dr_i \label{161'}
\eeq
on $Y$ called the  curvature form of a connection $\G$.

A  flat (or  curvature-free) connection is a connection $\G$ on a
fibre bundle $Y\to X$ which satisfies the following equivalent
conditions:

$\bullet$ its curvature vanishes everywhere on $Y$;

$\bullet$ its horizontal distribution is involutive;

$\bullet$ there exists a local integral section for the connection
$\G$ through any point $y\in Y$.

By virtue of Theorem \ref{to.1}, a flat connection $\G$  yields a
foliation of $Y$ which is transversal to the fibration $Y\to X$.
It called a  horizontal foliation. Its leaf through a point $y\in
Y$ is locally defined by an integral section $s_y$ for the
connection $\G$ through $y$. Conversely, let a fibre bundle $Y\to
X$ admit a horizontal foliation such that, for each point $y\in
Y$, the leaf of this foliation through $y$ is locally defined by a
section $s_y$ of $Y\to X$ through $y$. Then the map
\be
\G:Y\ni y\to j^1_{\pi(y)}s_y\in J^1Y
\ee
sets a flat connection on $Y\to X$. Hence, there is one-to-one
correspondence between the flat connections and the horizontal
foliations of a fibre bundle $Y\to X$.

Given a horizontal foliation of a fibre bundle $Y\to X$, there
exists the associated atlas of bundle coordinates $(x^\la, y^i)$
on $Y$ such that every leaf of this foliation is locally given by
the equations $y^i=$const., and the transition functions $y^i\to
{y'}^i(y^j)$ are independent of the base coordinates $x^\la$
\cite{book09}. It is called the  atlas of constant local
trivializations. Two such atlases are said to be equivalent if
their union also is an atlas of the same type. They are associated
to the same horizontal foliation. Thus, the following is proved.

\begin{theorem} \label{gena113} \mar{gena113}
There is one-to-one correspondence between the flat connections
$\G$ on a fibre bundle $Y\to X$ and the equivalence classes of
atlases of constant local trivializations of $Y$ such that
$\G=dx^\la\ot\dr_\la$ relative to the corresponding atlas.
\end{theorem}

\begin{example} \label{spr852} \mar{spr852}
Any trivial bundle has flat connections corresponding to its
trivializations. Fibre bundles over a one-dimensional base have
only flat connections.
\end{example}

Let $Y\to X$ be a vector bundle equipped with linear bundle
coordinates $(x^\la,y^i)$. It admits a  linear connection
\mar{167}\beq
\G =dx^\la\ot(\dr_\la + \G_\la{}^i{}_j(x) y^j\dr_i). \label{167}
\eeq
There are the following standard constructions of new linear
connections from old ones.

$\bullet$  Any linear connection $\G$ (\ref{167}) on a vector
bundle $Y\to X$ defines the dual linear connection
\mar{spr300}\beq
\G^*=dx^\la\ot(\dr_\la - \G_\la{}^j{}_i(x) y_j\dr^i)
\label{spr300}
\eeq
on the dual bundle $Y^*\to X$.

$\bullet$ Let $\G$ and $\G'$ be linear connections on vector
bundles $Y\to X$ and $Y'\to X$, respectively. The direct sum
connection $\G\oplus\G'$ on the Whitney sum $Y\oplus Y'$ of these
vector bundles is defined as the product connection (\ref{b1.96}).

$\bullet$ Similarly, the tensor product $Y\ot Y'$ of vector
bundles possesses the  tensor product connection
\mar{b1.92}\beq
\G\otimes\G'=dx^\la\ot\left[\dr_\la +(\G_\la{}^i{}_j
y^{ja}+\G'_\la{}^a{}_b y^{ib}) \frac{\dr}{\dr y^{ia}}\right].
\label{b1.92}
\eeq

The curvature of a linear connection $\G$ (\ref{167}) on a vector
bundle $Y\to X$ is usually written as a $Y$-valued two-form
\mar{mos4}\ben
&&R=\frac12 R_{\la\m}{}^i{}_j(x)y^j dx^\la\w dx^\m\ot e_i,\label{mos4}\\
&&R_{\la\m}{}^i{}_j = \dr_\la \G_\m{}^i{}_j - \dr_\m
\G_\la{}^i{}_j + \G_\la{}^h{}_j \G_\m{}^i{}_h - \G_\m{}^h{}_j
\G_\la{}^i{}_h, \nonumber
\een
due to the canonical vertical splitting $VY= Y\times Y$, where
$\{\dr_i\}=\{e_i\}$. For any two vector fields $\tau$ and $\tau'$
on $X$, this curvature yields the zero order differential operator
\mar{+98}\beq
R(\tau,\tau')
s=([\nabla_\tau^\G,\nabla_{\tau'}^\G]-\nabla_{[\tau,\tau']}^\G)s
\label{+98}
\eeq
on section $s$ of a vector bundle $Y\to X$.

Let us consider the composite bundle $Y\to\Si\to X$ (\ref{1.34}),
coordinated by $(x^\la, \si^m, y^i)$. Let us consider the jet
manifolds $J^1\Si$, $J^1_\Si Y$, and $J^1Y$ of the fibre bundles
$\Si\to X$, $Y\to \Si$ and $Y\to X$, respectively. They are
parameterized respectively by the coordinates
\be
( x^\la ,\si^m, \si^m_\la),\quad ( x^\la ,\si^m, y^i, \wt y^i_\la,
y^i_m),\quad ( x^\la ,\si^m, y^i, \si^m_\la ,y^i_\la).
\ee
There is the canonical map
\mar{1.38}\beq
\vr : J^1\Si\op\times_\Si J^1_\Si Y\ar_Y J^1Y, \qquad
y^i_\la\circ\vr=y^i_m{\si}^m_{\la} +\wt y^i_{\la}. \label{1.38}
\eeq
Using the canonical map (\ref{1.38}), we can consider the
relations between connections on fibre bundles $Y\to X$, $Y\to\Si$
and $\Si\to X$ \cite{book00,sau}.

Connections on fibre bundles $Y\to X$, $Y\to\Si$ and $\Si\to X$
read
\mar{spr290,1}\ben
&& \g=dx^\la\ot (\dr_\la +\g_\la^m\dr_m + \g_\la^i\dr_i), \label{spr290}\\
&&  A_\Si=dx^\la\ot (\dr_\la + A_\la^i\dr_i) +d\si^m\ot (\dr_m + A_m^i\dr_i),
\label{spr291}\\
&& \G=dx^\la\ot (\dr_\la + \G_\la^m\dr_m). \nonumber
\een
The canonical map $\vr$ (\ref{1.38}) enables us to obtain a
connection $\g$ on $Y\to X$ in accordance with the diagram
\be
\begin{array}{rcccl}
 & J^1\Si\op\xx_\Si J^1_\Si Y & \ar^\vr & J^1Y & \\
_{(\G,A)} & \put(0,-10){\vector(0,1){20}} & &
\put(0,-10){\vector(0,1){20}} & _{\g} \\
 & \Si\op\xx_X Y & \longleftarrow & Y &
\end{array}
\ee
This connection, called the composite connection, reads
\mar{b1.114}\beq
\g=dx^\la\ot [\dr_\la +\G_\la^m\dr_m + (A_\la^i +
A_m^i\G_\la^m)\dr_i]. \label{b1.114}
\eeq
It is a unique connection such that the horizontal lift $\g\tau$
on $Y$ of a vector field $\tau$ on $X$ by means of the connection
$\g$ (\ref{b1.114}) coincides with the composition $A_\Si(\G\tau)$
of horizontal lifts of $\tau$ onto $\Si$ by means of the
connection $\G$ and then onto $Y$ by means of the connection
$A_\Si$. For the sake of brevity, let us write $\g=A_\Si\circ\G$.

Given the composite bundle $Y$ (\ref{1.34}), there are the exact
sequences
\mar{63a}\ben
&& 0\to V_\Si Y\to VY\to Y\op\times_\Si V\Si\to 0, \label{63a}\\
&&  0\to Y\op\times_\Si V^*\Si\to V^*Y \to V_\Si^* Y\to 0,
\label{63b}
\een
where $V_\Si Y$ denotes the vertical tangent bundle of a fibre
bundle $Y\to\Si$ coordinated by $(x^\la,\si^m,y^i, \dot y^i)$. Let
us consider the splitting
\mar{63c}\ben
&& B: VY\ni v=\dot y^i\dr_i +\dot \si^m\dr_m \to v\rfloor B= \label{63c}\\
&& \qquad  (\dot y^i -\dot \si^m B^i_m)\dr_i\in V_\Si
Y, \nonumber \\
&& B=(\ol d y^i- B^i_m\ol d\si^m)\ot\dr_i\in
V^*Y\op\ot_Y V_\Si Y, \nonumber
\een
of the exact sequence (\ref{63a}). Then the connection $\g$
(\ref{spr290}) on $Y\to X$ and the splitting $B$ (\ref{63c})
define a connection
\mar{nnn}\ben
&& A_\Si=B\circ \g: TY\to VY \to V_\Si Y, \nonumber\\
&& A_\Si=dx^\la\ot(\dr_\la +(\g^i_\la - B^i_m\g^m_\la)\dr_i)+
\label{nnn}\\
&& \qquad d\si^m\ot (\dr_m +B^i_m\dr_i), \nonumber
\een
on the fibre bundle $Y\to\Si$.

Conversely, every connection $A_\Si$ (\ref{spr291}) on a fibre
bundle $Y\to\Si$ provides the splittings
\mar{46a,b}\ben
&& VY=V_\Si Y\op\oplus_Y A_\Si(Y\op\times_\Si V\Si),\label{46a}\\
&& \dot y^i\dr_i + \dot\si^m\dr_m= (\dot y^i -A^i_m\dot\si^m)\dr_i
+ \dot\si^m(\dr_m+A^i_m\dr_i), \nonumber\\
&& V^*Y = (Y\op\times_\Si V^*\Si) \op\oplus_Y A_\Si(V^*_\Si Y), \label{46b} \\
&& \dot y_i\ol dy^i +\dot\si_m \ol d\si^m= \dot y_i(\ol dy^i -A^i_m\ol
d\si^m) + (\dot \si_m + A^i_m\dot y_i)\ol d\si^m, \nonumber
\een
of the exact sequences (\ref{63a}) -- (\ref{63b}). Using the
splitting (\ref{46a}), one can construct the first order
differential operator
\mar{7.10}\beq
\wt D: J^1Y\to T^*X\op\otimes_Y V_\Si Y, \quad \wt D=
dx^\la\otimes(y^i_\la- A^i_\la -A^i_m\si^m_\la)\dr_i, \label{7.10}
\eeq
called the  vertical covariant differential, on the composite
fibre bundle $Y\to X$.

The vertical covariant differential (\ref{7.10}) possesses the
following important property. Let $h$ be a section of a fibre
bundle $\Si\to X$, and let $Y_h\to X$ be the restriction of a
fibre bundle $Y\to\Si$ to $h(X)\subset \Si$.  This is a subbundle
$i_h:Y_h\to Y$ of a fibre bundle $Y\to X$. Every connection
$A_\Si$ (\ref{spr291}) induces the pull-back connection
(\ref{mos82}):
\mar{mos83}\beq
A_h=i_h^*A_\Si=dx^\la\ot[\dr_\la+((A^i_m\circ h)\dr_\la h^m
+(A\circ h)^i_\la)\dr_i] \label{mos83}
\eeq
on $Y_h\to X$. Then the restriction of the vertical covariant
differential $\wt D$ (\ref{7.10}) to $J^1i_h(J^1Y_h)\subset J^1Y$
coincides with the familiar covariant differential $D^{A_h}$
(\ref{2116}) on $Y_h$ relative to the pull-back connection $A_h$
(\ref{mos83}).

\subsection{Differential operators and connections on modules}

This Section addresses the notion of a linear differential
operator on a module over an arbitrary commutative ring
\cite{kras,book00,book12,sard09}.

Let $\cK$ be a commutative ring and $\cA$ a commutative
$\cK$-ring. Let $P$ and $Q$ be $\cA$-modules. The $\cK$-module
$\hm_\cK (P,Q)$ of $\cK$-module homomorphisms $\Phi:P\to Q$ can be
endowed with the two different $\cA$-module structures
\mar{5.29}\beq
(a\Phi)(p)= a\Phi(p),  \qquad  (\Phi\bll a)(p) = \Phi (a p),\qquad
a\in \cA, \quad p\in P. \label{5.29}
\eeq
For the sake of convenience, we refer to the second one as an
$\cA^\bll$-module structure. Let us put
\be
\dl_a\Phi= a\Phi -\Phi\bll a, \qquad a\in\cA.
\ee

\begin{definition} \label{ws131} \mar{ws131}
An element $\Delta\in\hm_\cK(P,Q)$ is called a $Q$-valued
differential operator of order $s$ on $P$ if
\be
\dl_{a_0}\circ\cdots\circ\dl_{a_s}\Delta=0
\ee
for any tuple of $s+1$ elements $a_0,\ldots,a_s$ of $\cA$. The set
$\dif_s(P,Q)$ of these operators inherits the $\cA$- and
$\cA^\bll$-module structures (\ref{5.29}).
\end{definition}

For instance, zero order differential operators obey the condition
\be
\dl_a \Delta(p)=a\Delta(p)-\Delta(ap)=0, \qquad a\in\cA, \qquad
p\in P,
\ee
and, consequently, they coincide with $\cA$-module morphisms $P\to
Q$. A first order differential operator $\Delta$ satisfies the
condition
\be
\dl_b\circ\dl_a\,\Delta(p)= ba\Delta(p) -b\Delta(ap)
-a\Delta(bp)+\Delta(abp) =0, \quad a,b\in\cA.
\ee

\begin{definition} \label{1016} \mar{1016}
A  connection on an $\cA$-module $P$ is an $\cA$-module morphism
\mar{1017}\beq
\gd\cA\ni u\to \nabla_u\in \dif_1(P,P) \label{1017}
\eeq
such that the first order differential operators $\nabla_u$ obey
the   Leibniz rule
\be
\nabla_u (ap)= u(a)p+ a\nabla_u(p), \quad a\in \cA, \quad p\in P.
\ee
\end{definition}

Though $\nabla_u$ (\ref{1017}) is called a connection, it in fact
is a covariant differential on a module $P$.

For instance, let $Y\to X$ be a smooth vector bundle. Its global
sections form a $C^\infty(X)$-module $Y(X)$. The well-known Serre
-- Swan theorem \cite{book09} states the categorial equivalence
between the vector bundles over a smooth manifold $X$ and
projective modules of finite rank over the ring $C^\infty(X)$ of
smooth real functions on $X$. A corollary of this equivalence is
that the derivation module of the real ring $C^\infty(X)$
coincides with the $C^\infty(X)$-module $\cT(X)$ of vector fields
on $X$. If $P$ is a $C^\infty(X)$-module, one can reformulate
Definition \ref{1016} of a connection on $P$ as follows.

\begin{definition} \label{0210} \mar{0210}
A connection on a $C^\infty(X)$-module $P$ is a
$C^\infty(X)$-module morphism
\be
\nabla: P\to \cO^1(X)\ot P,
\ee
which satisfies the Leibniz rule
\be
\nabla(fp)=df\ot p +f\nabla(p), \qquad f\in C^\infty(X), \qquad
p\in P.
\ee
It associates to any vector field $\tau\in\cT(X)$ on $X$ a first
order differential operator $\nabla_\tau$ on $P$ which obeys the
Leibniz rule
\be
\nabla_\tau(fp)=(\tau\rfloor df)p +f\nabla_\tau p.
\ee
\end{definition}

In particular, let $Y\to X$ be a vector bundle and $Y(X)$ its
structure module. The notion of a connection on the structure
module $Y(X)$ is equivalent to the standard geometric notion of a
connection on a vector bundle $Y\to X$ \cite{book00}.

\end{document}